\documentclass[nofootinbib,twocolumn,pra,letterpaper,longbibliography]{revtex4-2}
\usepackage{latexsym,amsmath,amssymb,amsfonts,graphicx,color,amsthm}
\usepackage{enumerate}
\usepackage{enumitem}
\usepackage{braket}
\usepackage{adjustbox}
\usepackage{footnote}
\usepackage[dvipsnames]{xcolor}
\usepackage{tipa}
\usepackage{textcomp}
\usepackage{fontenc}
\usepackage{mathtools}
\usepackage{mathrsfs}
\usepackage{color}
\usepackage{colortbl}
\usepackage{graphics,graphicx}
\usepackage{txfonts}
\usepackage{lipsum, babel}
\usepackage{bbm}
\usepackage{amsmath}
\usepackage{amssymb}

\usepackage[unicode=true,pdfusetitle, bookmarks=true,bookmarksnumbered=false,bookmarksopen=false, breaklinks=false,pdfborder={0 0 0},backref=false,colorlinks=false] {hyperref}
\hypersetup{ colorlinks,linkcolor=myurlcolor,citecolor=myurlcolor,urlcolor=myurlcolor}

\definecolor{myurlcolor}{rgb}{0,0,0.7}

\usepackage{float}
\usepackage{hyperref}

\newcommand{\cB}{\mathcal{B}}
\newcommand{\cC}{\mathcal{C}}
\newcommand{\cD}{\mathcal{D}}

\newcommand{\cF}{\mathcal{F}}
\newcommand{\cG}{\mathcal{G}}
\newcommand{\cH}{\mathcal{H}}
\newcommand{\cI}{\mathcal{I}}
\newcommand{\cJ}{\mathcal{J}}
\newcommand{\cK}{\mathcal{K}}
\newcommand{\cL}{\mathcal{L}}
\newcommand{\cM}{\mathcal{M}}
\newcommand{\cN}{\mathcal{N}}
\newcommand{\cO}{\mathcal{O}}
\newcommand{\cP}{\mathcal{P}}

\newcommand{\cR}{\mathcal{R}}
\newcommand{\cS}{\mathcal{S}}
\newcommand{\cT}{\mathcal{T}}

\newcommand{\cV}{\mathcal{V}}
\newcommand{\cW}{\mathcal{W}}
\newcommand{\cX}{\mathcal{X}}
\newcommand{\cY}{\mathcal{Y}}

\newcommand{\Id}{\mathbbm{1}}
\newcommand{\tr}{Tr}

\newtheorem{theorem}{Theorem}
\newtheorem{proposition}{Proposition}
\newtheorem{lemma}{Lemma}
\newtheorem{corollary}{Corollary}
\newtheorem{definition}{Definition}
\newtheorem{example}{Example}

\newtheorem{remark}{Remark}

\begin{document}

\title{Instrument-based quantum resources: quantification, hierarchies and constructing resource theories}

\author{Jatin Ghai}

\email{jghai98@gmail.com}

\affiliation{Optics \& Quantum Information Group, The Institute of Mathematical Sciences, CIT Campus, Taramani, Chennai 600113, India.}
\affiliation{Homi Bhabha National Institute, Training School Complex, Anushakti Nagar, Mumbai 400085, India.}

\author{Arindam Mitra}
\email{am56@iitbbs.ac.in}
\email{arindammitra143@gmail.com}
 \affiliation{Department of Physics, School of Basic Sciences, Indian Institute of Technology Bhubaneswar, Odisha 752050, India.}
\affiliation{Korea Research Institute of Standards and Science, Daejeon 34113, South Korea.}

\date{\today}

\begin{abstract}
Quantum resources are certain features of the quantum world that provide advantages in certain information-theoretic, thermodynamic, or other useful operational tasks that are outside the realm of what classical theories can achieve. Quantum resource theories provide us with an elegant framework for studying these resources quantitatively and rigorously. While numerous state-based quantum resource theories have already been investigated, and to some extent, measurement-based resource theories have also been explored, instrument-based resource theories remain largely unexplored, with only a few notable exceptions. As quantum instruments are devices that provide both the classical outcomes of induced measurements and the post-measurement quantum states, they are quite important, especially for scenarios where multiple parties sequentially act on a quantum system. In this work, we study several instrument-based resource theories, namely (1) the resource theory of information preservability, (2) the resource theory of (strong) entanglement preservability, (3) the resource theory of (strong) incompatibility preservability, (4) the resource theory of traditional incompatibility, and (5) the resource theory of parallel incompatibility. Furthermore, we outline the hierarchies of these instrument-based resources and provide measures to quantify them. We then also established a relationship between our resource measure and the advantage in an information-theoretic task. In short, we provide \emph{a detailed framework} for a wide variety of instrument-based quantum resource theories.
\end{abstract}

\maketitle
\section{Introduction}

In quantum theory, there are certain features that do not have any classical analogues \textit{e.g.}~entanglement, coherence, incompatibility, etc. Such elements can be considered as quantum resources, as they provide advantages in certain information-theoretic \cite{Skrzypczyk_incomp_state_disc, bennet_teleportation, Uola_q_res_exclus,Carmeli_incomp_meas_qrac,Mori_incomp_chan_state_disc,debasish_incom_random_acces_code} and thermodynamical tasks \cite{kwon_coh_heat_engine,bresque_eng_ent,Wang_exp_eng_ent,shi_coh_heat_eng} beyond the scope of classical physics. Thus, it is important to quantify exactly how useful these resources are in various such tasks. A natural approach to accomplish this is the framework of quantum resource theories \cite{Chitambar_QRT_review,Gour_QRT_book}. A plethora of resource theories have been developed for a variety of quantum resources. Some of such examples would be resource theories of entanglement \cite{Horodecki_review_entang}, coherence \cite{Baumgratz_coh_RT, Winter_coh_RT, Bischof_coh_RT}, incompatibility of measurements \cite{Buscemi_meas_incomp}, measurement coherence \cite{baek_quantifying}, measurement sharpness \cite{Mitra_sharp_meas,Buscemi_sharp_meas}, incompatibility of channels \cite{Mori_incomp_chan_state_disc}, traditional incompatibility of  instruments \cite{chitambar_PID}, etc. Some of these resources are inherent properties of the quantum states (e.g., entanglement, coherence etc.), while others are properties of individual quantum measurements (e.g., measurement coherence, measurement sharpness etc.) or a set of measurements (e.g., measurement incompatibility) or, going a step further, even of the quantum instruments (e.g., traditional incompatibility of instruments).

As discussed above, numerous quantum state-based resource theories have already been widely explored in the literature, and measurement-based resource theories have also been studied to some extent. But except for a very few cases, quantum instrument-based resource theories have not been explored much to the best of our knowledge. Quantum instruments are the devices that provide both the classical outcomes of individual measurements and post-measurement states. These are the essential elements of quantum measurement theory\cite{Ozawa_1984_inst,Heinossari_2018_Qubit} and are useful devices in sequential or multiparty scenarios (such as nonlocality sharing\cite{Colbeck_2020_NonLocal}, steerability sharing\cite{Sasmal_2018_Steering}, advantage in quantum random access codes\cite{Mohan_2019} etc.) where for an example, first party may perform an instrument on their quantum state and the second party may perform an operation on the post-measurement state, depending on the classical outcome. Thus, there exist many properties of quantum instruments that can be considered as resources and, therefore, there is a potential to explore a large number of instrument-based resource theories that are useful for quantum information technologies. For example, the traditional incompatibility that is a property of a set of instruments has already been shown to be a resource for programmable quantum instruments in Ref. \cite{chitambar_PID}. However, there exist many other instrument-based quantum resources that have not yet been explored in detail. Here, our \emph{motivation} is to study several instrument-based quantum resources in a resource-theoretic framework.

In this work, we try to construct and study several instrument-based quantum resource theories, namely (1) the resource theory of information preservability, (2) the resource theory of (strong) entanglement preservability, (3) the resource theory of (strong) incompatibility preservability, (4) the resource theory of traditional incompatibility (already constructed in Ref. \cite{chitambar_PID} and therefore, here we provide more insight), and (5) the resource theory of parallel incompatibility. We also study the hierarchies of these instrument-based quantum resources and provide resource measures to quantify them in an elegant way. In short, we try to provide \emph{a complete framework} for several instrument-based quantum resource theories. Our work provides deep insight into all the above-said instrument-based resources and raises the scope of important future research directions (e.g., resource conversion, catalysis, etc., for all of the above-said resource theories). To the best of our knowledge, a detailed resource-theoretic characterisation of such a variety of resources for quantum instruments has not been done in the literature previously. For more details on the importance and acope of our work, we refer the readers to Sec. \ref{Sec:conc}.

The rest of the paper is organized as follows. In Sec. \ref{Sec:Prelim}, we discuss the preliminaries. More specifically, in Sec. \ref{Subsec:Prelim:QM_QC_QI}, we discuss the basic concepts of quantum measurements, quantum channels, and quantum instruments. In Sec. \ref{Subsec:Prelim:DM}, we discuss a distance measure for quantum measurements and quantum channels using the diamond norm. In Sec. \ref{Subsec:Prelim:QRT}, we discuss the resource-theoretic framework for a generic instrument-based resource. From Sec. \ref{Sec:Main}, we start presenting our main results. More specifically, in Sec. \ref{Subsec:Main:gen_QI_resource}, we study the quantification and a distance measure for a generic instrument-based quantum resource. In Sec. \ref{Sec:SDP_res_meas}, we try to provide a method to compute our distance measure and resource measures using SDP. In Sec. \ref{Subsec:Main:QI_QRTs}, we construct and explore several instrument-based resource theories and study their hierarchy. In Sec. \ref{Sec:inf_theo_task}, we establish a relationship between our resource measure and advantage in an information-theoretic task. In Sec. \ref{Sec:conc}, we summarize our work and discuss future research directions.

\section{Preliminaries}
\label{Sec:Prelim}
\subsection{Quantum Measurements, Quantum Channels and Quantum Instruments}
\label{Subsec:Prelim:QM_QC_QI}

A set $M=\{M(x)\in\cL(\cH)\}_{x\in\Omega_{M}}$ of positive semidefinite operators acting on a Hilbert space $\cH$ is said to constitute a quantum measurement if $\sum_{x\in\Omega_M}M(x)=\mathbbm{1}_{\cH}$ where $\mathbbm{1}_{\cH}$ is the identity matrix on Hilbert space $\cH$ \cite{Heinosaari_book_QF}. Here, $\Omega_M$ is the set of outcomes for the measurement $M$ and $\cL(\cH)$ is the set of all linear operators on $\cH$. Each $M(x)$ is termed as a POVM element of the measurement $M$ for given $x$. If we have $M^2(x)=M(x)~\forall x\in\Omega_M$, then $M$ is known to be a projective measurement. From now on, we will assume that all measurements have a finite number of outcomes and act on finite-dimensional Hilbert spaces. If the measurement M is performed on a system with arbitrary quantum state $\rho\in\cS(\cH)$, where $\cS(\cH)$ is the set of density matrices on $\cH$, then $\tr[\rho M(x)]$ gives the probability of obtaining the outcome $x$. We will refer to the set of all the measurements acting on the Hilbert space $\cH$ as $\mathscr{M}(\cH)$. Next, we denote the one-outcome trivial measurement acting on the Hilbert space $\cK$ as $\cT_\cK$ i.e., $\cT_\cK=\{\mathbbm{1}_\cK\}$. Operationally, this is equivalent to performing ``no measurement". If we have two quantum measurements $M\in\mathscr{M}(\cH)$ and $N\in\mathscr{M}(\cK)$ then $M\otimes N=\{M(x)\otimes N(y)\}_{x\in\Omega_M,y\in\Omega_N}\in\mathscr{M}(\cH\otimes\cK)$ with the outcome set being $\Omega_\cH\times\Omega_\cK$. A trivially enlarged version of $M$ can be defined as $\widehat{M}_{\cH\otimes\cK}=\{\widehat{M}_{\cH\otimes\cK}(x)=M(x)\otimes\mathbbm{1}_{\cK}\}\in\mathscr{M}(\cH\otimes\cK)$.

Two arbitrary measurements $M\in\mathscr{M}(\cH)$ and $N\in\mathscr{M}(\cH)$ with outcome sets $\Omega_M$ and $\Omega_N$ are said to be compatible if there exists a measurement $G=\{G(x,y)\in\mathscr{M}(\cH)\}$ with outcome set $\Omega_M\times\Omega_N$ such that \cite{Heinosaari_incomp_review}
\begin{align}
    M(x)&=\sum_{y\in\Omega_N}G(x,y)\qquad\forall x\in\Omega_M,\\
    N(y)&=\sum_{x\in\Omega_M}G(x,y)\qquad\forall y\in\Omega_N.
\end{align}
The measurement $G$ is known as the joint measurement of the pair of measurements $M$ and $N$. Thus, by performing the measurement $G$, we can simultaneously implement both the measurements $M$ and $N$. This definition of compatibility is generalised to an arbitrary number of measurements in a similar way as above. A set of measurements that does not have a joint measurement is called incompatible \cite{Heinosaari_incomp_review}.

A quantum channel transforms an arbitrary density matrix to another density matrix. Mathematically, it is represented by a completely positive and trace-preserving linear map (CPTP) $\Lambda:\cL(\cH)\rightarrow\cL(\cK)$ \cite{Heinosaari_book_QF}. Equivalently, the action of the dual map $\Lambda^{\dagger}:\cL(\cK)\rightarrow\cL(\cH)$ in the Heisenberg picture can be defined by the equation
\begin{equation}
    \tr[\Lambda(A)B]=\tr[A\Lambda^{\dagger}(B)],\label{Eq:Duality}
\end{equation}
where $A\in\cL(\cH)$ and $B\in\cL(\cK)$ \cite{Heinosaari_book_QF}. The set of all quantum channels with input Hilbert space $\cL(\cH)$ and output Hilbert space $\cL(\cK)$ is denoted by $\mathscr{C}(\cH,\cK)$. Composition of two quantum channels $\Lambda_1\in\mathscr{C}(\cH,\cH_1)$ and $\Lambda_2\in\mathscr{C}(\cH_1,\cK)$ is defined as $\Lambda(\rho):=\Lambda_1\circ\Lambda_2(\rho):=\Lambda_1(\Lambda_2(\rho))$. Evidently, $\Lambda\in\mathscr{C}(\cH,\cK)$. In the literature, people have studied a wide variety  of quantum channels due to their unique actions on the quantum states.

For example, we have a depolarising quantum channel $\Gamma_d^t:\cL(\cH)\rightarrow\cL(\cH)$ which probabilistically add white noise to any quantum state. The action of it on an arbitrary quantum state $\rho\in\cS(\cH)$ is mathematically represented as
\begin{align}
    \Gamma_d^t(\rho)=t\rho+(1-t)\frac{\mathbbm{1}_{\cH}}{d},\label{depolarising}
\end{align}
where $\frac{-1}{d^2-1}\leq t\leq1$ and $d$ is the dimension of the Hilbert space $\cH$. From the above definition, it is clear that the depolarizing channels are unital. For $d=2$ \textit{i.e.} for qubits the Krauss operators of this depolarising channel are given by $\sqrt{\frac{1+3t}{4}}\mathbbm{1}_{2\times 2},\sqrt{\frac{1-t}{4}}\sigma_x,\sqrt{\frac{1-t}{4}}\sigma_y$ and $\sqrt{\frac{1-t}{4}}\sigma_z$. We will use some properties of depolarsing channels to study the properties of some classes of quantum channels.

Another important class of quantum channels is the class of channels that break the entanglement of any bipartite state when it is acted on one side of that bipartite state. These are called entanglement-breaking channels (EBC). Formally, a quantum operation (i.e., a completely positive trace non-increasing map) $\Lambda:\cL(\cH_A)\rightarrow\cL(\cK)$ is entanglement-breaking if for all $\rho_{AB}\in\cS(\cH_A\otimes\cH_B)$, $\Lambda\otimes\mathbbm{I}_{\cH_{\cB}}(\rho_{AB})$ is a separable sub-normalised state for $\cH_\cB$ of an arbitrary dimension. Mathematically, an arbitrary entanglement-breaking quantum operation $\Lambda:\cL(\cH)\rightarrow\cL(\cK)$ can be written as
\begin{align}
    \Lambda(\rho)=\sum_a\rho_a\tr[\rho A(a)],\label{Eq:EB_Oper}
\end{align}
where $A(a)\geq 0~\forall a$, $\sum_a A(a)\leq \Id_{\cH}$ and $\rho_a\in\cS(\cK)~\forall a$ \cite{Horodecki_gen_EBC}. If $\{A(x)\}$ is a measurement, (i.e.,$\sum_a A(a)=\Id_{\cH}$) then $\Lambda$ is an entanglement-breaking channel. The entanglement-breaking channels form a convex set \cite{Horodecki_gen_EBC, Ruskai_qubit_EBC}. It is worth mentioning that the Choi matrix of a channel is separable iff it is entanglement-breaking \cite{Horodecki_gen_EBC, Ruskai_qubit_EBC}. A depolarising channel in Eq. (\ref{depolarising}) is EBC for $t\leq\frac{1}{1+d}$ \cite{Heinosaari_incomp_break_chan}. The set of EBC acting on a Hilbert space $\cH$ is denoted as $\mathscr{C}^{EBC}_{\cH}$. For an arbitrary quantum channel $\Theta:\cL(\cK)\rightarrow\cL(\cK^{\prime})$, it is known that $\tilde{\Lambda}:=\Theta\circ\Lambda\in\mathscr{C}^{EBC}_{\cH}$, if $\Lambda\in\mathscr{C}^{EBC}_{\cH}$. 


As we have already discussed, an arbitrary set of $n$ measurements $\{M_1, M_2,\ldots, M_n\}$ can be incompatible \textit{i.e.} there doesn't exist a joint measurement to produce the outcomes of all of them simultaneously with accurate marginal probability. There exists a class of channels called $n$-incompatibility breaking channels ($n$-IBC) whose action on this set of measurements renders them compatible. Mathematically, $\Lambda$ is $n$-IBC if the set $\{(\Lambda)^{\dagger}(M_1),(\Lambda)^{\dagger}(M_2),\ldots,(\Lambda)^{\dagger}(M_n)\}$ is compatible for an arbitrary set of measurements $\{M_1,\ldots,M_n\}$ \cite{Heinosaari_incomp_break_chan}. Just like EBCs, the set of all $n$-IBCs also forms a convex set. The channel $\Gamma^t_d$ in Eq. (\ref{depolarising}) is $n$-IBC whenever $t\leq \frac{n+d}{n(1+d)}$ \cite{Heinosaari_incomp_break_chan}. The set of $n$-IBC acting on a Hilbert space ${\cH}$ is represented as $\mathscr{C}^{n-IBC}_{\cH}$. For an arbitrary quantum channel $\Theta:\cL(\cK)\rightarrow\cL(\cK^{\prime})$, it is known that $\tilde{\Lambda}:=\Theta\circ\Lambda\in\mathscr{C}^{n-IBC}_{\cH}$, if $\Lambda\in\mathscr{C}^{n-IBC}_{\cH}$. \cite{Heinosaari_incomp_break_chan}. 

A channel $\Lambda$ which is $n$-IBC for all $n\geq 2$ is called an incompatibility-breaking channel (IBC) \cite{Heinosaari_incomp_break_chan}. $\Gamma^t_d$ in Eq. (\ref{depolarising}) is IBC whenever $t\leq \frac{(3d-1)(d-1)^{(d-1)}}{d^d(d+1)}$ \cite{Heinosaari_incomp_break_chan}. For qubits, we have $t\leq\frac{5}{12}$. The set of IBCs acting on a Hilbert space $\cH$ is represented as $\mathscr{C}^{IBC}_{\cH}$. Again, for an arbitrary quantum channel $\Theta:\cL(\cK)\rightarrow\cL(\cK^{\prime})$, it is known that $\tilde{\Lambda}:=\Theta\circ\Lambda\in\mathscr{C}^{IBC}_{\cH}$, if $\Lambda\in\mathscr{C}^{IBC}_{\cH}$. \cite{Heinosaari_incomp_break_chan}. It should be mentioned that the hierarchy $\mathscr{C}^{IBC}_{\cH}\subset\ldots\subset\mathscr{C}^{n-IBC}_{\cH}\subset\ldots\subset\mathscr{C}^{2-IBC}_{\cH}$ along-with $\mathscr{C}^{EBC}_{\cH}\subset\mathscr{C}^{IBC}_{\cH}$ have been proved in Ref. \cite{Heinosaari_incomp_break_chan}. 

Similar to the measurements, the notion of incompatibility also exists for quantum channels. Two channels $\Lambda_1:\cL(\cH)\rightarrow\cL(\cK_1)$ and $\Lambda_2:\cL(\cH)\rightarrow\cL(\cK_2)$ are said to be compatible if there exists a channel $\Lambda:\cL(\cH)\rightarrow\cL(\cK_1\otimes\cK_2)$ such that \cite{Heinosaari_incomp_review, Heinosaari_incomp_chan} 
\begin{align}
    \Lambda_1&=\tr_{\cK_2}\circ\Lambda\nonumber\\
    \Lambda_2&=\tr_{\cK_1}\circ\Lambda.
\end{align}
Here, $\Lambda$ is called the joint channel of the pair of channels $(\Lambda_1, \Lambda_2)$. This definition of compatibility can again be extended to an arbitrary set of channels in a similar way.

One channel can be transformed into another channel through a superchannel \cite{Chiribella_sup_chan,Gour_compar_sup_chan}. Suppose we have a quantum channel $\Lambda\in\mathscr{C}(\cH_1,\cH_2)$. Than a superchannel $\hat{\Xi}$ transforms it into a channel $\hat{\Xi}(\Lambda)\in\mathscr{C}(\cK_1,\cK_2)$. Mathematically, it can be represented as \cite{Chiribella_sup_chan,Gour_compar_sup_chan}:
\begin{equation}
    \hat{\Xi}(\Lambda)=\Theta_{post}\circ(\Lambda\otimes\mathbbm{I}_\mathfrak{R})\circ\Theta_{pre},\label{supermap}
\end{equation}
where quantum channel $\Theta_{pre}:\cL(\cK_1)\rightarrow\cL(\cH_1\otimes\mathfrak{R})$ is called the pre-processing channel, $\Theta_{post}:\cL(\cH_2\otimes\mathfrak{R})\rightarrow\cL(\cK_2)$ is called the post-processing channel, and $\mathfrak{R}$ is an ancillary Hilbert space.

A quantum instrument is the simultaneous generalization of quantum measurements and quantum channels. Mathematically, it is defined as a set of CP maps, $\mathbf{I}=\{\Phi_a:\cL(\cH)\rightarrow\cL(\cK)\}_{a\in\Omega_\mathbf{I}}$ such that $\Phi=\sum_{a\in\Omega_\mathbf{I}}\Phi_a$ is a quantum channel \cite{Heinosaari_book_QF}. Given a quantum state $\rho$, the number $a$ is the classical output of the quantum instrument, while $\Phi_a(\rho)$ is its quantum output, both occurring with probability $\tr[\Phi_a(\rho)]$. We denote the set of all such quantum instruments with input Hilbert space $\cH$ and output Hilbert space $\cK$ as $\mathscr{I}(\cH,\cK)$. A quantum instrument induces a unique POVM $\{A(a)\}$ such that $\tr[\Phi_a(\rho)]=\tr[\rho A(a)]$ for all $\rho\in\cS(\cH)$ and $a\in\Omega_{\mathbf{I}}$. In fact it is straightforward to see that $A(a)=\Phi_a^{\dagger}(\mathbbm{1}_{\cK}) ,~\forall~a$ where $\Phi_a^{\dagger}$ is the dual map for the CP map $\Phi_a$ defined in the similar way as in Eq. \eqref{Eq:Duality}. Although this induced POVM is unique, but for a given POVM, there exist many different instruments that implement it.

Just as in the case of quantum channels, one quantum instrument can be transformed into another quantum instrument using the notion of post-processing. For two quantum instruments, $\mathbf{I}=\{\Phi_a\}_{a\in\Omega_{\mathbf{I}}}\in\mathscr{I}(\cH,\cK)$ and $\tilde{\mathbf{I}}=\{\tilde{\Phi}_b\}_{b\in\Omega_{\tilde{\mathbf{I}}}}\in\mathscr{I}(\cH,\tilde{\cK})$, instrument $\tilde{\mathbf{I}}$ is said to be the post processing of the instrument $\mathbf{I}$ if there exists a set of quantum instruments $\{\mathbf{P}^a=\{P^a_b\}_{b\in\Omega_{\tilde{\mathbf{I}}}}\in\mathscr{I}(\cK,\tilde\cK)\}_{a\in\Omega_{\mathbf{I}}}$ such that \cite{Leevi_postproc_instrument}
\begin{align}
    \tilde\Phi_b=\sum_{a\in\Omega_{\mathbf{I}}}P^a_b\circ\Phi_a\qquad\forall b\in\Omega_{\tilde{\mathbf{I}}}.
\end{align}
Consider a measurement $B=\{B(b)\}_{b\in\Omega_B}\in\mathscr{M}(\cK)$. Then the quantum instrument $\mathbf{I}=\{\Phi_a\}\in\mathscr{I}(\cH,\cK)$ transforms it into another measurement $\mathbf{I}^{\dagger}[B]$ as
\begin{align}
    \mathbf{I}^{\dagger}[B]=\{\Phi^{\dagger}_a(B(b))\}_{(a,b)\in(\Omega_{\mathbf{I}}\times\Omega_B)},
\end{align}
with its output set being $\Omega_{\mathbf{I^{\dagger}[B]}}=\Omega_{\mathbf{I}}\times\Omega_b$.
As mentioned above, quantum measurements and quantum channels are special cases of quantum instruments; the concept of compatibility for quantum instruments can be similarly defined. In fact, in the literature, there are two different notions of instrument compatibility. The first notion is called traditional compatibility, and it is defined as:
\begin{definition}[Traditional Compatibility]
    A pair of quantum instruments $\mathbf{I}_1=\{\Phi^1_a:\cL(\cH)\rightarrow\cL(\cK)\}$ and $\mathbf{I}_2=\{\Phi^2_a:\cL(\cH)\rightarrow\cL(\cK)\}$ are called traditionally compatible\cite{heinosaari_strong_incomp_dev,Lever_Ph.D._Thesis,Uola_quant_qd_inp_outp_games} if there exist a joint instrument $\mathbf{I}=\{\Phi_{(a,b)}:\cL(\cH)\rightarrow\cL(\cK)\}$ such that 
\begin{align}
    \Phi^1_a=&\sum_b\Phi_{(a,b)}\qquad\forall a,\\
\Phi^2_b=&\sum_a\Phi_{(a,b)}\qquad\forall b.
\end{align}
Otherwise, the pair is traditionally incompatible. \label{Def:trad_comp}
\end{definition}
 Here, $\mathbf{I}$ is called the traditional joint instrument of the pair of quantum instruments $(\mathbf{I}_1, \mathbf{I}_2)$. It simultaneously produces both the classical outputs of two instruments, and the quantum output of either one of the instruments can be recovered by classical post-processing. This can be extended to an arbitrary set of quantum instruments in a similar way. These are denoted by $\mathscr{I}_{TC}(\cH\otimes\cK)$.

There is another related concept of weak compatibility, which states that two instruments $\mathbf{I}_1=\{\Phi^1_a\}$ and $\mathbf{I}_2=\{\Phi^2_b\}$ are weakly compatible if there exists a quantum channel such that $\sum_a\Phi^1_a=\sum_b\Phi^2_b=\Phi$. Traditional compatibility implies weak compatibility, but the converse is not true in general. The another notion is called parallel compatibility and it is defined as:

\begin{definition}[Parallel Compatibility]
    A pair of quantum instruments $\mathbf{I}_1=\{\Phi^1_a:\cL(\cH)\rightarrow\cL(\cK_1)\}$ and $\mathbf{I}_2=\{\Phi^2_a:\cL(\cH)\rightarrow\cL(\cK_2)\}$ are called parallel compatible\cite{Mitra_comp_inst,Mitra_char_quantifying_incomp_inst,Leevi_incomp_inst} if there exist a joint instrument $\mathbf{I}=\{\Phi_{(a,b)}:\cL(\cH)\rightarrow\cL(\cK_1\otimes\cK_2)\}$ such that
\begin{align}
    \Phi^1_a=&\sum_b\tr_{\cK_2}\circ\Phi_{(a,b)}\qquad\forall a,\label{Eq:Chan_para_comp_1}\\
\Phi^2_b=&\sum_a\tr_{\cK_1}\circ\Phi_{(a,b)}\qquad\forall b\label{Eq:Chan_para_comp_2}.
\end{align} 
Otherwise, the pair is parallel incompatible.\label{Def:para_comp}
\end{definition}

Here in R.H.S. of Eq. \eqref{Eq:Chan_para_comp_1}, we have used shorthand notation $\tr_{\cK_2}\circ\Phi_{(a,b)}$ for $(\mathbbm{I}_{\cK_1}\otimes\tr_{\cK_2})\circ\Phi_{(a,b)}$. The notation of the R.H.S. of Eq.\eqref{Eq:Chan_para_comp_2} is similar, and this kind of shorthand notation will be used when needed. Again $\mathbf{I}$ is called the joint instrument of the pair of quantum instruments $(\mathbf{I}_1, \mathbf{I}_2)$. Here, both classical and quantum outputs of two instruments are reproduced on a tensor product Hilbert space. This can be extended again to an arbitrary set of quantum instruments in a similar way. These are denoted by $\mathscr{I}_{PC}(\cH\otimes\cK)$

\begin{remark}
    \rm{It is worth mentioning that the two notions of incompatibility defined above are conceptually distinct. Not all parallel compatible instruments are traditionally compatible and vice versa i.e., one does not imply the other \cite{Mitra_comp_inst}. Furthermore, parallel compatibility can capture the notions of both measurement compatibility and channel compatibility while traditional compatibility can capture the notion of only measurement compatibility and cannot capture the notion of channel compatibility\cite{Mitra_comp_inst}.  It should also be mentioned that the parallel compatible set of instruments remains parallel compatible under post-processing \cite{Mitra_char_quantifying_incomp_inst}.  
    }\label{Rem:one}
\end{remark}

\subsection{A distance measure for measurements and channels via Diamond norm}
\label{Subsec:Prelim:DM}

For two quantum channels $\Lambda_1\in\mathscr{C}(\cH_A,\cK_A)$ and $\Lambda_2\in\mathscr{C}(\cH_B,\cK_B)$, the diamond distance between them is defined as
\begin{align}
    \cD_{\Diamond}(\Lambda_1, \Lambda_2):=&\mid\mid \Lambda_1-\Lambda_2\mid\mid_{\Diamond}\\
    =&\max_{\rho_{AB}\in\cS(\cH_A\otimes\cH_B)}\mid\mid \Lambda_1\otimes\mathbbm{I}_{\cH_B}(\rho)-\Lambda_2\otimes\mathbbm{I}_{\cH_B}(\rho)\mid\mid_1,\label{Eq:Dia_formula_H_p}
\end{align}
where $\mid\mid.\mid\mid_1$ is the trace norm, and $\dim(\cH_A)=\dim(\cH_B)$ \cite{Watrous_book_TQI}.

It is well-known that $\cD_{\Diamond}$ is monotonically non-increasing
under arbitrary pre-processing and post-processing channels,
or more generally, under an arbitrary super-channel. In
other words, for an arbitrary super channel $\hat{\Xi}$ that transforms
arbitrary $\Lambda_i\in\mathscr{C}(\cH_1,\cH_2)$ to  $\hat{\Xi}(\Lambda_i)\in\mathscr{C}(\cK_1,\cK_2)$ for $i\in\{1,2\}$,
we have \cite{Watrous_book_TQI}
\begin{align}
    \cD_{\Diamond}(\hat{\Xi}(\Lambda_1), \hat{\Xi}(\Lambda_2))\leq\cD_{\Diamond}(\Lambda_1, \Lambda_2).
\end{align}

Moreover, the diamond distance satisfies the joint convexity property. Mathematically, for quantum channels $\Lambda_1,\Lambda_2,\Psi_1,\Psi_2\in\mathscr{C}(\cH,\cK)$, we have
\begin{align}
    &\cD_{\Diamond}(p\Lambda_1+(1-p)\Psi_1,p\Lambda_2+(1-p)\Psi_2)\nonumber\\
    &\qquad\qquad\qquad\qquad\leq p\cD_{\Diamond}(\Lambda_1,\Lambda_2)+(1-p)\cD_{\Diamond}
(\Psi_1,\Psi_2),\label{Eq:diamond_dist_convexity}\end{align}
for all $0\leq p\leq1$ \cite{Watrous_book_TQI}.

Instead of two individual quantum channels, if we consider two ordered sets of channels $\cC_1=\{\Lambda_i\in\mathscr{C}(\cH,\cK)\}^n_{i=1}$ and $\cC_2=\{\Psi_i\in\mathscr{C}(\cH,\cK)\}^n_{i=1}$, then one way to define the distance between them is the following \cite{Mitra_distance_measure_MBQR}:

\begin{align}
    \overline{\cD}(\cC_1,\cC_2)=\max_{i\in\{1,\ldots,n\}}\cD_{\Diamond}(\Lambda_i,\Psi_i).\label{Eq:def_dist_meas_of_set_chan}
\end{align}

Clearly, $\overline{\cD}(\cC_1,\cC_2)=0$ \textit{iff} $\cC_1=\cC_2$. Hence, the distance measure is faithful. From now on, we will omit the word ordered. The distance defined in Eq. \eqref{Eq:def_dist_meas_of_set_chan} also satisfies the joint convexity property as\cite{Mitra_distance_measure_MBQR}
\begin{align}
    \overline{\cD}(p\cC_1+(1-p)\cC^{\prime}_1,p\cC_2+&(1-p)\cC^{\prime}_2)\nonumber\\
    &\leq p\overline{\cD}(\cC_1,\cC_2)+(1-p)\overline{\cD}(\cC^{\prime}_1,\cC^{\prime}_2),\label{Eq:Dist_set_chan_convex}
\end{align}
where $\cC^{\prime}_1=\{\Lambda^{\prime}_i\in\mathscr{C}(\cH,\cK)\}^n_{i=1}$ and $\cC^{\prime}_2=\{\Psi^{\prime}_i\in\mathscr{C}(\cH,\cK)\}^n_{i=1}$

Suppose that instead of transforming a single quantum channel into another quantum channel, we want to transform a finite set of quantum channels $\cC=\{\Lambda_i\in\mathscr{C}(\cH_1,\cH_2)\}$ to another finite set of quantum channels. Then a fairly general transformation can be written as \cite{Mitra_distance_measure_MBQR}
\begin{align}
    [\cV(\cC)]_j=\Theta^j_{post}\circ(\Sigma_C\otimes\mathbbm{I}_{\mathfrak{R}})\circ\Theta^j_{pre},\label{gensupermap}
\end{align}
where $[\cV(\cC)]_j$ is the $j$th element of the transformed set $\cV(\cC)$, the quantum channels $\Theta^j_{pre}:\cL(\cK_1)\rightarrow\cL(\cH_1\otimes\cH_I\otimes\mathfrak{R})$ and $\Theta^j_{post}:\cL(\cH_2\otimes\cH_I\otimes\mathfrak{R})\rightarrow\cL(\cK_2)$ for all $j$ with $\dim(\cH_I)=\mid\cC\mid$  where the cardinality of a set is denoted by the symbol $\mid.\mid$. Here, $\Sigma_C=\sum_i\Lambda_i\otimes\overline{\Phi}_i$, with $\overline{
\Phi}_i(.)=\bra{i}(.)\ket{i}\ket{i}\bra{i}$ for $\{\ket{i}\}$ being an orthonormal basis in the Hilbert space $\cH_\cI$, is called the \emph{controlled implementation} of channels in the set $\cC$. It is illustrated in Fig, \eqref{fig_cont_imp}. It is easy to see that if the sets $\cC$ and $\mathbf{\cV}(\cC)$ contain only one channel each, then the transformation in Eq. \eqref{gensupermap}
reduces to the one in Eq. \eqref{supermap}.

\begin{figure}[!h]
    \centering
    \includegraphics[height=140px, width =256px]{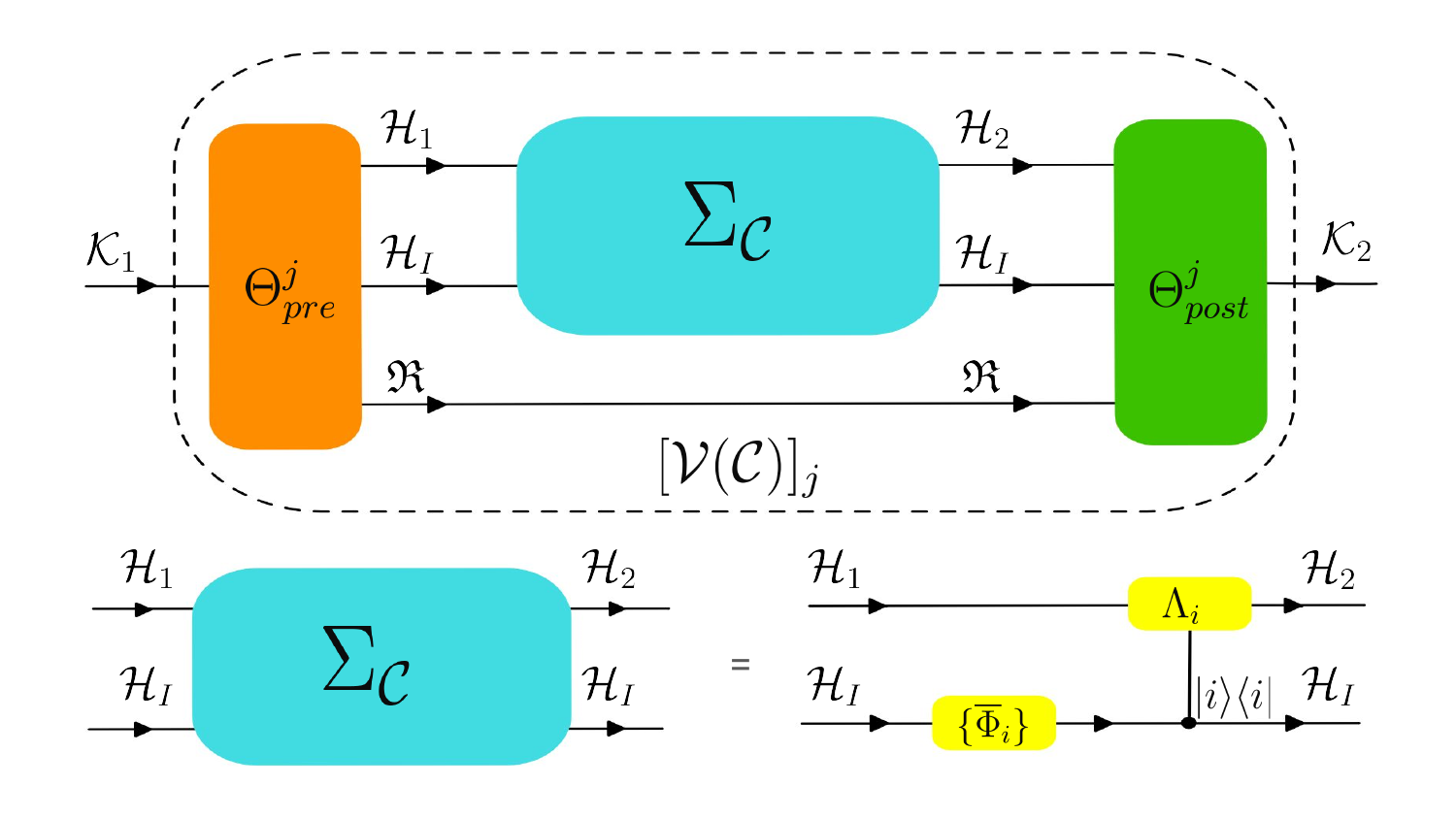}
     \caption{This represents a schematic diagram for the fairly general transformation of a set of quantum channels described in Eq. \eqref{gensupermap}. In this figure, we have shown the input and output Hilbert spaces of this transformation, as it is more relevant for our purpose of illustration. This general transformation will be used comprehensively throughout the paper.}\label{fig_cont_imp} 
\end{figure}

Similar to channels under the general transformation defined in Eq. (\ref{gensupermap}) the diamond distance is contractive \cite{Mitra_distance_measure_MBQR} i.e.,
\begin{align}
    \overline{\cD}(\cV[\cC_1],\cV[\cC_2])\leq\ \overline{\cD}(\cC_1, \cC_2).\label{Eq:supermap_contractive_distance}
\end{align}

An arbitrary measurement $M=\{M(x)\}\in\mathscr{M}(\cH_A)$ can be associated with a measure-prepare channel 
\begin{align}
    \Gamma_{M}(\rho)=\sum_{a}\tr[\rho M(a)]\ket{a}\bra{a}
\end{align}
for all $\rho\in\cL(\cH_A)$ where $\{\ket{a}\}$ is a chosen orthonormal basis of Hilbert space $\cH_{\Omega_M}$ with dimension $|\Omega_M |$ \cite{Tendick_dist_res_meas}. Then, the distance between two measurements can be defined as \cite{Tendick_dist_res_meas}

\begin{align}
    \cD_{\Diamond}(M_1, M_2)&:=\cD_{\Diamond}(\Gamma_{M_1}, \Gamma_{M_2})\nonumber\\
    &=\mid\mid \Gamma_{M_1}-\Gamma_{M_2}\mid\mid_{\Diamond}.
\end{align}

Similar to  the case of channels, if instead of two individual measurements, we have two sets of measurements $\cM=\{M_i\}$ and $\cN=\{N_i\}$, the distance between them can be defined as

\begin{align}
    \widetilde{\cD}(\cM,\cN):=&\overline{\cD}(\cG_{\cM},\cG_{\cN})\nonumber\\
    =&\max_{i\in\{1,\ldots,n\}}\cD_{\Diamond}(\Gamma_{\cM},\Gamma_{\cN})\label{Eq:def_dist_meas_of_set_meas},
\end{align}
where $\cG_{\cM}:=\{\Gamma_{M_i}\}$ and $\cG_{\cN}:=\{\Gamma_{N_i}\}$ \cite{Mitra_distance_measure_MBQR}.

\subsection{Formal aspects of a generic instrument-based quantum resource theory}
\label{Subsec:Prelim:QRT}

A generic quantum resource theory has two major constituents-\textit{(i) free objects} and \textit{(ii) free transformations}. Free objects are those elements of a given resource theory which does not contain that particular resource. These can be quantum states, quantum measurements, quantum channels, and quantum instruments, etc., depending on that particular resource theory. On the other hand, free transformations of a given resource theory are the transformations that transform a free object of that resource theory to another free object of the same. Once again, free transformations can be quantum channels, quantum supermaps, etc., depending on
the particular resource theory. If for a given resource theory, both the free objects and the free transformations form \textit{convex} sets, then that resource theory is classified as \textit{convex}.

Let us denote the set of free objects of a given resource theory as $\cF$ and the set of free transformations as $\cT$. As in this paper, we are only concerned about instrument-based resource theories; the free objects are some set of quantum instruments depending on the resource, and the free transformations are the transformations that transform one free set of quantum instruments to another free set of quantum instruments. Then, the following reasonable assumptions can be made 
\begin{enumerate}[label=A \arabic*.]
    \item Two objects $\cR_1$ and $\cR_2$ of an instrument-based  resource theory are free iff $\cR_1\otimes \cR_2$ is free.\label{A1}
    \item Two transformations $\cV_1$ and $\cV_2$ of an instrument-based resource theory are free iff $\cV_1\otimes V_2$ is free.\label{A2}
    \item Identity transformation $\mathfrak{I}$ (which maps an object to itself) is a free transformation.\label{A3}
    \item As for trace operation, the output Hilbert space is trivial for any input Hilbert space; it is a free object of an instrument-based resource theory.\label{A4}
\end{enumerate}

Another major ingredient of a resource theory is the quantification of the resource. It is accomplished by defining a resource measure $\mathbbm{R}$ satisfying the following properties:
\begin{enumerate}[label=R \arabic*.]
    \item (Non-negativity and faithfulness): $\mathbbm{R}(\cR)\geq0$ for all the objects of the given resource theory. and $\mathbbm{R}(\cR)=0$ iff $\cR\in\cF$.\label{R1}
    \item (Monotonicity): $\mathbbm{R}(\cR)\geq\mathbbm{R}(\cV(\cR))$ for all the objects $\cR$ and a free transformations $\cV\in\cT$ of the given resource theory.\label{R2}

    In addition to these necessary properties, another desirable property for a resource monotone is
    
    \item (Convexity): $\mathbbm{R}(\sum_i p_i\cR_i)\leq\sum_i p_i \mathbbm{R}(\cR_i)$ for any arbitrary set of objects $\{\cR_i\}_{i=1}^n$ of the given resource theory and for any probability distribution $\{p_i\}_{i=1}^n$.
\end{enumerate}

A resource measure satisfying all of the above properties is considered to be a good resource measure for a convex resource theory. For the remainder of this paper, we will concern ourselves with the instrument-based resource measure. A reasonable assumption is that any permutation of a free set of instruments is another free set of instruments.

Let us have an arbitrary set of quantum instruments $\cI=\{\mathbf{I}^i\}\subset\mathscr{I}(\cH_A,\cK_A)$ $\forall~i$ where we have $\mathbf{I}^i=\{\Lambda^i_a\}$ such that $\sum_a\Lambda^i_a=\Lambda^i,~\forall~i$. Then, for any given convex resource theory, resource robustness for an arbitrary set of quantum instruments $\cI$ is given as
\begin{align}
    \mathscr{R}(\cI)&=\min r\nonumber\\
    &~~~~~\text{s.t.}~~~\Big\{\Big\{\Phi^i_a=\frac{\Lambda^i_a+r\tilde{\Lambda}^i_a}{1+r}\Big\}_a\Big\}_i\in \cF_{(\cH_A,\cK_A)}\nonumber\\
    &\qquad\quad\tilde\cI=\{\tilde{\mathbf{I}}^i=\{\tilde\Lambda^i_a\}\}\subset\mathscr{I}(\cH_A,\cK_A),\label{Eq:robustness}
\end{align}
    where $\cF_{(\cH_A,\cK_A)}$ is the set of free sets of quantum instruments of the given resource theory.

 Similarly, the resource weight of an arbitrary set of quantum instruments $\cI$ is defined as
    \begin{align}
    \mathscr{W}(\cI)&=\min r\nonumber\\
    &~~~~~\text{s.t.}~~~ \cI=\Big\{\Big\{\Lambda^i_a=\frac{\Phi^i_a+r\tilde{\Lambda}^i_a}{1+r}\Big\}\Big\}\nonumber\\
    &\qquad\quad\tilde\cI=\{\tilde{\mathbf{I}}^i=\{\tilde\Lambda^i_a\}_a\}_i\subset\mathscr{I}(\cH_A,\cK_A)\nonumber\\
    &\qquad\quad\tilde\cJ=\{\tilde{\mathbf{J}}^i=\{\tilde\Phi^i_a\}\}\in \cF_{(\cH_A,\cK_A)}.\label{Eq:weight}
\end{align}

\section{Main Results}
\label{Sec:Main}
\subsection{Quantification of a generic instrument-based quantum resource}
\label{Subsec:Main:gen_QI_resource}

A quantum instrument $\mathbf{I}=\{\Phi_a\}_{a\in\Omega_{\mathbf{I}}}\in\mathscr{I}(\cH,\cK)$ can also be associated with a quantum channel $\hat{\Gamma}_{\mathbf{I}}\in\mathscr{C}(\cH,\cK\otimes\cH_{\Omega_\mathbf{I}})$ such that for all $\rho\in\cL(\cH)$
\begin{align}
     \hat{\Gamma}_{\mathbf{I}}(\rho)=\sum_{a}\Phi_a(\rho)\otimes\ket{a}\bra{a}
\end{align}
 where $\{\ket{a}\}$ is a chosen orthonormal basis of Hilbert space $\cH_{\Omega_{\mathbf{I}}}$ with dimension $|\Omega_{\mathbf{I}} |$. Clearly, for instruments $\mathbf{I}_1=\{\Phi^{1}_a\}$ and $\mathbf{I}_2=\{\Phi^{2}_a\}$ if $\Phi^{1}_a=\Phi^{2}_a ~\forall~a$ then $\hat{\Gamma}_{\mathbf{I}_1}=\hat{\Gamma}_{\mathbf{I}_2}$. For instruments $\mathbf{I}_1=\{\Phi^{1}_a\}$ and $\mathbf{I}_2=\{\Phi^{2}_a\}$ the distance between them can be defined as
\begin{align}
    \cD_{\Diamond}(\mathbf{I}_1, \mathbf{I}_2)&:=\cD_{\Diamond}(\hat{\Gamma}_{\mathbf{I}_1}, \hat{\Gamma}_{\mathbf{I}_2}).\label{Eq:Def_Dist_In}
\end{align} 
Even when $\Omega_{\mathbf{I}_1}\neq\Omega_{\mathbf{I}_2}$, we can still define the distance using Eq. \eqref{Eq:Def_Dist_In} by appending zero CP maps to the instrument with the smaller outcome set. We can also easily see that
\begin{align}
    (\mathbbm{I}_{\cK}\otimes\tr_{\cH_{\Omega_{\mathbf{I}}}})\circ\hat{\Gamma}_{\mathbf{I}}(\rho)=&\sum_{a}\Phi_a(\rho)\otimes\tr[\ket{a}\bra{a}]\nonumber\\
    =&\sum_{a}\Phi_a(\rho)=\Phi(\rho),\label{Eq:trace_first_instrument}
\end{align}
\begin{align}     (\tr_{\cK}\otimes\mathbbm{I}_{\cH_{\Omega_{\mathbf{I}}}})\circ\hat{\Gamma}_{\mathbf{I}}(\rho)=&\sum_{a}\tr[\Phi_a(\rho)]\otimes\ket{a}\bra{a}\nonumber\\
=&\sum_{a}\tr[\rho\Phi^*_a(\mathbbm{1}_{\cK})]\otimes\ket{a}\bra{a}\nonumber\\
=&\sum_{a}\tr[\rho A(a)]\otimes\ket{a}\bra{a}=\Gamma_{A}(\rho),\label{Eq:trace_second_instrument}
\end{align}
where $A=\{A(a)\}$ is the POVM induced by the instrument $\mathbf{I}$.

Similarly, the distance between two sets of instruments $\cI=\{\mathbf{I}_i\}$ and $\cJ=\{\mathbf{J}_i\}$ can be defined as

\begin{align}
    \widehat{\cD}(\cI,\cJ):=&\overline{\cD}(\hat{\cG}_{\cI},\hat{\cG}_{\cJ})\nonumber\\
    =&\max_{i\in\{1,\ldots,n\}}\cD_{\Diamond}(\hat{\Gamma}_{\mathbf{I}_i},\hat{\Gamma}_{\mathbf{J}_i})\label{Eq:def_dist_meas_of_set_inst},
\end{align}
where $\hat{\cG}_{\cI}:=\{\hat{\Gamma}_{\mathbf{I}_i}\}$ and $\hat{\cG}_{\cJ}:=\{\hat{\Gamma}_{\mathbf{J}_i}\}$.

From Eq. \eqref{Eq:Dist_set_chan_convex} it can be immediately conclude that
\begin{align}
    \overline{\cD}(p\cI+(1-p)\cI^{\prime},p\cJ+&(1-p)\cJ^{\prime})\nonumber\\
    &\leq p\overline{\cD}(\cI,\cJ)+(1-p)\overline{\cD}(\cI^{\prime},\cJ^{\prime}),\label{Eq:Dist_inst_convex}
\end{align}
where $\cI^{\prime}=\{\mathbf{I}^{\prime}_i\in\mathscr{I}(\cH,\cK)\}^n_{i=1}$ and $\cJ^{\prime}=\{\mathbf{J}^{\prime}_i\in\mathscr{I}(\cH,\cK)\}^n_{i=1}$.


 We start with proving the following simple but useful lemma.

\begin{lemma}
     If the instrument $\mathbf{I}_i=\{\Phi^i_a\}\in \mathscr{I}(\cH,\cK)$ implements $M_i$ and $\Phi_i$ for all $i\in\{1,2\}$ then
 \begin{align}
     \cD_{\Diamond}(M_1, M_2)\leq\cD_{\Diamond}(\mathbf{I}_1, \mathbf{I}_2)\nonumber\\
     \cD_{\Diamond}(\Lambda_1, \Lambda_2)\leq\cD_{\Diamond}(\mathbf{I}_1, \mathbf{I}_2).
 \end{align}\label{Le:ins_chan_meas_dist_ineq}
\end{lemma}
\begin{proof}
    \begin{align}
    \cD_{\Diamond}(M_1, M_2):=&\cD_{\Diamond}(\Gamma_{M_1}, \Gamma_{M_2})\nonumber\\   =&\cD_{\Diamond}\Big((\tr_{\cK}\otimes\mathbbm{I}_{\cH_{\Omega_{\mathbf{I}_1}}})\circ\hat{\Gamma}_{\mathbf{I}_1},(\tr_{\cK}\otimes\mathbbm{I}_{\cH_{\Omega_{\mathbf{I}_2}}})\circ\hat{\Gamma}_{\mathbf{I}_2}\Big)\nonumber\\ 
    \leq&\cD_{\Diamond}(\hat{\Gamma}_{\mathbf{I}_1}, \hat{\Gamma}_{\mathbf{I}_2}):=\cD_{\Diamond}(\mathbf{I}_1, \mathbf{I}_2).
    \end{align}
    where $\cH_{\Omega_{\mathbf{I}_1}}=\cH_{\Omega_{\mathbf{I}_2}}$. In the third line, we have used Eq. \eqref{Eq:trace_first_instrument} and the property that diamond distance is contractive under the composition of channels. Similarly, by using Eq. \eqref{Eq:trace_second_instrument} we get
    \begin{align}
    \cD_{\Diamond}(\Phi_1, \Phi_2)=&\cD_{\Diamond}\Big(\sum_{a}\Phi^1_a, \sum_{a}\Phi^2_a\Big)\nonumber\\    =&\cD_{\Diamond}\Big((\mathbbm{I}_{\cK}\otimes\tr_{\cH_{\Omega_{\mathbf{I}_1}}})\circ\hat{\Gamma}_{\mathbf{I}_1}), (\mathbbm{I}_{\cK}\otimes\tr_{\cH_{\Omega_{\mathbf{I}_2}}})\circ\hat{\Gamma}_{\mathbf{I}_2}\Big)\nonumber\\
    \leq&\cD_{\Diamond}(\hat{\Gamma}_{\mathbf{I}_1}, \hat{\Gamma}_{\mathbf{I}_2}):=\cD_{\Diamond}(\mathbf{I}_1, \mathbf{I}_2)
    \end{align}
    \end{proof}
    Similarly, if instead of two single instruments we consider two sets of instruments, we can prove the following proposition:
\begin{proposition}
    If the sets of instruments  $\cI$ and $\overline{\cI}$ implements the sets of measurements $\cM$ and $\overline{\cM}$ and the sets of channels $\cC$ and $\overline{\cC}$ then 
 \begin{align}
     \widetilde{\cD}(\cM, \overline{\cM})\leq\widehat{\cD}(\cI, \overline{\cI})\nonumber\\
     \overline{\cD}(\cC, \overline{\cC})\leq\widehat{\cD}(\cI, \overline{\cI}).
 \end{align}\label{Prop:Meas_Chan_Inst_Dist_order}
\end{proposition}
The proof of this proposition is given in Appendix \ref{App:Meas_Chan_Inst_Dist_order}. 
\begin{theorem}
     Distance $\widehat{D}$ is contractive under post-processing of instruments.\label{Th:Dist_cont_post_process}
\end{theorem}
Detailed proof of this theorem can be found in \ref{App:Dist_cont_post_process}.
Based on the above-said distance measure $\widehat{\cD}$, for a set of instruments $\cI\subset\mathscr{I}(\cH_A,\cK_A)$ we can also define a resource measure as
\begin{align}
    \mathbbm{R}(\cI)=\min_{\cJ\in\cF_{(\cH_A,\cK_A)}}~\widehat{\cD}(\cI,\cJ),\label{Eq:Def_res_meas}
\end{align}
where $\cF_{(\cH_A,\cK_A)}$ is the set of free sets of instruments with the input Hilbert space $\cH_{A}$ and the output Hilbert space $\cK_{A}$. We proceed by proving a simple proposition:

\begin{proposition}
  For a generic instrument-based resource theory, if $\widehat{\cD}$ is monotonically non-increasing under free transformations, then $\mathbbm{R}$ is a resource measure.\label{Propsi:res_meas_prop}
\end{proposition}
\begin{proof}
    
    \begin{align}
        \mathbbm{R}(\cI)=\widehat{\cD}(\cI,\cJ^*),
    \end{align}
    where we have assumed the minimum in Eq. \eqref{Eq:Def_res_meas} occurs for a certain $\cJ^*\in\cF_{(\cH_A,\cK_A)}$.
    
    Note that $\mathbbm{R}(\cI)\geq 0$ for all $\cI$ as $\widehat{D}$ is always positive. Now, if $\cI\notin\cF_{(\cH_A,\cK_A)}$ then $\mathbbm{R}(\cI)>0$ from the faithfulness property of the distance measure. Whenever, $\cI\in\cF_{(\cH_A,\cK_A)}$ we get $\mathbbm{R}(\cI)=0$ as $\widehat{\cD}(\cI,\cI)=0$. Thus, the resource measure in Eq. \eqref{Eq:Def_res_meas} satisfies the property \hyperref[R1]{R1}.

    For a free transformation $\cV$ such that $\cV[\cI]\subset\mathscr{I}(\cH_{\tilde{A}},\cK_{\tilde{A}})$ we have
    \begin{align}
        \mathbbm{R}(\cI)=&\widehat{\cD}(\cI,\cJ^*)\nonumber\\
        \geq&\widehat{\cD}(\cV[\cI],\cV[\cJ^*]),\nonumber\\
        \geq&\min_{\tilde{\cJ}\in\cF_{(\cH_{\tilde{A}},\cK_{\tilde{A}})}}~\widehat{\cD}(\cV[\cI],\tilde{\cJ})\nonumber\\
        =&\mathbbm{R}(\cV[\cI]),
    \end{align}
    where $\cF_{(\cH_{\tilde{A}},\cK_{\tilde{A}})}$ is the set of free sets of instruments with the input Hilbert space $\cH_{\tilde{A}}$ and the output Hilbert space $\cK_{\tilde{A}}$.  Here, in the second line, we have used the assumption that $\widehat{\cD}$ is contractive under free transformations. Thus, the resource measure satisfies the property \hyperref[R2]{R2}. Hence $\mathbbm{R}$ is a valid resource measure for an instrument-based resource theory.
\end{proof}

Next, we study another quantification of instrument-based resources defined on extended Hilbert spaces. A trivially enlarged version of $\mathbf{I}=\{\Phi_a\}\in\mathscr{I}(\cH_A,\cK_A)$ can be defined as $\widehat{\mathbf{I}}_{(\cH_A\otimes\cH_B,\cK_A)}=\{\hat{\Phi}^a_{(\cH_A\otimes\cH_B,\cK_A)}=\Phi_a\otimes\tr_{\cH_B}\}$. A set of such instruments is similarly denoted by $\widehat{\cI}_{(\cH_A\otimes\cH_B,\cK_A)}$.

Based on this, we can also define another resource measure as
\begin{align}
    \overline{\mathbbm{R}}(\cI)=\inf_{\cH_B}\min_{\cJ_{(\cH_{AB},\cK_A)}\in\cF_{(\cH_{AB},\cK_A)}}~\widehat{\cD}(\widehat{\cI}_{(\cH_{AB},\cK_A)},\cJ_{(\cH_{AB},\cK_A)}),\label{Eq:Def_res_ext_meas}
\end{align}
where $\cF_{(\cH_{AB},\cK_A)}$ is the set of free sets of instruments with the input Hilbert space $\cH_{AB}:=\cH_{A}\otimes\cH_B$ (for notational simplicity, we use this notation in many places in this paper) and the output Hilbert space $\cK_{A}$. Now, we prove the following proposition.
\begin{proposition}
   For a generic instrument-based resource theory, if $\widehat{\cD}$ is monotonically non-increasing under free transformations, then $\overline{\mathbbm{R}}$ is a resource measure. .\label{Propsi:res_meas_ext_prop}
\end{proposition}
\begin{proof}
  Let us, for an arbitrary $\cH_B$, define
  \begin{align}
      \mathbbm{K}(\cI,\cH_B):=&\min_{\cJ_{(\cH_{AB},\cK_A)}\in\cF_{(\cH_{AB},\cK_A)}}~\widehat{\cD}(\widehat{\cI}_{(\cH_{AB},\cK_A)},\cJ_{(\cH_{AB},\cK_A)}),\label{Eq:K_min}\\
      =&\widehat{\cD}(\widehat{\cI}_{(\cH_{AB},\cK_A)},\cJ^*_{(\cH_{AB},\cK_A)})
  \end{align}
  such that
    \begin{align}
        \overline{\mathbbm{R}}(\cI)=\inf_{\cH_B}\mathbbm{K}(\cI,\cH_B),
    \end{align}
    where $\cJ^*_{(\cH_{AB},\cK_A)}\in\cF_{(\cH_{AB},\cK_A)}$.
    
    Note that $\overline{\mathbbm{R}}(\cI)\geq 0$ for all $\cI$ as $\mathbbm{\cK}(\cI)\geq0~\forall~\cH_B$. Now, if $\cI\notin\cF_{(\cH_A,\cK_A)}$ then by assumptions \hyperref[A1]{A1} and \hyperref[A4]{A4} we have $\widehat{\cI}_{(\cH_{AB},\cK_B)}\notin\cF_{(\cH_{AB},\cK_A)}~\forall~\cH_B$, Thus for any $\cH_B$ we can conclude $\mathbbm{K}(\cI,\cH_B)>0$ which implies $\overline{\mathbbm{R}}(\cI)>0$ . However, when $\cI\in\cF_{(\cH_A,\cK_A)}$ we get $\mathbbm{K}(\cI,\cH_B)=0$ as $\widehat{\cD}(\cI,\cI)=0$ which leads to $\overline{\mathbbm{R}}(\cI)=0$. Thus, the resource measure in Eq. \eqref{Eq:Def_res_ext_meas} also satisfies the property \hyperref[R1]{R1}.

    Next, for an arbitrary $\cH_B$ consider a free transformation $\cV$ such that the set of instruments $\cV[\cI]$ has input Hilbert space $\cH_{\tilde{A}}$ and output Hilbert space $\cK_{\tilde{A}}$. Then, with $\mathfrak{I}_B$ being an identity superchannel which maps a set of instruments to itself,  we have
    \begin{align}
      \mathbbm{K}(\cI,\cH_B)=&\widehat{\cD}(\widehat{\cI}_{(\cH_{AB},\cK_A)},\cJ^*_{(\cH_{AB},\cK_A)})\nonumber\\\geq&\widehat{\cD}(\cV\otimes\mathfrak{I}_B[\widehat{\cI}_{(\cH_{AB},\cK_A)}],\cV\otimes\mathfrak{I}_B[\cJ^*_{(\cH_{AB},\cK_A)}])\nonumber\\
      \geq&\min_{\cJ_{(\cH_{\tilde{A}B},\cK_{\tilde{A}})}\in\cF_{(\cH_{\tilde{A}B},\cK_{\tilde{A}})}}\widehat{\cD}(\widehat{\cV}[\cI]_{(\cH_{\tilde{A}B},\cK_{\tilde{A}})},\cJ_{(\cH_{\tilde{A}B},\cK_{\tilde{A}})})\nonumber\\
      \geq&\mathbbm{K}(\cV[\cI],\cH_B)
      \end{align}
      where in the second line we have used the assumptions \hyperref[A2]{A2}, \hyperref[A3]{A3} and the assumption that $\widehat{\cD}$ is contractive under free transformations. Hence
     \begin{align}
         \mathbbm{K}(\cI,\cH_B)\geq&\mathbbm{K}(\cV[\cI],\cH_B)~\forall~\cH_B\nonumber\\
         \text{or},\inf_{\cH_B}\mathbbm{K}(\cI,\cH_B)\geq&\inf_{\cH_B}\mathbbm{K}(\cV[\cI],\cH_B)\nonumber\\
         \text{or}, \overline{\mathbbm{R}}(\cI)\geq&\overline{\mathbbm{R}}(\cV[\cI])
     \end{align}
     Thus, the resource measure in Eq. \eqref{Eq:Def_res_ext_meas} also satisfies the property \hyperref[R2]{R2}. Hence $\overline{\mathbbm{R}}$ is a valid resource measure for an instrument-based resource theory.
\end{proof}
 
\begin{proposition}
    For an arbitrary set of instruments $\cI$, the resource measures $\mathbbm{R}$ and $\overline{\mathbbm{R}}$ satisfy the following relations.
     \begin{align}
         \overline{\mathbbm{R}}(\cI)\leq\mathbbm{R}(\cI)\leq\min\Big\{\frac{2\mathscr{R}(\cI)}{1+\mathscr{R}(\cI)},\frac{2\mathscr{W}(\cI)}{1+\mathscr{W}(\cI)}\Big\}
    \end{align}\label{Propsi:res_rob_and_weight}
\end{proposition}
Refer to Appendix \ref{App:res_rob_and_weight} for a detailed proof.

\begin{remark}
    \rm{We would like to mention that even if we replace minimum with infimum in the definitions of the resource measures $\mathbbm{R}$ and $\overline{\mathbbm{R}}$, their monotonicity can still be proven in the same way as Propositions \ref{Propsi:res_meas_prop} and \ref{Propsi:res_meas_ext_prop}. Furthermore, Proposition \ref{Propsi:res_rob_and_weight} would also hold in this case.} 
\end{remark}
Note that, if set of free objects are permutation-invariant then resource measures $\mathbbm{R}$ is also permutation-invariant. For illustration, consider a set of two instruments $\cI=\{\mathbf{I}_1,\mathbf{I}_2\}$. If for this set minimum in Eq. \eqref{Eq:Def_res_meas} happens for the  free set of instruments $\cJ^*=\{\mathbf{J}^*_1,\mathbf{J}^*_2\}$ then, for the permuted set $\cI_p=\{\mathbf{I}_2,\mathbf{I}_1\}$, the minimum will happen for the the free set of instruments $\cJ_P^*=\{\mathbf{J}^*_2,\mathbf{J}^*_1\}$. Hence $\mathbbm{R}(\cI)=\mathbbm{R}(\cI_P)$. Similar argument holds for $\overline{\mathbbm{R}}$ also.
\begin{remark}
    \rm{ Results similar to Proposition \ref{Propsi:res_meas_ext_prop}, and Proposition \ref{Propsi:res_rob_and_weight} have also been proved in Ref. \cite{Mitra_distance_measure_MBQR} for measurement-based quantum resources. But for completeness, we have proved Proposition \ref{Propsi:res_meas_ext_prop}, and Proposition \ref{Propsi:res_rob_and_weight} for \emph{instrument-based quantum resources} instead of measurement-based quantum resources.}
\end{remark}

\subsection{SDP formulation for computation of the resource measures $\mathbbm{R}~ \text{and}~\overline{\mathbbm{R}}$ }
\label{Sec:SDP_res_meas}
In this section, we will formulate an SDP to efficiently compute both the resource measures $\mathbbm{R}~ \text{and}~\overline{\mathbbm{R}}$ in Eqs. \eqref{Eq:Def_res_meas} and \eqref{Eq:Def_res_ext_meas} respectively, numerically. First, we will express the computation of $\mathbbm{R}$ as an optimization problem. Following \cite{Watrous_SDP}, for each $, \cD_{\Diamond}(\hat{\Gamma}_{\mathbf{I}_i},\hat{\Gamma}_{\mathbf{J}_i})$ in Eq. \eqref{Eq:def_dist_meas_of_set_inst} can be computed as the solution of the following optimization problem:
\begin{align}
    &\min_{Z_i}~2||\tr_\cY[Z_i]||_{\infty}\label{Eq:Primal_form_1}\\
    &\qquad\quad\text{s.t.}\nonumber\\
    &Z_i\geq J(\hat{\Gamma}_{\mathbf{I}_i})-J(\hat{\Gamma}_{\mathbf{J}_i})\quad\forall~i\nonumber\\
    &Z_i\geq0\quad\forall~i.\nonumber
\end{align}
Here $\cY=\cK_A\otimes\cH_{\Omega_I}$ and $J(\hat{\Gamma}_{\mathbf{I}_i}), J(\hat{\Gamma}_{\mathbf{J}_i})$ are the Choi matrices corresponding to the quantum channels $\hat{\Gamma}_{\mathbf{I}_i}, \hat{\Gamma}_{\mathbf{J}_i}\in\mathscr{C}(\cH_A,\cY)$ respectively.

As $Z_i\geq0$, the spectral norm $||\tr_\cY[Z_i]||_{\infty}$ is just the largest eigenvalue of the matrix $\tr_\cY[Z_i]$ which can be reformulated as the minimal value $m_i$ such that $m_i\mathbbm{1}\succeq\tr_{\cY}[Z_i]~\forall~i$. Thus optimization problem in Eqn. \eqref{Eq:Primal_form_1} can be recasted in the following form\cite{Tendick_dist_res_meas}:
\begin{align}
    &\min_{m_i,Z_i}~2m_i\label{Eq:Primal_form_2}\\
    &\qquad\quad\text{s.t.}\nonumber\\
    &Z_i\geq J(\hat{\Gamma}_{\mathbf{I}_i})-J(\hat{\Gamma}_{\mathbf{J}_i})\quad\forall~i\nonumber\\
    &Z_i\geq0\quad\forall~i.\nonumber\\
    &m_i\mathbbm{1}\succeq\tr_{\cY}[Z_i]\quad\forall~i\nonumber
\end{align}
From Eq. \eqref{Eq:def_dist_meas_of_set_inst}, we already know that $\widehat{\cD}(\cI,\cJ)=\max_{i\in\{1,\ldots,n\}}\cD_{\Diamond}(\hat{\Gamma}_{\mathbf{I}_i},\hat{\Gamma}_{\mathbf{J}_i})$. Thus, it can be computed by maximizing over the index $i$ in Eq. \eqref{Eq:Primal_form_2}. It can be brought into a more useful form by epigraph reformulation\cite{boyd_convex_opt}. As it is possible to write $\max_i m_i=\min_t~t~\text{subject to}~m_i\leq t ~\forall~i$, $\widehat{\cD}(\cI,\cJ)$ can labeled as the solution the following optimization problem:
\begin{align}
    &\min_{t,Z_i}~2t\label{Eq:Primal_form_3}\\
    &\qquad\quad\text{s.t.}\nonumber\\
    &Z_i\geq J(\hat{\Gamma}_{\mathbf{I}_i})-J(\hat{\Gamma}_{\mathbf{J}_i})\quad\forall~i\nonumber\\
    &Z_i\geq0\quad\forall~i.\nonumber\\
    &t\mathbbm{1}\geq\tr_{\cY}[Z_i]\quad\forall~i\nonumber
\end{align}

Now, from the definition of the resource measure $\mathbbm{R}(\cI)$ in Eq. \eqref{Eq:Def_res_meas}, it is clear that it can be computed by solving the optimization problem:
\begin{align}
    &\min_{\cJ,t,Z_i}~2t\label{Eq:Primal_form_res_meas}\\
    &\qquad\quad\text{s.t.}\nonumber\\
    &Z_i\geq J(\hat{\Gamma}_{\mathbf{I}_i})-J(\hat{\Gamma}_{\mathbf{J}_i})\quad\forall~i\nonumber\\
    &Z_i\geq0\quad\forall~i.\nonumber\\
    &t\mathbbm{1}\geq\tr_{\cY}[Z_i]\quad\forall~i\nonumber\\
    &\cJ=\{\mathbf{J}_i\}\in \cF_{(\cH_A,\cK_A)}\nonumber,
\end{align}
which is an SDP. Here $\cF_{(\cH_A,\cK_A)}$ is the set of free sets of instruments with the input Hilbert space $\cH_{A}$ and the output Hilbert space $\cK_{A}$, and it is considered to be a convex set.

Similarly, the resource measure $\overline{\mathbbm{R}}$ defined in Eq. \eqref{Eq:Def_res_ext_meas} can also be posed as a solution to an optimisation problem in a similar way as above. Proceeding as previously, first we compute the distance measure in Eq. \eqref{Eq:Def_res_ext_meas}. First, we note that $\mathbbm{K}(\cI,\cH_B)$ in Eq. \eqref{Eq:K_min} can be calculated as the solution of the optimization problem of the form Eq. \eqref{Eq:Primal_form_res_meas}. The solution clearly depends on the Hilbert space $\cH_B$. Also, we know that any Hilbert space $\cH_B$ of dimension $\dim{\cH_B}=n$ is isomorphic to $\mathbbm{C}^n$. Thus, for each $n$, we can solve the optimization problem of the form in Eq. \eqref{Eq:Primal_form_res_meas} (and let us denote the solution as $\mathbbm{R}_n(\cI)$) and then take an infimum over $n$ . For any finite $m$ and $1\leq n\leq m$, this amounts to finding a minimum value from a list having m entries. This can be mathematically represented as:
\begin{align}
    &\min_n \mathbbm{R}_n(\cI)\label{Eq:Opt_list_R}\\
    &\qquad s.t.\nonumber\\
    &n\in\{1,\ldots,m\},\nonumber
\end{align}
which can solved in the time of the order $m$\cite{cormen_algo}. Hence its is efficiently computable. More formally, we can write down the optimization problem in \eqref{Eq:Opt_list_R} as follows:
\begin{align}
    &\min_n\min_{\cJ^n,t_n,Z^n_i}~2t_n\label{Eq:Primal_form_ext_1}\\
    &\qquad\quad\text{s.t.}\nonumber\\
    &Z^n_i\geq J(\hat{\Gamma}_{\mathbf{\widehat{I}}^n_i})-J(\hat{\Gamma}_{\mathbf{J}^n_i})\quad\forall~i,n\nonumber\\
    &Z^n_i\geq0\quad\forall~i,n.\nonumber\\
    &t_n\mathbbm{1}\geq\tr_{\cY}[Z^n_i]\quad\forall~i,n\nonumber\\
    &\cJ^n=\{\mathbf{J}^n_i\}\in \cF_{(\cH_A\otimes\mathbbm{C}^n,\cK_A)}\nonumber\\
    &n\in\{1,\ldots,m\}\nonumber.
\end{align}
 Here $\widehat{\mathbf{I}}_i^n:=(\widehat{\mathbf{I}}_i)_{(\cH_A\otimes\mathbbm{C}^n,\cK_A)}$  which is of the form in Eq. \eqref{Eq:Def_res_ext_meas}. 
Thus, by solving this optimization problem, we can compute the resource measure $\overline{\mathbbm{R}}$ up to an \emph{arbitrary precision} specified by $m$. In other words, if we denote the solution of the optimization problem  \eqref{Eq:Primal_form_ext_1} as $\overline{\mathbbm{R}}_m(\cI)$, then
\begin{equation} \overline{\mathbbm{R}}(\cI)=\lim_{m\rightarrow\infty}\overline{\mathbbm{R}}_m(\cI).
\end{equation} 

\subsection{Some specific instrument-based quantum resources: characterization, quantification, hierarchies, and constructing resource theories}
\label{Subsec:Main:QI_QRTs}

In this section, we study the characterization, quantification, and hierarchies of some instrument-based quantum resources and try to construct their resource theories. In the following, we enlist and study (from a resource-theoretic point of view) some specific types of quantum instruments that can be considered as free objects for some resource theories.

\subsubsection{Trash-and-prepare instruments and the resource theory of information preservability}

Transmission of (classical or quantum or both) information through quantum channels (or more generally, through quantum instruments in the scenarios of sequential information extraction) is an important aspect of quantum communication technology. Therefore, the ability of quantum instruments to \emph{preserve information} is an important avenue to be explored. This \emph{motivates} us to construct a resource theory that helps us to study and quantify the ability of quantum instruments to preserve information in an elegant way. We call this resource theory the resource theory of \emph{information preservability}.

\begin{definition}
    A trash-and-prepare quantum instrument is a special type of one-outcome quantum instrument that contains a single trash-and-prepare quantum channel i.e., a quantum channel of the form
\begin{align}
    \Phi(\rho)=\tr[\rho]\sigma.
\end{align}
Here $\Phi:\cL(\cH)\rightarrow\cL(\cK)$, $\rho\in\cL(\cH)$ and $\sigma\in\cL(\cK)$. We denote a set of trash-and-prepare instruments with input Hilbert space $\cH$ and output Hilbert space $\cK$ as $\mathscr{I}_{TP}(\cH, \cK)$.
\end{definition}

As these instruments only provide a fixed classical output and a fixed quantum output, irrespective of the input, they destroy all the classical and quantum information present in the input state. Thus, the instruments that belong to the complement of the set of trash-and-prepare instruments are expected to preserve some information, and the ability to preserve information can be considered as a resource.

So the trash-and-prepare quantum instruments can be considered as free objects of this resource theory of information preservability, and they form a convex set. Its free transformations, which transform one set of trash-and-prepare instruments to another set of trash-and-prepare instruments, can be formulated as follows:
\begin{theorem}
     Consider a set of trash-and-prepare  instruments $\cI=\{\mathbf{I}^a=\{\Phi^a\}\in\mathscr{I}_{TP}(\cH, \cK)\}$. Let $\overline{\cJ}=\{\overline{\mathbf{J}}^b=\{\overline{\Phi}^b\}\in\mathscr{I}(\overline{\cH}, \overline{\cK})\}\}$ be a set of one-outcome instruments such that
    \begin{align}
       \overline{\Phi}^b=&\sum_{k_1}q(k_1)\sum_{a}p(a|b,k_1)\tilde{\Theta}^{b,k_1}\circ(\Phi^{a}\otimes\mathbbm{I}_{Q})\circ\Lambda^{\prime b,k_1}\nonumber\\
       +&\sum_{k_2}q(k_2)\sum_{a}p(a|b,k_2)\Theta^{\prime b,k_2}\circ(\Phi^{a}\otimes\mathbbm{I}_{Q^{\prime}})\circ\tilde{\Lambda}^{ b,k_2}\nonumber\\
       +&\sum_{k_3}q(k_3)\sum_{a}p(a|b,k_3)\tilde{\Gamma}^{ b,k_3}\circ\Phi^{a}\circ\tilde{\Delta}^{ b,k_3}\label{Eq:free_operation_tras_prep}
    \end{align}
    with $\sum_{k_1}q(k_1)+\sum_{k_2}q(k_2)+\sum_{k_3}q(k_3)=1$. Here, $\{\cI^{\prime k_1}=\{\mathbf{I}^{\prime b,k_1}=\{\Lambda^{\prime b,k_1}\}\in\mathscr{I}_{TP}(\overline{\cH},\cH\otimes Q)\}_b\}_{k_1}$ and $\{\cJ^{\prime k_2}=\{\mathbf{J}^{\prime b,k_2} =\{\Theta^{\prime b,k_2}\}\in\mathscr{I}_{TP}(\cK\otimes Q^{\prime},\overline{\cK})\}_b\}_{k_2}$ are two sets of sets of trash-and-prepare instruments and $\{\tilde{\cJ}^{k_1}=\{\tilde{\mathbf{J}}^{b,k_1}=\{\tilde{\Theta}^{b,k_1}\}\in\mathscr{I}(\cK\otimes Q,\overline{\cK})\}_b\}_{k_1}$ and $\{\tilde\cI^{k_2}=\{\tilde{\mathbf{I}}^{b,k_2}=\{\tilde{\Lambda}^{b,k_2}\}\in\mathscr{I}(\overline{\cH},\cH\otimes Q^{\prime})\}_b\}_{k_2}$ are two another sets of sets of instruments where each $\tilde{\Theta}^{b,k_1}$ and $\tilde{\Lambda}^{b,k_2}$ are general quantum channels $\forall~b,k_1,k_2$. Also, $\{\tilde{\cX}^{k_3}=\tilde{\mathbf{X}}^{b,k_3}=\{\tilde{\Gamma}^{b,k_3}\in\mathscr{I}(\cK,\overline{\cK})\}_b\}_{k_3}$ and $\{\tilde{\cO}^{k_3}=\tilde{\mathbf{O}}^{b,k_3}=\{\tilde{\Delta}^{b,k_3}\in\mathscr{I}(\overline{\cH},\cH)\}_b\}_{k_3}$ are two sets of set of quantum instruments with $\tilde{\Gamma}^{b,k_3}$ and $\tilde{\Delta}^{b,k_3}$ being generic quantum channels.
    Then 
    \begin{enumerate}
        \item $\overline{\cJ}$ is also a trash-and-prepare instrument. In other words, the transformation of the form given in Eq. \eqref{Eq:free_operation_tras_prep} can be considered as a free transformation of the resource theory of information preservability.
    \item Furthermore, a given arbitrary set of trash-and-prepare instruments can be transformed to a given arbitrary set of trash-and-prepare instruments through this transformation.\end{enumerate}\label{Th:free_op_trash_prepare}
\end{theorem}
\begin{proof}
    Note that $\Lambda^{\prime b,k_1}$ is trash-and-prepare for all $b,k_1$. Then the following channel
    \begin{align}
        \sum_{a}p(a|b,k_1)\tilde{\Theta}^{b,k_1}&\circ(\Phi^{a}\otimes\mathbbm{I}_{Q})\circ\Lambda^{\prime b,k_1}
    \end{align}
    is also trash-and-prepare for all $b,k_1$ as probabilistic mixing of an arbitrary set of trash-and-prepare instruments and the composition of a trash-and-prepare channel with any quantum channel also results in another trash-and-prepare channel. Similarly, again note that $\Theta^{\prime b,k_2}$ is trash-and-prepare for all $b,k_2$. Therefore, due to similar reasons, the channel 
    \begin{align}
        \sum_{a}p(a|b,k_2)\Theta^{\prime b,k_2}\circ(\Phi^{a}\otimes\mathbbm{I}_{Q^{\prime}})\circ\tilde{\Lambda}^{ b,k_2}
    \end{align}
    is also trash-and-prepare all $b,k_2$.

    Also as $\Phi^a$ is a trash-and-prepare for all $a$ by the same logic as above the channel
    \begin{align}
        \sum_{a}p(a|b,k_3)\tilde{\Gamma}^{ b,k_3}\circ\Phi^{a}\circ\tilde{\Delta}^{ b,k_3}
    \end{align}
is trash-and-prepare for all $b,k_3$.
    We know that trash-and-prepare channels form a convex set and therefore, the channel given in Eq. \eqref{Eq:free_operation_tras_prep} is trash-and-prepare for all $b$.  Thus an arbitrary transformation of the form given in Eq. \eqref{Eq:free_operation_tras_prep} always transforms a set of trash-and-prepare instrument to another set of trash-and-prepare instrument and therefore, it can be considered as a free transformation.

    The remaining thing is to show that given a pair of arbitrary sets of trash-and-prepare instruments, there exists a transformation of the form given in Eq. \eqref{Eq:free_operation_tras_prep} that transforms one set of that pair to the other set of the same pair. In order to do so, let us denote one of the sets of that pair as $\cI=\{\mathbf{I}^a=\{\Phi^a\}\in\mathscr{I}_{TP}(\cH, \cK)\}$. Our goal is to show that it can be transformed to the other set of the same pair, denoted as $\overline{\cJ}=\{\overline{\mathbf{J}}^b=\{\overline{\Phi}^b\}\in\mathscr{I}_{TP}(\overline{\cH}, \overline{\cK})\}\}$ using transformations of the form in Eq. \eqref{Eq:free_operation_tras_prep}. We proceed by defining
    \begin{align}
        \Lambda^{\prime b,1}&=~\tilde{\Gamma}_0\circ\overline{\Phi}^b,\nonumber\\
        \tilde{\Theta}^{b,1}&=(\tr_{\cK}\otimes\mathbbm{I}_{Q}),
    \end{align}
      where $\tilde{\Gamma}_0:\cL(\overline{\cK})\rightarrow\cL(\cH\otimes Q)$ with $\overline{\cK}=Q$ such that for all $\sigma\in\cL(\overline{\cK}),~ \tilde{\Gamma}_0(\sigma)= \ket{0}\bra{0}\otimes\sigma$  and clearly, $\mathbbm{I}_{Q}=\mathbbm{I}_{\overline{\cK}}$. Then, choosing $q(k_1)=\delta_{k_1,1}$ with $q(k_2),q(k_3)=0 ~\forall~k_2,k_3$, it can be easily shown that
    \begin{align}
   \sum_{k_1}q(k_1)\sum_{a}&p(a|b,k_1)\tilde{\Theta}^{b,k_1}\circ(\Phi^{a}\otimes\mathbbm{I}_{Q})\circ\Lambda^{\prime b,k_1}\nonumber\\
       +&\sum_{k_2}q(k_2)\sum_{a}p(a|b,k_2)\Theta^{\prime b,k_2}\circ(\Phi^{a}\otimes\mathbbm{I}_{Q^{\prime}})\circ\tilde{\Lambda}^{ b,k_2}\nonumber\\
       +&\sum_{k_3}q(k_3)\sum_{a}p(a|b,k_3)\tilde{\Gamma}^{ b,k_3}\circ\Phi^{a}\circ\tilde{\Delta}^{ b,k_3}= \overline{\Phi}^b
    \end{align}
    Hence a given arbitrary set of trash-and-prepare instruments $\cI=\{\mathbf{I}^a=\{\Phi^a\}\in\mathscr{I}_{TP}(\cH, \cK)\}$ can be transformed to another given arbitrary set of trash-and-prepare instruments $\overline{\cJ}=\{\mathbf{\overline{J}}^j=\{\overline{\Phi}^b\}\in\mathscr{I}_{TP}(\overline{\cH}, \overline{\cK})\}$ using the free transformations defined in Eq. (\ref{Eq:free_operation_tras_prep}).
\end{proof}
\begin{theorem}
    $\widehat{\cD}$ is monotonically non-increasing under the free transformations of the resource theory of information preservability. 
\end{theorem}
\begin{proof}
    From Theorem \ref{Th:free_op_trash_prepare}, we know that the free transformation of the resource theory of information preservability is given by Eq. \eqref{Eq:free_operation_tras_prep}. We observe that Eq. \eqref{Eq:free_operation_tras_prep} can be written in the form
\begin{align}
   \overline{\Phi}^b=\sum_{k_1}q(k_1)&\Theta^{b,k_1}_{post}\circ(\Sigma_{\cI}\otimes\mathbbm{I}_{Q})\circ\Theta^{b,k_1}_{pre}\nonumber\\
       +&\sum_{k_2}q(k_2)\Theta^{b,k_2}_{post}\circ(\Sigma_{\cI}\otimes\mathbbm{I}_{Q'})\circ \Theta^{b,k_2}_{pre}\nonumber\\
       +&\sum_{k_3}q(k_3)\Theta^{b,k_3}_{post}\circ\Sigma_{\cI}\circ\Theta^{b,k_3}_{pre},
\end{align}
by identifying $\Theta^{b,k_1}_{post}=\tilde{\Theta}^{b,k_1}$, $\Theta^{b,k_1}_{pre}=\Lambda^{\prime b,k_1}$, $\Theta^{b,k_2}_{post}=\Theta^{\prime b, k_2}$, $\Theta^{b,k_2}_{pre}=\tilde{\Lambda}^{b,k_2},\Theta^{b,k_3}_{post}=\tilde{\Gamma^{b,k_3}},\Theta^{b,k_3}_{pre}=\tilde{\Delta^{b,k_3}}$, and $\Sigma_{\cI}=\sum_a p(a|b) \Phi^a$. Symbolically, let us denote it by $\overline{\cJ}=\sum_{k_1}q(k_1)\cV^{k_1}[\cI]+\sum_{k_2}q(k_2)\cV^{k_2}[\cI]+\sum_{k_3}q(k_3)\cV^{k_3}[\cI]:=\cW[\cI]$ where the set of instruments $\cV^{k_1}[\cI]=\{\cV^{k_1}[\cI]^b=\tilde{\Theta}^{b,k_1}_{post} \circ(\Sigma_{\cI}\otimes\mathbbm{I}_{Q})\circ\tilde{\Theta}^{ b,k_1}_{pre}\}$, $\cV^{k_2}[\cI]:=\{\cV^{k_2}[\cI]^b=\{\Theta^{b,k_2}_{post} \circ(\Sigma_{\cI}\otimes\mathbbm{I}_{Q^{\prime}})\circ\Theta^{ b,k_2}_{pre}\}$ and $\cV^{k_3}[\cI]:=\{\cV^{k_3}[\cI]^b=\{\Theta^{b,k_3}_{post} \circ\Sigma_{\cI}\circ\Theta^{ b,k_3}_{pre}\}$ . Consider sets of instruments $\tilde\cI_1$ and $\tilde\cI_2$ such that $\tilde\cI_i=\cW[\cI_i]~\forall ~ i=1,2$. Then we can write
\begin{align}
    \widehat{\cD}(\tilde\cI_1,\tilde\cI_2):=&\widehat{\cD}(\cW[\cI_1],\cW[\cI_2]),\nonumber\\
    \leq &\sum_{k_1}q(k_1)\widehat{\cD}(\cV^{k_1}[\cI_1],\cV^{k_1}[\cI_2])\nonumber\\
    &\qquad\qquad+\sum_{k_2}q(k_2)\widehat{\cD}(\cV^{k_2}[\cI_1],\cV^{k_2}[\cI_2]),\nonumber\\
    &\qquad\qquad\qquad+\sum_{k_3}q(k_3)\widehat{\cD}(\cV^{k_3}[\cI_1],\cV^{k_3}[\cI_2])\nonumber\\
    \leq &\sum_{k_1}q(k_1)\widehat{\cD}(\cI_1,\cI_2)+\sum_{k_2}q(k_2)\widehat{\cD}(\cI_1,\cI_2)\nonumber\\
    &\qquad\qquad\qquad+\sum_{k_3}q(k_3)\widehat{\cD}(\cI_1,\cI_2),\nonumber\\
    = & \widehat{\cD}(\cI_1,\cI_2),
\end{align}
where in the second line we have used Eq. \eqref{Eq:Dist_inst_convex} and in the third line we have used Eq. \eqref{Eq:supermap_contractive_distance}. Hence $\widehat{D}$ is monotonically non-increasing under the free transformations of resource theory information preservability.
\end{proof}
As a result of the above theorem, using Proposition \ref{Propsi:res_meas_prop} and \ref{Propsi:res_meas_ext_prop}, we can also conclude that the distance-based resource measures in Eqs. \eqref{Eq:Def_res_meas} and \eqref{Eq:Def_res_ext_meas} are valid resource measures for the resource theory of information preservability. These are denoted as $\mathbbm{R}_{IP}$ and $\overline{\mathbbm{R}}_{IP}$, respectively.

\subsubsection{(Weak) Entanglement-breaking instruments and the resource theory of (strong) entanglement  preservability}

It is well-known that entanglement of a bipartite quantum state is a necessary resource for several information-theoretic tasks, e.g., quantum teleportation\cite{bennet_teleportation}, superdense coding\cite{harrow_superdense_coding}, quantum key distribution\cite{Shor_BB84}, etc. Therefore, the ability of a quantum channel ( or more generally of a quantum instrument in the scenarios of sequential use\cite{Colbeck_2020_NonLocal} of entanglement) to \emph{preserve entanglement} when it is acted on one side of a bipartite quantum state can be considered as a resource. This \emph{motivates} us to construct a resource theory that helps us to study and quantify the ability of quantum instruments to preserve entanglement in an elegant way. In some scenarios, probabilistic entanglement preservation might be enough, while in others we may require deterministic entanglement preservation. Therefore, we need two variants of a resource theory based on the ability of quantum instruments to preserve entanglement probabilistically or deterministically. We call these resource theories the resource theory of \emph{entanglement preservability} and the resource theory of \emph{strong entanglement preservability}, respectively, and study them one by one.
\begin{definition}
    An instrument $\mathbf{I}=\{\Phi_a\}\in\mathscr{I}(\cH,\cK)$ with $\sum_a\Phi_a=\Phi$
\begin{enumerate}
    \item is weak entanglement-breaking if $\Phi$ is entanglement-breaking. The set of such instruments is denoted as $\mathscr{I}_{WEB}(\cH,\cK)$.
    \item is entanglement-breaking if $\Phi_a$ is entanglement-breaking for all $a$. The set of such instruments is denoted as $\mathscr{I}_{EB}(\cH,\cK)$.
\end{enumerate}
\end{definition}

Clearly, the set of all trash-and-prepare instruments $\mathscr{I}_{TP}(\cH, \cK)$ is a subset of the set of entanglement-breaking quantum instruments $\mathscr{I}_{EB}(\cH,\cK)$ i.e., $\mathscr{I}_{TP}(\cH, \cK)\subseteq\mathscr{I}_{EB}(\cH,\cK)$ as trash-and-prepare form is a special case of measure-and-prepare form.

Before we start exploring the two variants of the resource theories as mentioned above, we prove the following proposition.

\begin{proposition}
    The set of all entanglement-breaking instruments is a subset of the set of all weak entanglement-breaking instruments for any given input Hilbert space $\cH$ and output Hilbert space $\cK$. \label{Proposi:eb_in_web}
\end{proposition}
\begin{proof}
    As the induced channel of an entanglement-breaking quantum instrument is entanglement-breaking, an entanglement-breaking quantum instrument is also weak entanglement-breaking. Therefore, the set of all entanglement-breaking quantum instruments is a subset of the set of all weak entanglement-breaking instruments. Hence, $\mathscr{I}_{EB}(\cH, \cK)\subseteq\mathscr{I}_{WEB}(\cH,\cK)$.
\end{proof}
In the following example, we show that in the qubit case, entanglement-breaking instruments are a strict subset of weak entanglement-breaking instruments.
\begin{example}
   \rm{ From Proposition \ref{Proposi:eb_in_web}, we know that the set of all entanglement-breaking instruments is a subset of the set of weak entanglement-breaking instruments for any given input Hilbert space $\cH$ and output Hilbert space $\cK$.
    Now, consider a four-outcome qubit instrument $\mathbf{I}=\{\Phi_a\}$ where for all $\rho\in\cL(\cH^{\mathbf{Q}})$
    \begin{align}
        \Phi_1(\rho)&=\frac{1}{2}\rho;\nonumber\\
        \Phi_2(\rho)&=\frac{1}{6}\sigma_x\rho\sigma_x;\nonumber\\
        \Phi_3(\rho)&=\frac{1}{6}\sigma_y\rho\sigma_y;\nonumber\\
        \Phi_4(\rho)&=\frac{1}{6}\sigma_z\rho\sigma_z.
    \end{align}
    where $\cH^{\mathbf{Q}}$ is the qubit Hilbert space. Note that all $\Phi_a$s have Kraus rank $1$ and none of the $\Phi_a$s has measure-and-prepare form and therefore, is not entanglement-breaking. But
    \begin{align}
        \Phi(\rho)&=\sum_a\Phi_a(\rho)\nonumber\\
        &=\frac{1}{3}\rho+\frac{2}{3}\frac{\Id_{2\times 2}}{2},
    \end{align}
    which is well-known to be entanglement-breaking\cite{Heinosaari_incomp_break_chan}. Hence, the set of all qubit entanglement-breaking instruments is a \emph{strict subset} of the set of all qubit weak entanglement-breaking instruments.}\label{Examp:eb_strict_in_web}
\end{example}

In the resource theory of entanglement preservability, the sets of entanglement-breaking instruments are the free objects that do not allow the preservation of entanglement even probabilistically. We construct the free transformations for it as follows.

Let us consider a set of sets of instruments $\{\cJ^{\prime k_2}=\{\mathbf{J}^{\prime j,k_2}=\{\Phi^{\prime j,k_2}_{a}\}\in\mathscr{I}(\overline{\cH},\cH\otimes Q)\}_j\}_{k_2}$ and a set of sets of sets of intruments $\{\{\tilde{\cI}^{j,k_1}=\{\tilde{\mathbf{I}}^{j,b,k_1}=\{\tilde{\Phi}^{j,b,k_1}_c\}\in\mathscr{I}(\cK\otimes Q^{\prime},\overline{\cK})\}_b\}_j\}_{k_1}$. Let us also consider a set of sets of entanglement-breaking instruments $\{\cJ^{\prime *k_1}=\{\mathbf{J}^{\prime * j,k_1}=\{\Phi^{\prime * j,k_1}_{a}\}\in\mathscr{I}_{EB}(\overline{\cH},\cH\otimes Q^{\prime})\}_j\}_{k_1}$ and a set of sets of sets of entanglement-breaking instruments $\{\{\tilde{\cI}^{*j,k_2}=\{\tilde{\mathbf{I}}^{*j,b,k_2}=\{\tilde{\Phi}^{*j,b,k_2}_c\}\in\mathscr{I}_{EB}(\cK\otimes Q,\overline{\cK})\}_b\}_j\}_{k_2}$. Then the following results hold.

\begin{theorem}
     Let $\cI=\{\mathbf{I}^a=\{\Phi^{a}_b\}\in\mathscr{I}_{EB}(\cH,\cK)\}$ be a set of entanglement-breaking instruments and $\overline{\cJ}=\{\overline{\mathbf{J}}^j=\{\overline{\Phi}^j_c\}\in\mathscr{I}(\overline{\cH},\overline{\cK})\}$ be a set of instruments such that
    \begin{align}
       \overline{\Phi}^j_c=\sum_{k_1}q(k_1)\sum_{a,b}&\tilde{\Phi}^{j,b,k_1}_c\circ(\Phi^{a}_b\otimes\mathbbm{I}_{Q^{\prime}})\circ\Phi^{\prime * j,k_1}_{a}\nonumber\\
       +&\sum_{k_2}q(k_2)\sum_{a,b}\tilde{\Phi}^{*j,b,k_2}_c\circ(\Phi^{a}_b\otimes\mathbbm{I}_{Q})\circ\Phi^{\prime j,k_2}_{a}\nonumber\\
       +&\sum_{k_3}q(k_3)\sum_{a,b}\tilde{\Gamma}^{ j,b,k_3}_c\circ\Phi^{a}_b\circ\tilde{\Delta}^{ j,k_3}_a.\label{Eq:free_operation_ent_break}
    \end{align}
    with $\sum_{k_1}q(k_1)+\sum_{k_2}q(k_2)+\sum_{k_3}q(k_3)=1$. Here, $\{\{\tilde{\cX}^{j,k_3}=\tilde{\mathbf{X}}^{j,b,k_3}=\{\tilde{\Gamma}^{j,b,k_3}_c\in\mathscr{I}(\cK,\overline{\cK})\}_b\}_j\}_{k_3}$ is a set of sets of sets of quantum instruments and $\{\tilde{\cO}^{k_3}=\tilde{\mathbf{O}}^{j,k_3}=\{\tilde{\Delta}_a^{j,k_3}\in\mathscr{I}(\overline{\cH},\cH)\}_j\}_{k_3}$ is another set of sets of quantum instruments. Then
    \begin{enumerate}
        \item $\overline{\cJ}$ is also a set of entanglement-breaking instruments. In other words, the transformation of the form given in Eq. \eqref{Eq:free_operation_ent_break} can be considered as a free transformation of the resource theory of entanglement preservability.
        \item A given arbitrary set of entanglement-breaking instruments can be transformed to another given arbitrary set of entanglement-breaking instruments through a transformation of the form given in Eq. \eqref{Eq:free_operation_ent_break}.
    \end{enumerate}
    \label{Th:free_op_ent_break}
\end{theorem}

\begin{proof}
Note that $\Phi^{\prime * j,k_1}_a$ is entanglement-breaking CP map $\forall~j,a,k_1$. Then the following CP map   
\begin{align}
    \sum_{a,b}\tilde{\Phi}^{j,b,k_1}_c\circ&(\Phi^{a}_b\otimes\mathbbm{I}_{Q^{\prime}})\circ\Phi^{\prime * j,k_1}_{a},\nonumber
\end{align}
is again entanglement-breaking $\forall~j,c,k_1$. This is because of the fact that the composition of an entanglement-breaking CP trace non-increasing map with any other CP trace non-increasing map, and the sum of entanglement-breaking CP trace non-increasing maps, both result in another entanglement-breaking CP trace non-increasing map. Similarly, as $\tilde{\Phi}^{*j,b,k_2}_c$ is also an entanglement-breaking CP trace non-increasing map $\forall~j,b,c,k_2$ by the same logic the CP map
\begin{align}
    \sum_{a,b}\tilde{\Phi}^{*j,b,k_2}_c\circ(\Phi^{a}_b\otimes\mathbbm{I}_{Q})\circ\Phi^{\prime j,k_2}_{a},
\end{align}
is also entanglement-breaking.

As $\Phi^a_b$ is entanglement-breaking $\forall~a,b$, by the same logic as above the CP map
\begin{align}
    \sum_{a,b}\tilde{\Gamma}^{ j,b,k_3}_c\circ\Phi^{a}_b\circ\tilde{\Delta}^{ j,k_3}_a,
\end{align}
is entanglement-breaking as well
.

From the structure entanglement-breaking CP trace non-increasing maps\cite{Horodecki_gen_EBC}, we know that they form a convex set. Hence $\overline{\Phi}^j_c$ in Eq. \eqref{Eq:free_operation_ent_break} is also entanglement-breaking CP trace non-increasing map $\forall~j,c$. Thus, a transformation of the form written in Eq. \eqref{Eq:free_operation_ent_break} transforms a set of entanglement-breaking instruments to another set of entanglement-breaking instruments, and because of this reason, it can be considered as a free transformation of the resource theory of entanglement preservability.

Once again, the remaining thing to show is whether, for a given pair of sets of arbitrary entanglement-breaking instruments, there exists a transformation of the form in Eq. \eqref{Eq:free_operation_ent_break} which transforms one set of instruments from that given pair to the other set of instruments from the same given pair. We proceed by considering a given arbitrary set of entanglement-breaking instruments $\cI=\{\mathbf{I}^a=\{\Phi^{a}_b\}\in\mathscr{I}_{EB}(\cH,\cK)\}$. Our goal is to show that it can be transformed into another given arbitrary set of entanglement-breaking instruments $\overline{\cJ}=\{\overline{\mathbf{J}}^j=\{\overline{\Phi}^j_c\}\in\mathscr{I}_{EB}(\overline{\cH},\overline{\cK})\}$. We define

     \begin{align}
        \Phi^{\prime  j,1}_{a}&=p(a)\Gamma_0\nonumber\\
        \tilde\Phi^{*j,b,1}_c&=\overline{\Phi}^j_c \circ(\tr_{\cK}\otimes\mathbbm{I}_{Q}),\qquad~\forall~j,b ,
    \end{align}
    where $\sum_a p(a)=1$, $\Gamma_0:\cL(\overline{\cH})\rightarrow\cL(\cH\otimes Q)$ with $\overline{\cH}=Q$ such that for all $\sigma\in\cL(\overline{\cH}),~ \Gamma_0(\sigma)= \ket{0}\bra{0}\otimes\sigma$  and clearly, $\mathbbm{I}_{Q}=\mathbbm{I}_{\overline{\cH}}$. Then, for all $\sigma\in\cL(\overline{\cK})$, we have
    
    \begin{align}
        \sum_{a,b}\tilde{\Phi}^{*j,b,1}_c\circ&(\Phi^{a}_{b}\otimes\mathbbm{I}_{Q})\circ\Phi^{\prime j,1}_{a}(\sigma)\nonumber\\
        =&\sum_{a,b}p(a)\overline{\Phi}^j_c \circ(\tr_{\cK}\otimes\mathbbm{I}_{Q})\circ(\Phi^{a}_{b}\otimes\mathbbm{I}_{Q})\circ\Gamma_0(\sigma)\nonumber\\
         =&\sum_{a}p(a)\overline{\Phi}^j_c \circ(\tr_{\cK}\otimes\mathbbm{I}_{Q})\circ\sum_b(\Phi^{a}_{b}\otimes\mathbbm{I}_{Q})\circ\Gamma_0(\sigma)\nonumber\\
          =&\sum_{a}p(a)\overline{\Phi}^j_c \circ(\tr_{\cK}\otimes\mathbbm{I}_{Q})\circ(\Phi^{a}\otimes\mathbbm{I}_{Q})\circ\Gamma_0(\sigma)\nonumber\\
          =&\sum_{a}p(a)\overline{\Phi}^j_c \circ(\tr_{\cK}\otimes\mathbbm{I}_{Q})\circ(\Phi^{a}(\ket{0}\bra{0})\otimes\mathbbm{I}_{Q}(\sigma))\nonumber\\
           =&\sum_{a}p(a)\overline{\Phi}^j_c(\sigma)\nonumber\\
           =&\overline{\Phi}^j_c(\sigma).
    \end{align}
    
    Thus by choosing $q(k_2)=\delta_{k_2,1}$ with $q(k_1),q(k_3)=0 ~\forall~k_1,k_3$, it follows that
    \begin{align}
       \sum_{k_1}q(k_1)\sum_{a,b}&\tilde{\Phi}^{j,b,k_1}_c\circ(\Phi^{a}_b\otimes\mathbbm{I}_{Q^{\prime}})\circ\Phi^{\prime * j,k_1}_{a}\nonumber\\
       +&\sum_{k_2}q(k_2)\sum_{a,b}\tilde{\Phi}^{*j,b,k_2}_c\circ(\Phi^{a}_{b}\otimes\mathbbm{I}_{Q})\circ\Phi^{\prime j,k_2}_{a}\nonumber\\
       +&\sum_{k_3}q(k_3)\sum_{a,b}\tilde{\Gamma}^{ j,b,k_3}_c\circ\Phi^{a}_b\circ\tilde{\Delta}^{ j,k_3}_a= \overline{\Phi}^j_c.
    \end{align}
Thus, any given set of arbitrary entanglement-instruments $\cI=\{\mathbf{I}^a=\{\Phi^{a}_b\}\in\mathscr{I}_{EB}(\cH,\cK)\}$ can be transformed to another given arbitrary set of entanglement-breaking instruments $\overline{\cJ}=\{\overline{\mathbf{J}}^j=\{\overline{\Phi}^j_c\}\in\mathscr{I}_{EB}(\overline{\cH},\overline{\cK})\}$ using the free transformations given in Eq. \eqref{Eq:free_operation_ent_break}.
\end{proof}

\begin{theorem}
    $\widehat{\cD}$ is monotonically non-increasing under the free transformations of the resource theory of entanglement preservability.\label{Th:Dist_monotone_EP}
\end{theorem}

The detailed proof of this theorem is given in Appendix \ref{App:Dist_monotone_EP}.
As a result of the above theorem, using Proposition \ref{Propsi:res_meas_prop} and \ref{Propsi:res_meas_ext_prop}, we can also conclude that the distance-based resource measures in Eqs. \eqref{Eq:Def_res_meas} and \eqref{Eq:Def_res_ext_meas} are valid resource measures for the resource theory of entanglement preservability. These are denoted as $\mathbbm{R}_{EP}$ and $\overline{\mathbbm{R}}_{EP}$, respectively.

Alternatively, in the resource theory of strong entanglement preservability, the sets of weak entanglement-breaking instruments are the free objects that allow the preservation of entanglement probabilistically but do not allow the preservation of entanglement deterministically.  We construct the free transformations as follows.

Let us have a set of sets of arbitrary instruments $\{\cJ^{\prime,k_1}=\{\mathbf{J}^{\prime j,k_1}=\{\Phi^{\prime j,k_1}_{a}\}\in\mathscr{I}(\overline{\cH},\cH\otimes Q)\}_j\}_{k_1}$ and a set of sets of sets of weak entanglement breaking instruments $\{\{\tilde{\cI}^{b,k_1}=\{\tilde{\mathbf{I}}^{j,b,k_1}=\{\tilde{\Phi}^{j,b,k_1}_c\}\in\mathscr{I}_{WEB}(\cK\otimes Q,\overline{\cK})\}_b\}_j\}_{k_1}$. Also consider a set of sets of instruments  $\{\{\tilde{\cX}^{j,k_2}=\{\tilde{\mathbf{X}}^{j,a,k_2}=\{\tilde{\Gamma}_c^{j,a,k_2}\}\in\mathscr{I}(\cK,\overline{\cK})\}_a\}_j\}_{k_2}$ where $\sum_c\tilde{\Gamma}_c^{j,a,k_2}=\tilde{\Gamma}^{j,a,k_2}$ is a quantum channel for all $j,a,k_2$ and another set of sets of sets of quantum instruments $\{\tilde{\cO}^{j,k_2}=\{\tilde{\mathbf{O}}^{j,k_2}=\{\tilde{\Delta}^{j,k_2}\}\in\mathscr{I}(\overline{\cH},\cH)\}_j\}_{k_2}$. Then the following theorem holds:
\begin{theorem}
     Let $\cI=\{\mathbf{I}^a=\{\Phi^{a}_b\}\in\mathscr{I}_{WEB}(\cH,\cK)\}$ be a set of weak entanglement-breaking instruments and $\overline{\cJ}=\{\overline{\mathbf{J}}^j=\{\overline{\Phi}^j_c\}\in\mathscr{I}(\overline{\cH},\overline{\cK})\}\}$ be a set of instruments such that
    \begin{align}
       \overline{\Phi}^j_c=\sum_{k_1}q(k_1)\sum_{a,b}&\tilde{\Phi}^{j,b,k_1}_c\circ(\Phi^{a}_b\otimes\mathbbm{I}_{Q})\circ\Phi^{\prime  j,k_1}_{a}\nonumber\\
       &+\sum_{k_2}q(k_2)\sum_{a,b}\tilde{\Gamma}_c^{j,a,k_2}\circ\Phi^a_b\circ\tilde{\Delta}_a^{j,k_2}.\label{Eq:free_operation_weak_ent_break}
    \end{align}
    with $\sum_{k_1}q(k_1)+\sum_{k_2}q(k_2)=1$. Then
    \begin{enumerate}
        \item $\overline{\cJ}$ is also a set of weak entanglement-breaking instruments. In other words, the transformation of the form given in Eq. \eqref{Eq:free_operation_weak_ent_break} can be considered as a free transformation of the resource theory of strong entanglement preservability.
        \item A given arbitrary set of weak entanglement-breaking instruments can be transformed to another given arbitrary set of weak entanglement-breaking instruments through a transformation of the form given in Eq. \eqref{Eq:free_operation_weak_ent_break}.
    \end{enumerate}
    \label{Th:free_op_weak_ent_break}
\end{theorem}
\begin{proof}
To prove $\overline{\cJ}$ to be a set of weak entanglement-breaking instruments, we have to prove that $\sum_c\overline{\Phi}^j_c~ \forall~j$ is an entanglement-breaking channel for all $j$. We know that $\sum_c\tilde{\Phi}^{j,b,k_1}_c:=\tilde{\Phi}^{j,b,k_1}$ is an entanglement-breaking channel $\forall~j,b,k_1$. Also, $\sum_c\sum_{a,b}\tilde{\Phi}^{j,b,k_1}_c\circ(\Phi^{a}_b\otimes\mathbbm{I}_{Q})\circ\Phi^{\prime j,k_1}_{a}=\sum_{a,b}\tilde{\Phi}^{j,b,k_1}\circ(\Phi^{a}_b\otimes\mathbbm{I}_{Q})\circ\Phi^{\prime j,k_1}_{a}$. We can see that each term in this sum can be considered as a composition of the entanglement-breaking channel $\tilde{\Phi}^{j,b,k_1}$ with CP trace non-increasing maps which is again an entanglement-breaking CP trace non-increasing map\cite{Ruskai_qubit_EBC}. But note that $\sum_{a,b}\tilde{\Phi}^{j,b,k_1}\circ(\Phi^{a}_b\otimes\mathbbm{I}_{Q})\circ\Phi^{\prime j,k_1}_{a}$ is also trace preserving and the sum of entanglement-breaking CP trace non-increasing maps is entanglement-breaking. Thus, $\sum_{a,b}\tilde{\Phi}^{j,b,k_1}\circ(\Phi^{a}_b\otimes\mathbbm{I}_{Q})\circ\Phi^{\prime j,k_1}_{a}$ is an entanglement-breaking quantum channel $\forall~j,k_1$. Similarly, as $\sum_b\Phi^a_b=\Phi^a$ is an entanglement-breaking channel for all $a$, $\sum_c\sum_{a,b}\tilde{\Gamma}_c^{j,a,k_2}\circ\Phi^a_b\circ\tilde{\Delta}_a^{j,k_2}=\sum_{a}\tilde{\Gamma}^{j,a,k_2}\circ\Phi^a\circ\tilde{\Delta}_a^{j,k_2}$ is also an entanglement-breaking quantum channel $\forall~j,k_2$. Now we know that a convex combination of entanglement-breaking channels is again an entanglement-breaking channel. Thus $\sum_c\overline{\Phi}^j_c$ is also an entanglement-braking channel $\forall~j$ or equivalently, $\overline{\cJ}$ is a set of weak entanglement-breaking instruments. Hence, the transformations of the form given in Eq. \eqref{Eq:free_operation_weak_ent_break} can be considered as the free transformations for resource theory of strong incompatibility preservability. 

    The next thing we have to show is that for two given arbitrary sets of weak entanglement-breaking instruments, there exists a transformation of the form given in Eq. \eqref{Eq:free_operation_weak_ent_break} that transforms one set of the given pair to the other set of the same given pair. To show this, we first consider two given sets of arbitrary weak entanglement-breaking instruments $\cI=\{\mathbf{I}^a=\{\Phi^{a}_b\}\in\mathscr{I}_{WEB}(\cH,\cK)\}$ and $\overline{\cJ}=\{\overline{\mathbf{J}}^j=\{\overline{\Phi}^j_c\}\in\mathscr{I}_{WEB}(\overline{\cH},\overline{\cK})\}$. Next, we define
     \begin{align}
        \Phi^{\prime  j,1}_{a}&=p(a)\Gamma_0\nonumber\\
        \tilde\Phi^{j,b,1}_c&=\overline{\Phi}^j_c\circ(\tr_{\cK}\otimes\mathbbm{I}_{Q}),\qquad~\forall~j,b ,
    \end{align}
    where $\sum_ap(a)=1$, $\Gamma_0:\cL(\overline{\cH})\rightarrow\cL(\cH\otimes Q)$ with $\overline{\cH}=Q$ such that for all $\sigma\in\cL(\overline{\cH}),~ \Gamma_0(\sigma)= \ket{0}\bra{0}\otimes\sigma$  and clearly, $\mathbbm{I}_{Q}=\mathbbm{I}_{\overline{\cH}}$. Then, by choosing $q(k_1)=\delta_{k_1,1}$ with $q(k_2)=0 ~\forall~k_2$, it can be easily shown that
    \begin{align}
         \sum_{k_1}q(k_1)\sum_{a,b}&\tilde{\Phi}^{j,b,k_1}_c\circ(\Phi^{a}_b\otimes\mathbbm{I}_{Q})\circ\Phi^{\prime  j,k_1}_{a}\nonumber\\
       &+\sum_{k_2}q(k_2)\sum_{a,b}\tilde{\Gamma}_c^{j,a,k_2}\circ\Phi^a_b\circ\tilde{\Delta}_a^{j,k_2}=\overline{\Phi}^j_c.
    \end{align}
    Thus, the transformation of the form given in Eq. \eqref{Eq:free_operation_weak_ent_break} transforms a given arbitrary set of weak entanglement-breaking instruments $\cI=\{\mathbf{I}^a=\{\Phi^{a}_b\}\in\mathscr{I}_{WEB}(\cH,\cK)\}$ to another given arbitrary set of weak entanglement-breaking instruments $\overline{\cJ}=\{\overline{\mathbf{J}}^j=\{\overline{\Phi}^j_c\}\in\mathscr{I}_{WEB}(\overline{\cH},\overline{\cK})\}$
\end{proof}
\begin{theorem}
    $\widehat{\cD}$ is monotonically non-increasing under the free transformations of the resource theory of strong entanglement preservability.\label{Th:Dist_monotone_SEP}
\end{theorem}

The detailed proof is provided in Appendix \ref{App:Dist_monotone_SEP}. As a result of the above theorem, using Proposition \ref{Propsi:res_meas_prop} and \ref{Propsi:res_meas_ext_prop}, we can also conclude that the distance-based resource measures in Eqs. \eqref{Eq:Def_res_meas} and \eqref{Eq:Def_res_ext_meas} are valid resource measures for the resource theory of strong entanglement preservability. These are denoted as $\mathbbm{R}_{SEP}$ and $\overline{\mathbbm{R}}_{SEP}$, respectively.

\begin{remark}
    \rm{It should be mentioned here that in Appendix \ref{App:EBC_compact} we have shown that the minimum exists in Eq. \eqref{Eq:Def_res_meas} and Eq. \eqref{Eq:Def_res_ext_meas} in the context of the resource theory of entanglement preservability. This conclusion holds for all instrument-based resource theories considered in this work.}\label{Re:EBC_compact}
\end{remark}

\subsubsection{(Weak) Incompatibility-breaking instruments and the resource theory of (strong) incompatibility  preservability}

It is well-known that the incompatibility of measurements is a necessary resource for several information-theoretic tasks, e.g., quantum state discrimination\cite{Skrzypczyk_incomp_state_disc}, quantum random access codes\cite{debasish_incom_random_acces_code, Carmeli_incomp_meas_qrac}, etc. Therefore, the ability of quantum channels (or more generally of an instrument in sequential scenarios\cite{Sasmal_2018_Steering,Mohan_2019}) to \emph{preserve incompatibility} of measurements, when it is acted on a set of measurements in the Heisenberg picture, can be considered as a resource. This \emph{motivates} us to construct the resource theory that helps us to study and quantify the ability of quantum instruments to preserve incompatibility of measurements in an elegant way. In some scenarios, we might need incompatibility preservation at least when the classical outcomes of the instrument are recorded, while some other scenarios can arise where we need incompatibility preservation even when the classical outcomes of the instrument are unknown. Therefore, again similar to the case of entanglement preservability, we need two variants of a resource theory based on the ability of quantum instruments to preserve incompatibility when the classical outcomes of the instrument are recorded or even when the classical outcomes of the instrument are unknown. We call these resource theories the resource theory of \emph{incompatibility preservability} and the resource theory of \emph{strong incompatibility preservability}, respectively.
\begin{definition}
    An instrument $\mathbf{I}=\{\Phi_a\}\in\mathscr{I}(\cH,\cK)$ with $\sum_a\Phi_a=\Phi$ 
\begin{enumerate}
   
    \item is weak incompatibility-breaking if $\Phi$ is incompatibility-breaking. The set of such instruments is denoted as $\mathscr{I}_{WIB}(\cH,\cK)$.
    \item is incompatibility-breaking if $\mathbf{I}^{\dagger}[\cM]$ is compatible for an arbitrary set $\cM\subset\mathscr{M}(\cK)$. The set of such instruments is denoted as $\mathscr{I}_{IB}(\cH,\cK)$.
\end{enumerate}
\end{definition}

From the convexity of the comaptible measurements, it can be easily seen that the set of incompatibility-breaking instruments is convex. it can Before we start exploring the two variants of the resource theories as mentioned above, we prove the following results.

\begin{proposition}
    The set of all incompatibility-breaking instruments is a subset of the set of all weak incompatibility-breaking instruments for any given input Hilbert space $\cH$ and output Hilbert space $\cK$.
    \label{Proposi:ib_in_wib}
\end{proposition}

\begin{proof}
    As the induced channel of an incompatibility-breaking quantum instrument is incompatibility-breaking, an incompatibility-breaking quantum instrument is also weak incompatibility-breaking. Therefore, the set of all incompatibility-breaking quantum instruments is a subset of the set of all weak incompatibility-breaking instruments. Hence, $\mathscr{I}_{IB}(\cH, \cK)\subseteq\mathscr{I}_{WIB}(\cH,\cK)$.
\end{proof} 
\begin{example} \rm{ From Proposition \ref{Proposi:ib_in_wib}}, we know that the set of all incompatibility-breaking instruments is a subset of the set of weak incompatibility-breaking instruments for any given input Hilbert space $\cH$ and output Hilbert space $\cK$. Now, consider the same four-outcome qubit instrument $\mathbf{I}=\{\Phi_a\}$ that has been used in the Example \ref{Examp:eb_strict_in_web}. Now, as $\Phi=\sum^4_{a=1}\Phi_a$ is entanglement-breaking, it is also incompatibility-breaking. Now, consider two measurements $A=\{A(1)=\ket{0}\bra{0}, A(2)=\ket{1}\bra{1}\}$ and $B=\{B(1)=\ket{+}\bra{+}, B(2)=\ket{-}\bra{-}\}$. Clearly, the pair $(A,B)$ is incompatible. Then
    \begin{align}
        \mathbf{I}^{\dagger}[A](1,1)=\frac{1}{2}\ket{0}\bra{0} ;~\mathbf{I}^{\dagger}[B](1,1)=\frac{1}{2}\ket{+}\bra{+},\nonumber\\
        \mathbf{I}^{\dagger}[A](2,1)=\frac{1}{6}\ket{1}\bra{1} ;~\mathbf{I}^{\dagger}[B](2,1)=\frac{1}{6}\ket{+}\bra{+},\nonumber\\
        \mathbf{I}^{\dagger}[A](3,1)=\frac{1}{6}\ket{1}\bra{1} ;~\mathbf{I}^{\dagger}[B](3,1)=\frac{1}{6}\ket{-}\bra{-},\nonumber\\
        \mathbf{I}^{\dagger}[A](4,1)=\frac{1}{6}\ket{0}\bra{0} ;~\mathbf{I}^{\dagger}[B](4,1)=\frac{1}{6}\ket{-}\bra{-},\nonumber
        \end{align}
        \begin{align}
        \mathbf{I}^{\dagger}[A](1,2)=\frac{1}{2}\ket{1}\bra{1} ;~\mathbf{I}^{\dagger}[B](1,2)=\frac{1}{2}\ket{-}\bra{-},\nonumber\\
        \mathbf{I}^{\dagger}[A](2,2)=\frac{1}{6}\ket{0}\bra{0} ;~\mathbf{I}^{\dagger}[B](2,2)=\frac{1}{6}\ket{-}\bra{-},\nonumber\\
        \mathbf{I}^{\dagger}[A](3,2)=\frac{1}{6}\ket{0}\bra{0} ;~\mathbf{I}^{\dagger}[B](3,2)=\frac{1}{6}\ket{+}\bra{+},\nonumber\\
        \mathbf{I}^{\dagger}[A](4,2)=\frac{1}{6}\ket{1}\bra{1} ;~\mathbf{I}^{\dagger}[B](4,2)=\frac{1}{6}\ket{+}\bra{+}.\nonumber\\
    \end{align}
    Note that if a pair of measurements is compatible, then all its post-processings are compatible. Therefore, if we can show that there exists a post-processing of the pair $(\mathbf{I}^{\dagger}[A],\mathbf{I}^{\dagger}[B])$ is incompatible then the pair $(\mathbf{I}^{\dagger}[A],\mathbf{I}^{\dagger}[B])$ is incompatible. Now, consider the post-processing
    
    \begin{align}
        M(z)=\sum_{x,y}\nu_{zxy}\mathbf{I}^{\dagger}[A](x,y);~N(z)=\sum_{x,y}\nu_{zxy}\mathbf{I}^{\dagger}[B](x,y),
    \end{align}
    where $\nu_{111}=\nu_{141}=\nu_{122}=\nu_{132}=\nu_{221}=\nu_{231}=\nu_{212}=\nu_{242}=1$ and all $\nu_{zxy}$s are zero. Clearly, it is a valid post-processing as $\sum_{z}\nu_{zxy}=1$. Note that 

    \begin{align}
        M&=\{N(1)=\ket{0}\bra{0},M(2)=\ket{1},\bra{1}\}\nonumber\\
        N&=\{N(1)=\frac{2}{3}\ket{+}\bra{+}+\frac{1}{3}\ket{-}\bra{-},\nonumber\\
        &\qquad \qquad \qquad \qquad  N(2)=\frac{1}{3}\ket{+}\bra{+}+\frac{2}{3}\ket{-}\bra{-}\}.
    \end{align}
    Note that $M$ is a PVM and $M$ does not commute with $N$. Therefore, the pair $(M,N)$ is incompatible\cite{Heinosaari_incomp_review}. Hence, the pair $(\mathbf{I}^{\dagger}[A],\mathbf{I}^{\dagger}[B])$ is incompatible. Therefore, the instrument $\mathbf{I}$ is not incompatibility-breaking. Hence, the set of all qubit incompatibility-breaking instruments is a \emph{strict subset} of the set of all qubit weak incompatibility-breaking instruments.\label{Examp:ib_strict_in_wib} 
\end{example}

\begin{proposition}
    The set of all qubit incompatibility-breaking instruments is not a subset of the set of weak entanglement-breaking instruments.\label{Proposi:ib_notin_web}
\end{proposition}

\begin{proof}
Consider a one-outcome quantum instrument (i.e., a quantum channel)
    
    \begin{align}
        \Lambda(\rho)=\frac{5}{12}\rho+ \frac{7}{12}\frac{\Id_{2\times 2}}{2}.\label{Eq:IB_not_WEB}
    \end{align}
    It is known that it is incompatibility-breaking, but not weak entanglement-breaking\cite{Heinosaari_incomp_break_chan} (for one-outcome instruments, the notion of weak entanglement-breaking is the same as the notion of entanglement-breaking).
\end{proof}

\begin{proposition}
    The set of all qubit weak entanglement-breaking instruments is not a subset of the set of incompatibility-breaking instruments.\label{Proposi:web_notin_ib}
\end{proposition}

\begin{proof}
    Consider the same four-outcome instrument $\mathbf{I}=\{\Phi_a\}$ that has been used in the proof of example \ref{Examp:eb_strict_in_web}. From the Example \ref{Examp:eb_strict_in_web} and Proposition \ref{Examp:ib_strict_in_wib}, please note that $\mathbf{I}$ weak entanglement-breaking instrument. But $\mathbf{I}$ is not an incompatiblity-breaking instrument.
\end{proof}

\begin{proposition}
    The set of all entanglement-breaking instruments is a subset of the set of all incompatibility-breaking instruments for any given input Hilbert space $\cH$ and output Hilbert space $\cK$.\label{Proposi:eb_in_ib}
\end{proposition}

\begin{proof}
    Consider an entanglement-breaking instrument $\mathbf{I}=\{\Phi_x\}\in\mathscr{I}_{EB}(\cH,\cK)$. Then, each $\Phi_x$ acting on an arbitrary density matrix $\rho\in\cL(\cH)$ has the form \cite{Horodecki_gen_EBC}
    \begin{align}
        \Phi_x(\rho)=\sum_i\tr[A(x,i)\rho]\sigma^x_i,
    \end{align}
    where $\sigma^x_i\in\cL(\cK)$ is a valid density matrix $\forall~x,i$ and $A(x,i)\geq0,~\forall~x,i$ with $\sum_{i,x}A(x,i)=\mathbbm{1}_\cH$. Next, consider a set of arbitrary measurements $\cM=\{M_1, M_2, M_3,\ldots, M_n\}$ where $n$ is arbitrary and $M_y=\{M_y(m_y)\}_{m_y}$ for $y=1,2,\ldots,n$. Here $m_y$ denotes the outcome for the measurement $M_y~\forall y$. Under the action of the instrument $\mathbf{I}$ in the Heisenberg picture, the given set of measurements is transformed as $\{\mathbf{I}^{\dagger}(M_1), \mathbf{I}^{\dagger}(M_2), \ldots, \mathbf{I}^{\dagger}(M_n)\}$ with 
    \begin{align}
        \mathbf{I}^{\dagger}(M_y)=\{\sum_i\tr[\sigma^{x_y}_iM_y(m_y)]A(x_y,i)\}_{(x_y,m_y)\in\Omega_{\mathbf{I}}\times\Omega_{M_y}}~\forall~y.
    \end{align}
    Let $g_y=(x_y,m_y)~\forall y$ denotes the outcome of the transformed measurement for every y. Consider the matrix $G(g_1,\ldots,g_n)$ defined as
    \begin{align}
        G(g_1&,\ldots,g_n):=
        &\begin{cases}
        \sum_i\prod_y\tr[\sigma^{x}_iM_y(m_y)]A(x,i),  ~\text{if } x_y=x,~\forall~y,\\
        0,\qquad\qquad\qquad\qquad\qquad \text{otherwise}.
        \end{cases}\label{Eq:G_join_meas_EB_IB}
    \end{align}
    Then from Eq. \eqref{Eq:G_join_meas_EB_IB}, it can be easily shown that each $G(g_1,\ldots,g_n)\geq 0$ and $\sum_{g_1,\ldots,g_n}G(g_1,\ldots,g_n)=\Id_{\cH}$. Hence, $G=\{G(g_1,g_2,g_3,\ldots,g_n)\}$ is valid measurement.

    It can be shown easily then
    \begin{align}
       \mathbf{I}^{\dagger}(M_y)(g_y)=&\sum_{\{g_1,g_2,g_3,\ldots,g_n\}\setminus g_y}  G(g_1,g_2,g_3,\ldots,g_n)\nonumber\\
       =&\sum_{\{m_1,m_2,\ldots,m_n\}\setminus m_y}  G(x_y,m_1,m_2,\ldots,m_n)
    \end{align}
    where, $G(x_y,m_1,m_2,\ldots,m_n):=\sum_i\prod_y\tr[\sigma^{x}_iM_y(m_y)]A(x,i)$ with $x_y=x~\forall y$. Hence, $G$ is joint measurement of the set $\{\mathbf{I}^{\dagger}(M_1),\mathbf{I}^{\dagger}(M_2),\ldots,\mathbf{I}^{\dagger}(M_n)\}$ and therefore, the set of measurements $\{\mathbf{I}^{\dagger}(M_1),\mathbf{I}^{\dagger}(M_2),\ldots,\mathbf{I}^{\dagger}(M_n)\}$ is compatible. Thus, $\mathbf{I}$ is an incompatibility-breaking instrument. Hence, $\mathscr{I}_{EB}(\cH,\cK)\subseteq\mathscr{I}_{IB}(\cH,\cK)$.
\end{proof} 
\begin{remark}
    \rm{From Eq. \eqref{Eq:IB_not_WEB}, we know that one-outcome qubit quantum instrument $\Lambda$ is incompatibility-breaking but not entanglement-breaking. Thus, we can conclude that the set of all qubit entanglement-breaking instruments is a strict subset of the set of all incompatibility-breaking instruments. Also, as we know, entanglement-breaking channels are a strict subset of incompatibility-breaking channels \cite{Heinosaari_incomp_break_chan}, thus it follows that $\mathscr{I}_{WEB}(\cH,\cK)\subseteq\mathscr{I}_{WIB}(\cH,\cK)$, although incompatibility-breaking properties of quantum instruments have not been studied in Ref. \cite{Heinosaari_incomp_break_chan} as its main focus was quantum channels. For clarity on the subset relations among the discussed class of instruments, we refer the reader to Fig. \ref{fig_venn}.}
\end{remark}

\begin{figure}[!h]
    \centering
    \includegraphics[height=150px, width =236px]{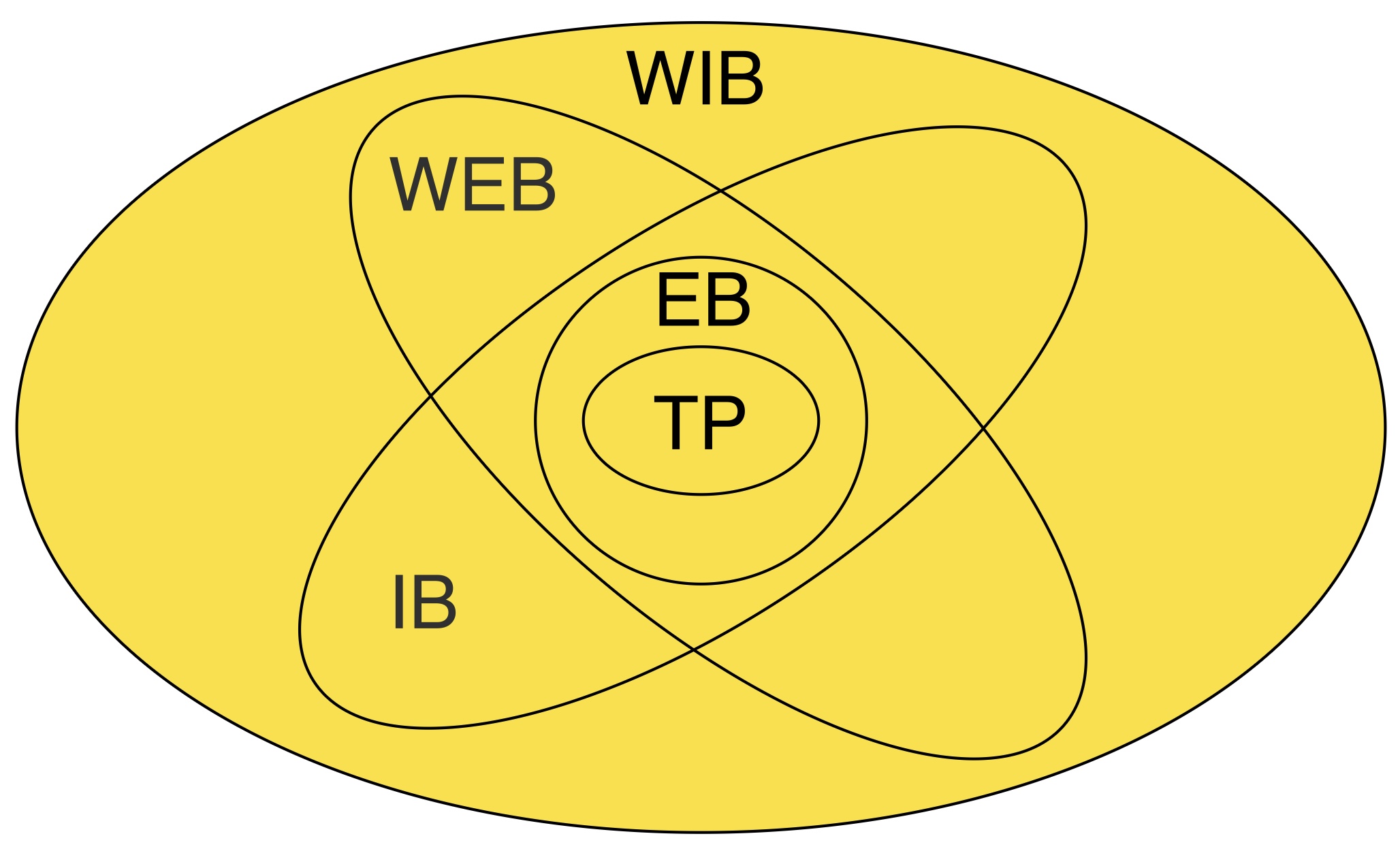}
     \caption{This Venn diagram qualitatively shows the hierarchies (subset relations) among different classes of instruments. More specifically, from the discussion till now, we have $\mathscr{I}_{TP}(\cH, \cK)\subseteq\mathscr{I}_{EB}(\cH,\cK)\subseteq \mathscr{I}_{WEB}(\cH,\cK)\subseteq\mathscr{I}_{WIB}(\cH,\cK) $ and $\mathscr{I}_{TP}(\cH, \cK)\subseteq\mathscr{I}_{EB}(\cH,\cK)\subseteq \mathscr{I}_{IB}(\cH,\cK)\subseteq\mathscr{I}_{WIB}(\cH,\cK) $ for arbitrary $\cH$ and $\cK$.}\label{fig_venn} 
\end{figure}

In the resource theory of incompatibility preservability, the sets of incompatibility-breaking instruments are the free objects that destroy the incompatibility of any sets of measurements, even if the classical outcomes of the instrument are recorded. We construct the free transformations as follows.

Let us consider a set of sets of instruments $\{\cJ^{\prime k_2}=\{\mathbf{J}^{\prime j,k_2}=\{\Phi^{\prime j,k_2}_{a}\}\in\mathscr{I}(\overline{\cH},\cH\otimes Q)\}_j\}_{k_2}$ and a set of sets of sets of intruments $\{\{\tilde{\cI}^{j,k_1}=\{\tilde{\mathbf{I}}^{j,b,k_1}=\{\tilde{\Phi}^{j,b,k_1}_c\}\in\mathscr{I}(\cK\otimes Q^{\prime},\overline{\cK})\}_b\}_j\}_{k_1}$. Let us also consider a set of sets of incompatibility-breaking channels $\{\{\Phi^{\prime * j,k_1}\}_j\}_{k_1}$, where $\Phi^{\prime * j,k_1}:\cL(\overline{\cH})\rightarrow\cL(\cH\otimes Q^{\prime})$ and a set of sets of sets of incompatibility-breaking instruments $\{\{\tilde{\cI}^{*j,k_2}=\{\tilde{\mathbf{I}}^{*j,b,k_2}=\{\tilde{\Phi}^{*j,b,k_2}_c\}\in\mathscr{I}_{IB}(\cK\otimes Q,\overline{\cK})\}_b\}_j\}_{k_2}$. Then the following results hold.

\begin{theorem}
     Let $\cI=\{\mathbf{I}^a=\{\Phi^{a}_b\}\in\mathscr{I}_{IB}(\cH,\cK)\}$ be a set of incompatibility-breaking instruments and $\overline{\cJ}=\{\overline{\mathbf{J}}^j=\{\overline{\Phi}^j_c\}\in\mathscr{I}(\overline{\cH},\overline{\cK})\}$ be a set of instruments such that
    \begin{align}
      \overline{\Phi}^j_c=&\sum_{k_1}q(k_1)\sum_{a,b}p(a|j)\tilde{\Phi}^{j,b,k_1}_c\circ(\Phi^{a}_b\otimes\mathbbm{I}_{Q^{\prime}})\circ\Phi^{\prime * j,k_1}\nonumber\\
       &+\sum_{k_2}q(k_2)\sum_{a,b}\tilde{\Phi}^{*j,b,k_2}_c\circ(\Phi^{a}_b\otimes\mathbbm{I}_{Q})\circ\Phi^{\prime j,k_2}_{a}\nonumber\\
       &+\sum_{k_3}q(k_3)\sum_{a,b}\tilde{\Gamma}^{ j,a,k_3}_c\circ\Phi^{a}_b\circ\tilde{\Delta}^{ j,k_3}_a,\label{Eq:free_operation_incom_break}
    \end{align}
     where $\sum_a p(a|j)=1,~\forall j$ with $\sum_{k_1}q(k_1)+\sum_{k_2}q(k_2)+\sum_{k_3}q(k_3)=1$. Here, $\{\{\tilde{\cX}^{j,k_3}=\tilde{\mathbf{X}}^{j,a,k_3}=\{\tilde{\Gamma}^{j,a,k_3}_c\in\mathscr{I}(\cK,\overline{\cK})\}_b\}_j\}_{k_3}$ is an arbitrary set of sets of sets of quantum instruments and $\{\tilde{\cO}^{k_3}=\tilde{\mathbf{O}}^{j,k_3}=\{\tilde{\Delta}_a^{j,k_3}\in\mathscr{I}(\overline{\cH},\cH)\}_j\}_{k_3}$ is another set of sets of quantum instruments. Then
    \begin{enumerate}
        \item $\overline{\cJ}$ is also a set of incompatibility-breaking instruments. In other words, the transformation of the form given in Eq. \eqref{Eq:free_operation_incom_break} can be considered as a free transformation of the resource theory of incompatibility preservability.
        \item A given arbitrary set of incompatibility-breaking instruments can be transformed to another given arbitrary set of incompatibility-breaking instruments through a transformation of the form given in Eq. \eqref{Eq:free_operation_incom_break}.
    \end{enumerate}
      \label{Th:free_op_incom_break}
\end{theorem}
\begin{proof}
     Consider a set of arbitrary measurements $\cM=\{M_1, M_2,\ldots, M_n\}\subset\mathscr{M}(\cK)$ where $n$ is arbitrary and $M_y=\{M_y(m_y)\}_{m_y}$ for $y=1,2,\ldots,n$. Here, $m_y$ denotes the outcome for the measurement $M_y~\forall y$. The action of the instrument $\mathbf{E}^{j,k_1}:=\{\sum_{a,b}p(a|j)\tilde{\Phi}^{j,b,k_1}_c\circ(\Phi^{a}_b\otimes\mathbbm{I}_{Q^{\prime}})\circ\Phi^{\prime * j,k_1}_{a}\}$ in the first term, on this set of measurements in the Heisenberg picture is given by
    \begin{align}
    \mathbf{E}^{j,k_1}(\cM)=&(\sum_{a,b}p(a|j)\tilde{\Phi}^{j,b,k_1}_c\circ(\Phi^{a}_b\otimes\mathbbm{I}_{Q^{\prime}})\circ\Phi^{\prime * j,k_1})^{\dagger}(\cM)\nonumber\\
    =&\sum_{a,b}p(a|j)(\Phi^{\prime * j,k_1})^{\dagger}\circ(\Phi^{a}_b\otimes\mathbbm{I}_{Q^{\prime}})^{\dagger}\circ(\tilde{\Phi}^{j,b,k_1}_c)^{\dagger}(\cM),\nonumber\\
    =&(\Phi^{\prime * j,k_1})^{\dagger}\circ(\Lambda^{j,k_1}_c)^{\dagger}(\cM),\label{Eq:Inc_Br_Inst_First}
    \end{align}
    where $(\Lambda^{j,k_1}_c)^{\dagger}=\sum_{a,b}p(a|j)(\Phi^{a}_b\otimes\mathbbm{I}_{Q^{\prime}})^{\dagger}\circ(\tilde{\Phi}^{j,b,k_1}_c)^{\dagger}$. We denote the transformed set of measurements as $\tilde{\cM}^{j,k_1}=\{\tilde{M}_1^{j,k_1},\tilde{M}_2^{j,k_1},\ldots,\tilde{M}_n^{j,k_1}\}$ with $\tilde{M}_y^{j,k_1}=\{\tilde{M}_y^{j,k_1}(c_y,m_y
    )\}$ where $\tilde{M}_y^{j,k_1}(c_y,m_y
    ):=(\Lambda^{j,k_1}_c)^{\dagger}(M_y(m_y))~\forall~a,j,k_1$ with $c_y=c$ and $y=1,2,\ldots,n$, is also a set of sets of measurements. Now as $\Phi^{'*j,k_1}$ is an incompatibility-breaking channel for all $j,k_1$, from Eq. \eqref{Eq:Inc_Br_Inst_First} the set of transformed measurements $\mathbf{E}^{j,k_1}(\cM)$ is compatible for all $j,k_1$. Thus the instrument
    $\mathbf{E}^{j,k_1}$ 
    is an incompatibility-breaking instrument for all $j,k_1$.
    
    Again, we know that $\tilde{\cI}^{*j,k_1}=\{\tilde{\mathbf{I}}^{*j,b,k_1}=\{\tilde{\Phi}^{*j,b,k_1}_c\}\}$ is also a set of sets of incompatibility-breaking instruments .Consider a set of arbitrary $\cM=\{M_1, M_2,\ldots, M_n\}$ where $n$ is arbitrary and $M_y=\{M_y(m_y)\}_{m_y}$ for $y=1,2,\ldots,n$. Here $m_y$ denotes the outcome for the measurement $M_y~\forall y$. The action of the instrument in second term, denoted as $\mathbf{F}^{j,k_2}:=\{\sum_{a,b}\tilde{\Phi}^{*j,b,k_2}_c\circ(\Phi^{a}_b\otimes\mathbbm{I}_{Q})\circ\Phi^{\prime  j,k_2}_{a}\}$, on this set of measurements in the Heisenberg picture is given by
    \begin{align}
    \mathbf{F}^{j,k_2}(\cM)=&(\sum_{a,b}\tilde{\Phi}^{*j,b,k_2}_c\circ(\Phi^{a}_b\otimes\mathbbm{I}_{Q})\circ\Phi^{\prime j,k_2}_{a})^{\dagger}(\cM)\nonumber\\
    =&\sum_{a,b}(\Phi^{\prime j,k_2}_{a})^{\dagger}\circ(\Phi^{a}_b\otimes\mathbbm{I}_{Q})^{\dagger}\circ(\tilde{\Phi}^{*j,b,k_2}_c)^{\dagger}(\cM),\nonumber\\
    =&\sum_{b}(\Lambda^{j,k_2}_b)^{\dagger}\circ(\tilde{\Phi}^{*j,b,k_2}_c)^{\dagger}(\cM),
    \end{align}
    where $(\Lambda^{j,k_2}_b)^{\dagger}=\sum_a (\Phi^{\prime j,k_2}_{a})^{\dagger}\circ(\Phi^{a}_b\otimes\mathbbm{I}_{Q})^{\dagger}$. We denote the transformed set of measurements as $\tilde{\cM}^{j,k_2}=\{\tilde{M}^{j,k_2}_1,\tilde{M}^{j,k_2}_2,\ldots,\tilde{M}^{j,k_2}_n\}$ with $\tilde{M}^{j,k_2}_y(c_y,m_y)=\sum_b(\Lambda^{j,k_2}_{b})^{\dagger}(\tilde{M}_y^{j,b,k_2}(c_y,m_y))$ where $\tilde{\cM}^{j,b,k_2}:=\tilde{\mathbf{I}}^{j,b,k_2}[\cM]=\{\tilde{M}_y^{j,b,k_2}:=\{\tilde{M}_y^{j,b,k_2}(c_y,m_y
    )\}\}$ with $\tilde{M}_y^{j,b,k_2}(c_y,m_y
    ):=(\tilde{\Phi}^{*j,b,k_2}_c)^{\dagger}(M_y(m_y))~\forall~b,j,k_2$ and $y=1,2,\ldots,n$ is a set  of compatible measurements as $\tilde{\mathbf{I}}^{j,b,k_2}$ is an incompatibility-breaking instrument for all $j,b,k_2$. If we define $g_y:=(c_y,m_y)$,  mathematically, it we can write
    \begin{align}
        \tilde{M}^{j,b,k_2}_y(c_y,m_y)=\sum_{\{g_1,g_2,\ldots,g_n\}\setminus g_y}\tilde{G}^{j,b,k_2}(g_1,g_2,\ldots,g_n).
    \end{align}
    Here $\tilde{G}^{j,b,k_2}:=\{\tilde{G}^{j,b,k_2}(c,m_1,m_2,\ldots,m_n)\}$ is the joint measurement for the set of measurements $\tilde{\cM}^{j,b,k_2}$. We then define a matrix
    \begin{align}
        G^{j,k_2}(g_1,g_2,\ldots,g_n)=\sum_b(\Lambda^{j,k_2}_b)^{\dagger}(\tilde{G}^{j,b,k_2}(g_1,g_2,\ldots,g_n)).
    \end{align}
    We define $G^{j,k_2}=\{G^{j,k_2}(g_1,g_2,\ldots,g_n)\}$. From the properties of the elements of the set $\tilde{G}^{j,b,k_2}$ it is easy to verify that $G^{j,k_2}(g_1,g_2,\ldots,g_n)\geq0$ and $\sum_{g_1,g_2,\ldots,g_n}G^{j,k_2}(g_1,g_2,\ldots,g_n)=\mathbbm{1}_{\overline{\cH}}~\forall~j,k_2$. Thus, $G^{j,k_2}$ forms a valid measurement for all $j,k_2$.

    Then it can be easily proved that
    \begin{align}
        \tilde{M}^{j,k_2}_y(c_y,m_y)=&\sum_b(\Lambda^{j}_{b})^{\dagger}(\tilde{M}_y^{j,b,k_2}(c_y,m_y))\nonumber\\
        =&\sum_b(\Lambda^{j,k_2}_b)^{\dagger}(\tilde{G}^{j,b,k_2}(g_1,g_2,\ldots,g_n))\nonumber\\
        =&\sum_{\{g_1,g_2,\ldots,g_n\}\setminus g_y}G^{j,k_2}(g_1,g_2,\ldots,g_n).
    \end{align}
    Hence $G^{j,k_2}$ is the joint measurement for the set of measurements $\tilde{\cM}^{j,k_2}~\forall~j,k_2$. Thus the instrument
    $\mathbf{F}^{j,k_2}$ 
    is also an incompatibility-breaking instrument for all $j,k_2$.

    Similarly the action of the instrument $\mathbf{H}^{j,k_3}:=\{\sum_{a,b}\tilde{\Gamma}^{ j,a,k_3}_c\circ\Phi^{a}_b\circ\tilde{\Delta}^{ j,k_3}_a$, in the third term, on the set of measurements $\cM$ in the Heisenberg picture is given by
    \begin{align}
    \mathbf{H}^{j,k_3}(\cM)=&(\sum_{a,b}\tilde{\Gamma}^{ j,a,k_3}_c\circ\Phi^{a}_b\circ\tilde{\Delta}^{ j,k_3}_a)^{\dagger}(\cM)\nonumber\\
    =&\sum_{a,b}(\tilde{\Delta}^{ j,k_3}_a)^{\dagger}\circ(\Phi^{a}_b)^{\dagger}\circ(\tilde{\Gamma}^{ j,a,k_3}_c)^{\dagger}(\cM),\nonumber\\
    =&\sum_{a}(\tilde{\Delta}^{ j,k_3}_a)^{\dagger}\circ(\Phi^{a})^{\dagger}\circ(\tilde{\Gamma}^{j,a,k_3}_c)^{\dagger}(\cM).
    \end{align}
Note that $\{(\Phi^{a})^{\dagger}\circ(\tilde{\Gamma}^{j,a,k_3}_c)^{\dagger}\}$ is a valid quantum instrument for all $j,a,k_3$. We denote the transformed set of measurements as $\tilde{\cM}^{j,k_3}=\{\tilde{M}^{j,k_3}_1,\tilde{M}^{j,k_3}_2,\ldots,\tilde{M}^{j,k_3}_n\}$ with $\tilde{M}^{j,k_3}_y(c_y,m_y)=\sum_a(\tilde{\Delta}^{ j,k_3}_a)^{\dagger}(\tilde{M}_y^{j,a,k_3}(c_y,m_y))$ where $\tilde{\cM}^{j,a,k_3}=\{\tilde{M}_y^{j,a,k_3}:=\{\tilde{M}_y^{j,a,k_3}(c_y,m_y
    )\}\}$ with $\tilde{M}_y^{j,a,k_3}(c_y,m_y
    ):=(\Phi^a)^{\dagger}\circ(\tilde{\Gamma}^{j,a,k_3}_c)^{\dagger}(M_y(m_y))~\forall~a,j,k_3$, and $y=1,2,\ldots,n$ is a set  of measurements. As $\mathbf{I}^a=\{\Phi^a_b\}$ is an incompatibility-breaking instrument for all $a$, the channel $\sum_b\Phi^a_b=\Phi^a$ is also incompatibility-breaking. Therefore, the set $\tilde{\cM}^{j,a,k_3}$ is compatible for all $j,a,k_3$ with the joint measurement $\tilde{G}^{j,b,k_2}(g_1,g_2,\ldots,g_n)$ where $g_y:=(c_y,m_y)$.

    Proceeding in the same way as before, it can be shown that
     \begin{align}
        \tilde{M}^{j,k_3}_y(c_y,m_y)=&\sum_a(\tilde{\Delta}^{ j,k_3}_a)^{\dagger}(\tilde{M}_y^{j,a,k_3}(c_y,m_y))\nonumber\\
        =&\sum_b(\Lambda^{j,k_3}_b)^{\dagger}(\tilde{G}^{j,b,k_3}(g_1,g_2,\ldots,g_n))\nonumber\\
        =&\sum_{\{g_1,g_2,\ldots,g_n\}\setminus g_y}G^{j,k_3}(g_1,g_2,\ldots,g_n).
    \end{align}
    where $G^{j,k_3}=\sum_b(\Lambda^{j,k_3}_b)^{\dagger}(\tilde{G}^{j,b,k_3}(g_1,g_2,\ldots,g_n))$ is the joint measurement for the set of measurements $\tilde{\cM}^{j,k_3}~\forall~j,k_3$. Thus the instrument
    $\mathbf{H}^{j,k_3}$ is also an incompatibility-breaking instrument for all $j,k_3$.

    Now we know that incompatibility-breaking instruments form a convex set. As $\{\overline{\Phi}^j_c\}$ in Eq. \eqref{Eq:free_operation_incom_break} is a convex mixture, of $\mathbf{E}^{j,k_1}$, $\mathbf{F}^{j,k_2}$, and $\mathbf{H}^{j,k_3}$ it is also an incompatibility-breaking instrument. Hence, a transformation of the form given in Eq. \eqref{Eq:free_operation_incom_break} transforms a set of incompatibility-breaking instruments to another set of incompatibility-breaking instruments, and because of this reason, it can be considered as a free transformation of the resource theory of incompatibility-preservability.

    Next thing we will show is that, there exists a transformation of the form in Eq. \eqref{Eq:free_operation_incom_break} which transforms a given arbitrary set of incompatibility-breaking instruments to another given arbitrary set of incompatibility-breaking instruments. We first consider an arbitrary set of incompatibility-breaking instruments $\cI=\{\mathbf{I}^a=\{\Phi^{a}_b\}\in\mathscr{I}_{IB}(\cH,\cK)\}$. Our goal is to show that it can be transformed into another given arbitrary set of incompatibility-breaking instruments $\overline{\cJ}=\{\overline{\mathbf{J}}^j=\{\overline{\Phi}^j_c\}\in\mathscr{I}_{IB}(\overline{\cH},\overline{\cK})\}$. We again define

    \begin{align}
        \Phi^{\prime j,1}_{a}&=p(a)\Gamma_0\nonumber\\
        \tilde\Phi^{*j,b,1}_c&=\overline{\Phi}^j_c\circ(\tr_{\cK}\otimes\mathbbm{I}_{Q}),\qquad~\forall~j,b ,
    \end{align}
    where $\sum_a p(a)=1$, $\Gamma_0:\cL(\overline{\cH})\rightarrow\cL(\cH\otimes Q)$ with $\overline{\cH}=Q$ such that for all $\sigma\in\cL(\overline{\cH}),~ \Gamma_0(\sigma)= \ket{0}\bra{0}\otimes\sigma$  and clearly, $\mathbbm{I}_{Q}=\mathbbm{I}_{\overline{\cH}}$. Then, choosing $q(k_2)=\delta_{k_2,1}$ with $q(k_1),q(k_3)=0 ~\forall~k_1,k_3$, it follows that
    \begin{align}
       \sum_{k_1}q(k_1)\sum_{a,b}&\tilde{\Phi}^{j,b,k_1}_c\circ(\Phi^{a}_b\otimes\mathbbm{I}_{Q^{\prime}})\circ\Phi^{\prime * j,k_1}_{a}\nonumber\\
       +&\sum_{k_2}q(k_2)\sum_{a,b}\tilde{\Phi}^{*j,b,k_2}_c\circ(\Phi^{a}_{b}\otimes\mathbbm{I}_{Q})\circ\Phi^{\prime j,k_2}_{a}\nonumber\\
       +&\sum_{k_3}q(k_3)\sum_{a,b}\tilde{\Gamma}^{ j,b,k_3}_c\circ\Phi^{a}_b\circ\tilde{\Delta}^{ j,k_3}_a= \overline{\Phi}^j_c.
    \end{align}
Thus, any given set of arbitrary incompatibility-breaking instruments $\cI=\{\mathbf{I}^a=\{\Phi^{a}I_b\}\in\mathscr{I}_{IB}(\cH,\cK)\}$ can be transformed to another given arbitrary set of incompatibility-breaking instruments $\overline{\cJ}=\{\overline{\mathbf{J}}^j=\{\overline{\Phi}^j_c\}\in\mathscr{I}_{IB}(\overline{\cH},\overline{\cK})\}$ using the free transformations given in Eq. \eqref{Eq:free_operation_incom_break}.
\end{proof}

\begin{theorem}
    $\widehat{\cD}$ is monotonically non-increasing under the free transformations of the resource theory of incompatibility preservability.\label{Th:Dist_mono_IP}
\end{theorem}

Appendix \ref{App:Dist_mono_IP} contains the detailed proof of this theorem. As a result of the above theorem, using Proposition \ref{Propsi:res_meas_prop} and \ref{Propsi:res_meas_ext_prop}, we can also conclude that the distance-based resource measures in Eqs. \eqref{Eq:Def_res_meas} and \eqref{Eq:Def_res_ext_meas} are valid resource measures for the resource theory of information preservability. These are denoted as $\mathbbm{R}_{MIP}$ and $\overline{\mathbbm{R}}_{MIP}$, respectively.

Alternatively, in the resource theory of strong incompatibility preservability, the sets of weak incompatibility-breaking instruments are the free objects that destroy the incompatibility of any sets of measurements, at least when the classical outcomes of the instrument are unknown.  We construct the free transformations as follows.

Let us have a set of sets of arbitrary instruments $\{\cJ^{\prime k_1}=\{\mathbf{J}^{\prime j,k_1}=\{\Phi^{\prime j,k_1}_{a}\}\in\mathscr{I}(\overline{\cH},\cH\otimes Q)\}_j\}_{k_1}$ and a set of sets of sets of weak-incomaptibility breaking instruments $\{\{\tilde{\cI}^{b,k_1}=\{\tilde{\mathbf{I}}^{j,b,k_1}=\{\tilde{\Phi}^{j,b,k_1}_c\}\in\mathscr{I}_{WIB}(\cK\otimes Q,\overline{\cK})\}_b\}_j\}_{k_1}$. Also consider set of sets of sets of instruments  $\{\{\tilde{\cX}^{j,k_2}=\{\tilde{\mathbf{X}}^{j,a,k_2}=\{\tilde{\Gamma}_c^{j,a,k_2}\}\in\mathscr{I}(\cK,\overline{\cK})\}_a\}_j\}_{k_2}$ and another set of sets of quantum instruments $\{\tilde{\cO}^{k_2}=\{\tilde{\mathbf{O}}^{j,k_2}=\{\tilde{\Delta}^{j,k_2}\}\in\mathscr{I}(\overline{\cH},\cH)\}_j\}_{k_2}$. Then the following theorem holds:
\begin{theorem}
     Let $\cI=\{\mathbf{I}^a=\{\Phi^{a}_b\}\in\mathscr{I}_{WIB}(\cH,\cK)\}$ be a set of weak incompatibility-breaking instruments and $\overline{\cJ}=\{\overline{\mathbf{J}}^j=\{\overline{\Phi}^j_c\}\in\mathscr{I}(\overline{\cH},\overline{\cK})\}\}$ be a set of instruments such that
    \begin{align}
       \overline{\Phi}^j_c=\sum_{k_1}q(k_1)\sum_{a,b}&\tilde{\Phi}^{j,b,k_1}_c\circ(\Phi^{a}_b\otimes\mathbbm{I}_{Q})\circ\Phi^{\prime  j,k_1}_{a}\nonumber\\
       &+\sum_{k_2}q(k_2)\sum_{a,b}\tilde{\Gamma}_c^{j,a,k_2}\circ\Phi^a_b\circ\tilde{\Delta}_a^{j,k_2}.\label{Eq:free_operation_weak_incom_break}
    \end{align}
     with $\sum_{k_1}q(k_1)+\sum_{k_2}q(k_2)=1$. Then
    \begin{enumerate}
        \item $\overline{\cJ}$ is also a set of weak incompatibility-breaking instruments. In other words, the transformation of the form given in Eq. \eqref{Eq:free_operation_weak_incom_break} can be considered as a free transformation of the resource theory of strong incompatibility preservability.
        \item A given arbitrary set of weak incompatibility-breaking instruments can be transformed to another given arbitrary set of weak incompatibility-breaking instruments through a transformation of the form given in Eq. \eqref{Eq:free_operation_weak_incom_break}.
    \end{enumerate}
    \label{Th:free_op_weak_incom_break}
\end{theorem}
\begin{proof}
To prove $\overline{\cJ}$ to be a set of weak incompatibility-breaking instruments, we have to prove that $\sum_c\overline{\Phi}^j_c~ \forall~j$ is an incompatibility-breaking channel for all $j$. We know that $\sum_c\tilde{\Phi}^{j,b,k_1}_c:=\tilde{\Phi}^{j,b,k_1}$ is an incompatibility-breaking channel $\forall~j,b,k_1$. Consider a set of arbitrary measurements $\cM=\{M_1, M_2,\ldots, M_n\}$ where $n$ is arbitrary and $M_y=\{M_y(m_y)\}_{m_y}$ for $y=1,2,\ldots,n$. Here $m_y$ denotes the outcome for the measurement $M_y~\forall y$. The action of the first term in Eq. \eqref{Eq:free_operation_weak_incom_break}, $\sum_{a,b,c}\tilde{\Phi}^{j,b,k_1}_c\circ(\Phi^{a}_b\otimes\mathbbm{I}_{Q})\circ\Phi^{\prime * j,k_1}_{a})^{\dagger}$, on this set of measurements in the Heisenberg picture is given by
    \begin{align}
    (\sum_{a,b,c}\tilde{\Phi}^{j,b,k_1}_c\circ&(\Phi^{a}_b\otimes\mathbbm{I}_{Q})\circ\Phi^{\prime * j,k_1}_{a})^{\dagger}(\cM)\nonumber\\
    =&\sum_{a,b,c}(\Phi^{\prime * j,k_1}_{a})^{\dagger}\circ(\Phi^{a}_b\otimes\mathbbm{I}_{Q})^{\dagger}\circ(\tilde{\Phi}^{j,b,k_1}_c)^{\dagger}(\cM),\nonumber\\
    =&\sum_{b}(\Lambda^{j,k_1}_b)^{\dagger}\circ(\tilde{\Phi}^{j,b,k_1})^{\dagger}(\cM),
    \end{align}
    where $(\Lambda^{j,k_1}_b)^{\dagger}=\sum_{a}(\Phi^{\prime * j,k_1}_{a})^{\dagger}\circ(\Phi^{a}_b\otimes\mathbbm{I}_{Q})^{\dagger}$. We denote the transformed set of measurements as $\tilde{\cM}^{j,k_1}=\{\tilde{M}^{j,k_1}_1,\tilde{M}^{j,k_1}_2,\ldots,\tilde{M}^{j,k_1}_n\}$ where $\tilde{M}^{j,k_1}_y(m_y):=\sum_{b}(\Lambda^{j,k_1}_b)^{\dagger}\circ(\tilde{\Phi}^{j,b,k_1})^{\dagger}(M_y(m_y))$. As the channel $\tilde{\Phi}^{j,b,k_1}$ is incompatibility-breaking we can write
    \begin{align}
        (\tilde{\Phi}^{j,b,k_1})^{\dagger}(M_y(m_y))=\sum_{\{m_1,m_2,\ldots,m_n\}\setminus m_y}\tilde{G}^{j,b,k_1}(m_1,m_2,\ldots,m_n),
    \end{align} 
    $\forall~j,b,k_1$ where $\tilde{G}^{j,b,k_1}:=\{\tilde{G}^{j,b,k_1}(m_1,m_2,\dots,m_n)\}$ is the joint measurement of the set of compatible measurements $\{(\tilde{\Phi}^{j,b,k_1})^{\dagger}(M_y)\}$. We construct a matrix
    \begin{align}
    G^{j,k_1}(m_1,m_2,\ldots,m_n)=\sum_b(\Lambda^{j,k_1}_b)^{\dagger}(\tilde{G}^{j,b,k_1}(m_1,m_2,\ldots,m_n)),
    \end{align}
    $\forall j,k_1$. We define $G^{j,k_1}:=\{G^{j,k_1}(m_1,m_2,\ldots,m_n)\}$. From the properties of the elements of the set $\tilde{G}^{j,k_1}$, it can be verified that $G^{j,k_1}(m_1,m_2,\ldots,m_n)\geq0$ and $\sum_{m_1,m_2,\ldots,m_n}G^{j,k_1}(m_1,m_2,\ldots,m_n)=\mathbbm{1}_{\overline{\cH}}$. Thus, $G^{j,k_1}$ is a valid measurement.

    It is then easy to show that
    \begin{align}
        \tilde{M}^{j,k_1}_y(m_y)=\sum_{\{m_1,m_2,\ldots,m_n\}\setminus m_y}G^{j,k_1}(m_1,m_2,\ldots,m_n)~\forall~j.
    \end{align}
    Hence $\tilde{\cM}^{j,k_1}$ is a set of compatible measurements for all $j,k_1$.

    Next, for the second term we have

    \begin{align}
    (\sum_{a,b,c}\tilde{\Gamma}^{ j,a,k_2}_c&\circ\Phi^{a}_b\circ\tilde{\Delta}^{ j,k_2}_a)^{\dagger}(\cM)\nonumber\\
    =&\sum_{a,b}(\tilde{\Delta}^{ j,k_2}_a)^{\dagger}\circ(\Phi^{a}_b)^{\dagger}\circ(\tilde{\Gamma}^{ j,a,k_2})^{\dagger}(\cM),\nonumber\\
    =&\sum_{a}(\tilde{\Delta}^{ j,k_2}_a)^{\dagger}\circ(\Phi^{a})^{\dagger}\circ(\tilde{\Gamma}^{j,a,k_2})^{\dagger}(\cM).
    \end{align}
Note that $\tilde{\Gamma}^{j,a,k_2}\circ\Phi^a$ is an incompatibility-breaking channel for all $j,a,k_2$ as $\Phi^a$ is an incompatibility breaking channel for all $a$. We denote the transformed set of measurements as $\tilde{\cM}^{j,a,k_2}=\{\tilde{M}_y^{j,a,k_2}:=\{\tilde{M}_y^{j,a,k_2}(m_y
    )\}\}$ with $\tilde{M}_y^{j,a,k_2}(m_y
    ):=(\Phi^a)^{\dagger}\circ(\tilde{\Gamma}^{j,a,k_2})^{\dagger}(M_y(m_y))~\forall~a,j,k_2$ and $y=1,2,\ldots,n$. This set is compatible for all $j,a,k_2$. Then we have
     \begin{align}
        \tilde{M}^{j,a,k_2}_y(m_y)=&(\Phi^a)^{\dagger}\circ(\tilde{\Gamma}^{j,a,k_2})^{\dagger}(M_y(m_y))\nonumber\\
        =&\sum_{\{m_1,m_2,\ldots,m_n\}\setminus m_y}\tilde{G}^{j,a,k_3}(m_1,m_2,\ldots,m_n).
    \end{align}
    where $\tilde{G}^{j,a,k_3}$ is the joint measurement for the set of measurements $\tilde{\cM}^{j,a,k_2}~\forall~j,a,k_2$. 

    Proceeding in the same way as previously, the set of measurements $\tilde{\cM}^{j,k_2}=\{\tilde{M}^{j,k_2}_1,\tilde{M}^{j,k_2}_2,\ldots,\tilde{M}^{j,k_2}_n\}$ with $\tilde{M}^{j,k_2}_y(m_y)=\sum_a(\tilde{\Delta}^{ j,k_2}_a)^{\dagger}(\tilde{M}_y^{j,a,k_2}(m_y))$ is compatible for all $j,k_2$ such that
    \begin{align}
        \tilde{M}^{j,k_2}_y(m_y)=&\sum_a(\tilde{\Delta}^{ j,k_2}_a)^{\dagger}\circ(\Phi^a)^{\dagger}\circ(\tilde{\Gamma}^{j,a,k_2})^{\dagger}(M_y(m_y))\nonumber\\
        =&\sum_{\{m_1,m_2,\ldots,m_n\}\setminus m_y}\sum_a(\tilde{\Delta}^{ j,k_2}_a)^{\dagger}\circ\tilde{G}^{j,a,k_3}(m_1,m_2,\ldots,m_n).
    \end{align}
    Thus, the channel
    $(\sum_{a,b,c}\tilde{\Gamma}^{ j,a,k_2}_c\circ\Phi^{a}_b\circ\tilde{\Delta}^{ j,k_2}_a)^{\dagger}$ is an incompatibility-breaking channel for all $j,k_2$.   
    
   As we know that the convex combination of incompatibility-breaking channels is also an incompatibility-breaking channel, $\sum_c\overline{\Phi}^j_c$ is an incompatibility-breaking channel for all $j$. Thus, the transformations of the form given in Eq. \eqref{Eq:free_operation_weak_incom_break} can be considered as the free transformations for the resource theory of strong incompatibility preservability. 

    The next thing we have to show is that for two given arbitrary sets of weak incompatibility-breaking instruments, there exists a transformation of the form given in Eq. \eqref{Eq:free_operation_weak_incom_break} that transforms one set of the given pair to the other set of the same given pair. To show this, we first consider two arbitrary weak incompatibility-breaking instruments $\cI=\{\mathbf{I}^a=\{\Phi^{a}_b\}\in\mathscr{I}_{WIB}(\cH,\cK)\}$ and $\overline{\cJ}=\{\overline{\mathbf{J}}^j=\{\overline{\Phi}^j_c\}\in\mathscr{I}_{WIB}(\overline{\cH},\overline{\cK})\}$. Next, we define
      \begin{align}
        \Phi^{\prime j,1}_{a}&=p(a)\Gamma_0\nonumber\\
        \tilde\Phi^{j,b,1}_c&=\overline{\Phi}^j_c\circ(\tr_{\cK}\otimes\mathbbm{I}_{Q}),\qquad~\forall~j,b ,
    \end{align}
    where $\sum_a p(a)=1$, $\Gamma_0:\cL(\overline{\cH})\rightarrow\cL(\cH\otimes Q)$ with $\overline{\cH}=Q$ such that for all $\sigma\in\cL(\overline{\cH}),~ \Gamma_0(\sigma)= \ket{0}\bra{0}\otimes\sigma$  and clearly, $\mathbbm{I}_{Q}=\mathbbm{I}_{\overline{\cH}}$. Then, by choosing $q(k_1)=\delta_{k_1,1}$ with $q(k_2)=0 ~\forall~k_2$, it can be easily shown that
    \begin{align}
         \sum_{k_1}q(k_1)\sum_{a,b}&\tilde{\Phi}^{j,b,k_1}_c\circ(\Phi^{a}_b\otimes\mathbbm{I}_{Q})\circ\Phi^{\prime  j,k_1}_{a}\nonumber\\
       &+\sum_{k_2}q(k_2)\sum_{a,b}\tilde{\Gamma}_c^{j,a,k_2}\circ\Phi^a_b\circ\tilde{\Delta}_a^{j,k_2}=\overline{\Phi}^j_c.
    \end{align}
    Thus, the transformation of the form given in Eq. \eqref{Eq:free_operation_weak_incom_break} transforms a given arbitrary set of weak incompatibility-breaking instruments $\cI=\{\mathbf{I}^a=\{\Phi^{a}_b\}\in\mathscr{I}_{WIB}(\cH,\cK)\}$ to another given arbitrary set of weak incompatibility-breaking instruments $\overline{\cJ}=\{\overline{\mathbf{J}}^j=\{\overline{\Phi}^j_c\}\in\mathscr{I}_{WIB}(\overline{\cH},\overline{\cK})\}$.
\end{proof}

\begin{remark}
    \rm{Note that in the free transformations of both resource theory of entanglement-preservability and resource theory of incompatibility-preservability (Eqs. \eqref{Eq:free_operation_ent_break} and \eqref{Eq:free_operation_incom_break} respectively), there are three terms, while in the case of the resource theory of strong entanglement-preservability and the resource theory of strong incompatibility-preservability (Eqs. \eqref{Eq:free_operation_weak_ent_break} and \eqref{Eq:free_operation_weak_incom_break} respectively), the free transformation of each has just two terms. The reason for this is the fact that post-processing may convert a weak entanglement-breaking instrument to an instrument that is not weak entanglement-breaking and a weak incompatibility-breaking instrument to an instrument that is not weak incompatibility-breaking, in general. This can be easily shown through an example. Consider three qubit unitary channels (denoted as $\Lambda_2, \Lambda_3,\Lambda_4$) corresponding to unitary matrices $\sigma_x$, $\sigma_y$, and $\sigma_z$ respectively and the same four outcome instrument $\mathbf{I}$ given in the Example \ref{Examp:eb_strict_in_web}. Now let us denote the qubit identity channel $\mathbbm{I}_{\cH^{\mathbf{Q}}}$ as $\Lambda_1$ Now suppose the instrument $\tilde{\mathbf{I}}=\{\tilde{\Phi}_a\}$ is \emph{post-processed} from the instrument $\mathbf{I}$ such that 
    
    \begin{align}
        \tilde{\Phi}_a=\Lambda_a\circ\Phi_a~\forall a\in\{1,...,4\}.
    \end{align}
    Then clearly, $\tilde{\Phi}:=\sum^4_{a=1}\tilde{\Phi}_a=\mathbbm{I}_{\cH^{\mathbf{Q}}}$ and therefore, the instrument $\tilde{\mathbf{I}}$ is neither a  weak entanglement-breaking instrument nor a weak incompatibility-breaking instrument. But the instrument $\mathbf{I}$ is both  weak entanglement-breaking  and weak incompatibility-breaking.
    
    }
\end{remark}

\begin{theorem}
    $\widehat{\cD}$ is monotonically non-increasing under the free transformations of the resource theory of strong incompatibility preservability.\label{Th:Dist_mono_SIP}
\end{theorem}

A detailed proof is given in Appendix \ref{App:Dist_mono_SIP}. As a result of the above theorem, using Proposition \ref{Propsi:res_meas_prop} and \ref{Propsi:res_meas_ext_prop}, we can also conclude that the distance-based resource measures in Eqs. \eqref{Eq:Def_res_meas} and \eqref{Eq:Def_res_ext_meas} are valid resource measures for the resource theory of strong incompatibility preservability. These are denoted as $\mathbbm{R}_{SMIP}$ and $\overline{\mathbbm{R}}_{SMIP}$, respectively.

\begin{remark}
    \rm{In Ref. \cite{hsieh_incom_preservability}, the authors have also developed a resource theory of incompatibility preservability in a very nice way, and each free object in their resource theory is \emph{a single incompatibility-breaking channel}. However, in this work, we have considered each \emph{set of incompatibility-breaking instruments} as a free object. Due to this, we have the scope to define the concept of weak incompatibility-breaking instruments here, which is not present in \cite{hsieh_incom_preservability}. In short, our approach is different from the approach of Ref. \cite{hsieh_incom_preservability}. Moreover, we would also like to mention that in Eq. \eqref{Eq:free_operation_weak_incom_break} the second term is similar to the free operation defined in Ref. \cite{hsieh_incom_preservability}.}
\end{remark}
\subsubsection{Traditional compatible instruments and the resource theory of traditional incompatibility}
    The definition of traditional compatibility of quantum instruments is given in Def. \ref{Def:trad_comp}. In Ref. \cite{chitambar_PID}, it is shown that the set of traditional incompatible instruments is a resource for programmable quantum instruments  and the resource theory of traditional incompatiblity has been constructed. In the resource theory of traditional incompatiblity, free objects are all sets of traditionally compatible instruments and free transformations are free PID supermaps that are defined as follows.
    
    Consider a set of quantum instruments $\cJ=\{\mathbf{J}^i=\{\Phi^i_a\}\in\mathscr{I}(\cH\otimes\cK)\}$ such that $\sum_a\Phi^i_a=\Phi$. More precisely, it is a PID (Programmable Instrument Device)\cite{chitambar_PID}.
     \begin{definition}
         A free PID supermap $\cV$ mapping $\cJ$ to another set of traditionally compatible instruments $\tilde\cJ=\{\tilde{\mathbf{J}}^j=\{\tilde\Phi_b^j\}\in\mathscr{I}_{TC}(\cH\otimes\cK)\}$ with $\sum_b\tilde\Phi_b^j=\tilde\Phi:\cL(\cH)\rightarrow\cL(\tilde\cK)$, is defined as
    \begin{align}
        \tilde{\Phi}_b^j=\sum_{\lambda,i,a}p(b\vert i, j,\lambda,a)q(i\vert j,\lambda)\mathbf{K}^\lambda\circ(\Phi_a^i\otimes \mathbbm{I}_{ Q})\circ\mathbf{F}
        \label{Eq:free_operation_trad_comp}
    \end{align}
    where $p(b\vert i, j,\lambda,a) ~\text{and}~ q(i\vert j,\lambda)$ are the conditional probabilities and $\cK=\{\mathbf{K}^\lambda\}$ is the set of quantum instruments. 
     \end{definition}
     A given arbitrary set of traditionally compatible instruments can be transformed to another given arbitrary set of traditionally compatible instruments through a transformation of the form in Eq. \eqref{Eq:free_operation_trad_comp}\cite{chitambar_PID}.

    Next, we prove the monotonicity of the distance measure  $\widehat{\cD}$ under free PID supermaps.

    \begin{theorem}
    $\widehat{\cD}$ is monotonically non-increasing under the free transformations of the resource theory of traditional compatibility.\label{Th:Dist_mono_TC}
\end{theorem}

We refer to Appendix \ref{App:Dist_mono_TC} for a detailed proof. As a result of the above theorem, using Proposition \ref{Propsi:res_meas_prop} and \ref{Propsi:res_meas_ext_prop}, we can also conclude that the distance-based resource measures in Eqs. \eqref{Eq:Def_res_meas} and \eqref{Eq:Def_res_ext_meas} are valid resource measures for the resource theory of traditional incompatibility. These are denoted as $\mathbbm{R}_{TI}$ and $\overline{\mathbbm{R}}_{TI}$, respectively.

\subsubsection{Parallel compatible instruments and the resource theory of parallel incompatibility}

The definition of parallel compatible instrument is given in Def. \ref{Def:para_comp}. The parallel incompatibility of single outcome instruments (or equivalently, incompatibility of channels), which is a special case of parallel compatibility of general quantum instruments, has been shown to be a resource for state discrimination tasks in Ref. \cite{Mori_incomp_chan_state_disc}. Here, we are \emph{motivated} to construct the resource theory of parallel incompatibility. In this resource theory, the free objects are all sets of parallel compatible instruments. We construct the free transformations as follows.

Let $\{\cJ^{k_1}=\{\mathbf{J}^{j,k_1}=\{\Phi^{\prime j,k_1}_{l}\}\in\mathscr{I}_{PC}(\overline{\cH},\cH\otimes Q)\}_j\}_{k_1}$ be a set of sets of parallel compatible instruments and $\{\{\tilde{\cI}^{j,k_1}=\{\tilde{\mathbf{I}}^{j,b,k_1}=\{\tilde{\Phi}^{j,b,k_1}_c\}\in\mathscr{I}(\cK\otimes  Q,\overline{\cK})\}_b\}_j\}_{k_1}$ be set of sets of sets of quantum instruments where $j\in\{1,\ldots,n\}$. Consider another set of sets of sets of instruments  $\{\{\tilde{\cX}^{j,k_2}=\{\tilde{\mathbf{X}}^{j,b,k_2}=\{\tilde{\Gamma}_c^{j,b,k_2}\}\in\mathscr{I}(\cK,\overline{\cK})\}_b\}_j\}_{k_2}$. If $\cI=\{\mathbf{I}^a=\{\Phi^{a}_b\}\in\mathscr{I}_{PC}(\cH,\cK)\}$ be a set of parallel compatible instruments for $a\in\{1,\ldots,m\}$ then we define an extension $\hat{\cI}=\{\hat{\mathbf{I}}^j=\{\hat{\Phi}^{j}_b\}\in\mathscr{I}(\cH,\cK)\}$ such that
\begin{align}
    &\hat{\Phi}^{j}_b=\Phi^{\pi(a)}_b\forall~j\in\{1,\ldots,m\}~\text{with}~a\in\{1,\ldots,m\}\\
    &\hat{\Phi}^{j}_b=\overline{\Phi}^{j}_b~\forall~j\in\{m+1,\ldots,n\}\label{Eq:arbi_chosen_chan}
\end{align}
where $\{\Phi^{j}_b\}\in\mathscr{I}(\cH_j,\cK_j)$ for all $j\in\{1,\ldots,n\}$ with $\cH_j=\cH$ and $\cK_j=\cK$ for all $j\in\{1,2,\ldots,n\}$, $\pi$ is an arbitrary permutation of the set $\{1,\ldots,n\}$. and the instrument $\{\overline{\Phi}^{j}_b\}$ is arbitrarily chosen for all $j\in\{m+1,\ldots,n\}$. If $n\leq m$ instead of extending the set, we may have to reduce the set, and in that case, Eq. \eqref{Eq:arbi_chosen_chan} is not required. Similarly, we also consider another extension of $\cI$ denoted as $\mathcal{I}^*=\{\mathbf{I}^{*l}=\{\Phi^{*l}_b\}\in\mathscr{I}(\cH,\cK)\}^r_{l=1}$ where $r$ is arbitrary, $r\geq m$, $\mathbf{I}^{*l}=\mathbf{I}^{l}~\forall l\in\{1,\ldots,m\}$, and $\mathbf{I}^{*l}$ is arbitrarily chosen for all $l\in\{m+1,\ldots,r\}$. Also, we consider a set of quantum channels $\{\tilde{\Delta}^{j,k_2}\}$ such that
\begin{align}
    \tilde{\Delta}^{j,k_2}:=&\Delta^{k_2}~\forall~j\in\{1,\ldots,m\}\nonumber\\
    \tilde{\Delta}^{j,k_2}:=&\Delta^{j,k_2}~\forall~j\in\{m+1,\ldots,n\},
\end{align}
where $\Delta^{k_2}\in\mathscr{C}(\overline{\cH},\cH)$ and $\Delta^{j,k_2}\in\mathscr{C}(\overline{\cH},\cH_j)$ such that there exists a joint channel $\tilde{\Delta}^{k_2}\in\mathscr
C(\overline{\cH},\cH\otimes(\bigotimes_{j=m+1}^n\cH_j))$ with
\begin{align}
    \tr_{\{\cH_{m+1},\cH_{m+2},\ldots,\cH_{n}\}}\circ\tilde{\Delta}^{k_2}=&\Delta^{k_2}\nonumber\\
    \tr_{\{\cH,\cH_{m+1},\cH_{m+2},\ldots,\cH_{n}\}\setminus \cH_j}\circ\tilde{\Delta}^{k_2}=&\Delta^{j,k_2}
\end{align}
where $j\in\{m+1,\ldots,n\}$. It is worth mentioning that $\Delta^{k_2}$ need \emph{not} to be $m$-copy self-compatible. Then the following results hold.

\begin{theorem}
        If $\cI=\{\mathbf{I}^a=\{\Phi^{a}_b\}\in\mathscr{I}_{PC}(\cH,\cK)\}$ be a set of parallel compatible instruments, $\hat{\cI}$ and $\cI^*$ are its extensions, and $\overline{\cJ}=\{\overline{\mathbf{J}}^j=\{\overline{\Phi}^j_c\}\in\mathscr{I}(\overline{\cH},\overline{\cK})\}$ be a set of instruments such that
    \begin{align}
      \overline{\Phi}^j_c=\sum_{k_1}q(k_1)\sum_{l,b}&\tilde{\Phi}^{j,b,k_1}_c\circ(\Phi^{*l}_b\otimes\mathbbm{I}_{Q})\circ\Phi^{\prime  j,k_1}_{l}\nonumber\\
       &+\sum_{k_2}q(k_2)\sum_{b}\tilde{\Gamma}_c^{j,b,k_2}\circ\hat{\Phi}^{j}_b\circ\tilde{\Delta}^{j,k_2}. \label{Eq:free_operation_para_comp}
    \end{align}
    with $\sum_{k_1}q(k_1)+\sum_{k_2}q(k_2)=1$. Then
   \begin{enumerate}
        \item $\overline{\cJ}$ is also a set of parallel compatible instruments. In other words, an arbitrary transformation of the form given in Eq. \eqref{Eq:free_operation_para_comp} can be considered as a free transformation of the resource theory of parallel compatibility.
        \item A given arbitrary set of parallel compatible instruments can be transformed to another given arbitrary set of parallel compatible instruments through a transformation of the form given in Eq. \eqref{Eq:free_operation_para_comp}.
    \end{enumerate}\label{Th:free_op_para_comp}
\end{theorem}
\begin{proof}
Let us denote $\cH \otimes Q$ as $\mathfrak{R}$. Clearly $\Phi^{\prime  j,k_1}_l:\cL(\overline{\cH})\rightarrow\cL(\mathfrak{R}_j)$ where $\mathfrak{R}_j=\mathfrak{R}$ for all $j,k_1$. 
    If $\cJ^{k_1}$ is a set of parallel compatible instruments, then we know that
    \begin{align}
        \Phi_l^{\prime j,k_1}=\Phi_{l_j}^{\prime j,k_1}=\sum_{\{l_1,\ldots,l_n\}\setminus l_j}\tr_{\{\mathfrak{R}_1,\ldots,\mathfrak{R}_n\}\setminus\mathfrak{R}_j}\circ[\Phi^{\prime k_1}_{(l_1,\ldots,l_n)}],
    \end{align}

    where in the L.H.S. we have written $l$ instead of $l_j$ as there is no chance of confusion and $\tr_{{\{\mathfrak{R}_1,\ldots,\mathfrak{R}_n\}\setminus\mathfrak{R}_j}}$ means taking trace over all the Hilbert spaces except $\mathfrak{R}_j$. Note that $\Phi^{\prime k_1}_{(l_1,l_2,l_3,.....,l_n)}:\cL(\overline{\cH})\rightarrow\cL(\bigotimes_{j=1}^n\mathfrak{R}_j)$ is the joint instrument for the set of instruments $\cJ^{k_1}$.

    Then
    \begin{align}
        \Lambda^{j,k_1}_b:&=\sum_l(\Phi^{*l}_b \otimes \mathbbm{I}_Q)\circ\Phi^{\prime j,k_1}_l\nonumber\\&=\sum_{l_j}(\Phi^{*l_j}_{b_j} \otimes \mathbbm{I}_Q)\circ\Phi^{\prime j,k_1}_{l_j}\nonumber\\
        &=\sum_{\{l_1,\ldots,l_n\}} (\Phi^{*l_j}_{b_j} \otimes \mathbbm{I}_Q)\circ\tr_{{\{\mathfrak{R}_1,\ldots,\mathfrak{R}_n\}\setminus\mathfrak{R}_j}}\circ[\Phi^{\prime k_1}_{(l_1,\ldots,l_n)}]
    \end{align}

    Let us consider an instrument given as:
    \begin{align}
        \Lambda^{k_1}_{(b_1,\ldots,b_n)}=\sum_{\{l_1,\ldots,l_n\}} \bigotimes_{j=1}^{n}(\Phi^{*l_j}_{b_j} \otimes \mathbbm{I}_Q)\circ\Phi^{\prime k_1}_{(l_1,\ldots,l_n)}
    \end{align}
    If we denote $\cK\otimes Q$ as $\mathfrak{R}^{\prime}$ then we have $\Lambda^{k_1}_{(b_1,\ldots,b_n)}:\cL(\overline{\cH})\rightarrow\cL(\bigotimes_{j=1}^n\mathfrak{R}^{\prime}_j)$. We can see that
    \begin{align}
         &\sum_{\{b_1,\ldots,b_n\}\setminus b_j}\tr_{{\{\mathfrak{R}^{\prime}_1,\ldots,\mathfrak{R}^{\prime}_n\}\setminus\mathfrak{R}^{\prime}_j}}\circ[\Lambda^{k_1}_{(b_1,\ldots,b_n)}]\nonumber\\
         &=\sum_{\{b_1,\ldots,b_n\}\setminus b_j}\tr_{{\{\mathfrak{R}^{\prime}_1,\ldots,\mathfrak{R}^{\prime}_n\}\setminus\mathfrak{R}^{\prime}_j}}\circ[\sum_{\{l_1,\ldots,l_n\}} \bigotimes_{j=1}^{n}(\Phi^{*l_j}_{b_j} \otimes \mathbbm{I}_Q)\circ\Phi^{\prime k_1}_{(l_1,\ldots,l_n)}]\nonumber\\
         &=\tr_{{\{\mathfrak{R}^{\prime}_1,\ldots,\mathfrak{R}^{\prime}_n\}\setminus\mathfrak{R}^{\prime}_j}}\circ[\sum_{\{l_1,\ldots,l_n\}}\sum_{\{b_1,\ldots,b_n\}\setminus b_j} \bigotimes_{j=1}^{n}(\Phi^{*l_j}_{b_j} \otimes \mathbbm{I}_Q)\circ\Phi^{\prime k_1}_{(l_1,\ldots,l_n)}]\nonumber\\
         &=\sum_{l_j}(\Phi^{*l_j}_{b_j} \otimes \mathbbm{I}_Q)\circ\sum_{\{l_1,\ldots,l_n\}\setminus l_j}\tr_{{\{\mathfrak{R}_1,\ldots,\mathfrak{R}_n\}\setminus\mathfrak{R}_j}}\circ[\Phi^{\prime k_1}_{(l_1,\ldots,l_n)}]\nonumber\\
         &=\sum_{l_j}(\Phi^{*l_j}_{b_j} \otimes \mathbbm{I}_Q)\circ\Phi^{\prime j,k_1}_{l_j}\nonumber\\
         &=\sum_l(\Phi^{*l}_b \otimes \mathbbm{I}_Q)\circ\Phi^{\prime j,k_1}_l=\Lambda^{j,k_1}_b.
    \end{align}
    Thus, we see that the set of instruments $\{\{\Lambda_b^{j,k_1}=\sum_l(\Phi^{*l}_b \otimes \mathbbm{I}_Q)\circ\Phi^{\prime j,k_1}_l\}\}$ is also parallel compatible.

    All that is left to prove is that the set of instruments $\{\{\sum_{b}\tilde{\Phi}^{j,b,k_1}_c\circ\Lambda_b^{j,k_1}\}_c\}_{j\in\{1,\ldots,n\}}$ is also parallel compatible. We observe that this is just post-processing of the set of parallel compatible instruments $\{\{\Lambda_b^{j,k_1}\}_b\}_{j\in\{1,\ldots,n\}}$ with the set of sets of instruments $\{\tilde{\cI}^{j,k_1}\}$. From Remark \hyperref[Rem:one]{1}, we already know that post-processing is a free transformation for parallel compatibility. Hence, the set of instruments $\{\{\sum_{b}\tilde{\Phi}^{j,b,k_1}_c\circ\Lambda_b^{j,k_1}\}_c\}_{j\in\{1,\ldots,n\}}$ is parallel compatible with the joint instrument 
    \begin{align}
         &\zeta^{k_1}_{(c_1,\ldots,c_n)}=\sum_{\{b_1,\ldots,b_n\}}\bigotimes_{j=1}^{n}(\tilde{\Phi}^{j,b_j,k_1}_{c_j})\circ\Lambda^{k_1}_{(b_1,\ldots,b_n)}\nonumber\\
         &=\sum_{\{b_1,\ldots,b_n\}}\sum_{\{l_1,\ldots,l_n\}} \bigotimes_{j=1}^{n}\Big((\tilde{\Phi}^{j,b_j,k_1}_{c_j})\circ(\Phi^{*l_j}_{b_j} \otimes \mathbbm{I}_Q)\Big)\circ\Phi^{\prime k_1}_{(l_1,\ldots,l_n)},
    \end{align}
    such that
    \begin{align}
        \sum_{\{c_1,c_2,\ldots,c_n\}\setminus c_j}\tr_{\{\overline{\cK}_{1},\overline{\cK}_{2},\ldots,\overline{\cK}_{n}\}\setminus \overline{\cK}_j}\circ\zeta^{k_1}_{(c_1,c_2,\ldots,c_n)}=\sum_{b}\tilde{\Phi}^{j,b,k_1}_c\circ\Lambda_b^{j,k_1}.
    \end{align}
    Let us focus on the second term in Eq. \eqref{Eq:free_operation_para_comp}. If we have a parallel compatible set of instruments $\{\{\Phi^a_b\}_b\}_a$ then $\{\{\Phi^{\pi(a)}_b\}_b\}_a$ is (trivially) parallel compatible and therefore, there exists a joint instrument $\{\Phi_{(b_1,b_2,\ldots,b_m)}\}\in\mathscr{I}(\cH,\bigotimes_{j=1}^m\cK_j)$ such that
\begin{align}
    \sum_{\{b_1,b_2,\ldots,b_m\}\setminus b_j}\tr_{\{\cK_{1},\cK_{2},\ldots,\cK_{m}\}\setminus \cK_j}\circ\Phi_{(b_1,b_2,\ldots,b_m)}=\hat{\Phi}^{j}_{b_j}.
\end{align}
with $j\in\{1,\ldots,m\}$.
Next, consider an instrument $ \{ \Big(\Phi_{(b_1,b_2,\ldots,b_m)}\otimes(\bigotimes_{j=m+1}^n\hat{\Phi}^{j}_{b_j})\Big)\circ\tilde{\Delta}^{k_2}\}$. We consider two cases:\\
\textbf{Case (i):} When we have $j\leq m$, we can write\\
\begin{align}
   \sum_{\{b_1,b_2,\ldots,b_n\}\setminus b_j}&\tr_{\{\cK_1,\cK_2,\ldots,\cK_n\}\setminus \cK_j} \circ\Big(\Phi_{(b_1,b_2,\ldots,b_m)}\otimes(\bigotimes_{j=m+1}^n\hat{\Phi}^{j}_{b_j})\Big)\circ\tilde{\Delta}^{k_2}\nonumber\\
   &=\sum_{\{b_1,b_2,\ldots,b_m\}\setminus b_j}\tr_{\{\cK_1,\cK_2,\ldots,\cK_m\}\setminus \cK_j} \nonumber\\
   &\quad\qquad\circ\Big(\sum_{\{b_{m+1},b_{m+2},\ldots,b_n\}}\tr_{\{\cK_{m+1},\cK_{m+2},\ldots,\cK_n\}}\nonumber\\
   &\quad\quad\qquad\circ\Big(\Phi_{(b_1,b_2,\ldots,b_m)}\otimes(\bigotimes_{j=m+1}^n\hat{\Phi}^{j}_{b_j})\Big)\circ\tilde{\Delta}^{k_2}\Big)\nonumber\\
   &=\sum_{\{b_1,b_2,\ldots,b_m\}\setminus b_j}\tr_{\{\cK_1,\cK_2,\ldots,\cK_m\}\setminus \cK_j} \nonumber\\
   &\quad\qquad\circ\Big(\Phi_{(b_1,b_2,\ldots,b_m)}\circ(\tr_{\{\cH_{m+1},\cH_{m+2},\ldots,\cH_n\}}\circ\tilde{\Delta}^{k_2})\Big)\nonumber\\
   &=\hat{\Phi}^{j}_b\circ\Delta^{k_2}
\end{align}

\textbf{Case (ii):} When we have $j> m$, we can write\\
\begin{align}
   \sum_{\{b_1,b_2,\ldots,b_n\}\setminus b_j}&\tr_{\{\cK_1,\cK_2,\ldots,\cK_n\}\setminus \cK_j} \circ\Big(\Phi_{(b_1,b_2,\ldots,b_m)}\otimes(\bigotimes_{j=m+1}^n\hat{\Phi}^{j}_{b_j})\Big)\circ\tilde{\Delta}^{k_2}\nonumber\\
   &=\sum_{\{b_{m+1},b_{m+2},\ldots,b_n\}\setminus b_j}\tr_{\{\cK_{m+1},\cK_{m+2},\ldots,\cK_n\}\setminus \cK_j}\nonumber\\
   &\quad\qquad\circ\Big(\sum_{\{b_1,b_2,\ldots,b_m\}}\tr_{\{\cK_1,\cK_2,\ldots,\cK_m\}}\nonumber\\
   &\quad\quad\qquad\circ\Big(\Phi_{(b_1,b_2,\ldots,b_m)}\otimes(\bigotimes_{j=m+1}^n\hat{\Phi}^{j}_{b_j})\Big)\circ\tilde{\Delta}^{k_2}\Big)\nonumber
   \end{align}
   \begin{align}
   &=\sum_{\{b_{m+1},b_{m+2},\ldots,b_n\}\setminus b_j}\tr_{\{\cK_{m+1},\cK_{m+2},\ldots,\cK_n\}\setminus \cK_j}\nonumber\\
   &\quad\qquad\circ\Big(\bigotimes_{j=m+1}^n\hat{\Phi}^{j}_{b_j}\Big)\circ(\tr_{\cH}\circ\tilde{\Delta}^{k_2})\nonumber\\
   &=\hat{\Phi}^{j}_b\circ\tr_{\{\cH,\cH_{m+1},\cH_{m+2},\ldots,\cH_{n}\}\setminus \cH_j}\circ\tilde{\Delta}^{k_2}\nonumber\\
   &=\hat{\Phi}^{j}_b\circ\Delta^{j,k_2}
\end{align}
 Thus $\{ \Big(\Phi_{(b_1,b_2,\ldots,b_m)}\otimes(\bigotimes_{j=m+1}^n\hat{\Phi}^{j}_{b_j})\Big)\circ\tilde{\Delta}^{k_2}\}$ is the joint instrument for the set of instruments $\{\{\hat{\Phi}^{j}_b\circ\tilde{\Delta}^{j,k_2}\}_b\}_{j\in\{1,\ldots,n\}}$ and hence they are parallel compatible. Furthermore, notice that the instrument $\{\{\sum_{b}\tilde{\Gamma}_c^{j,b,k_2}\circ\Phi^{j}_b\circ\tilde{\Delta}^{j,k_2}\}_c\}_{j\in\{1,\ldots,n\}}$
is the set of instruments obtained by post-processing the set of parallel compatible instruments $\{\{\hat{\Phi}^{j}_b\circ\tilde{\Delta}^{j,k_2}\}\}$ with the set of sets of instruments $\{\tilde{\cX}^{j,k_2}\}$ and hence, by Remark \hyperref[Rem:one]{1}, it is parallel compatible for all $k_2$.
     
Now, we know that an arbitrary convex combination of sets of parallel compatible instruments belonging to $\mathscr{I}(\overline{\cH},\overline{\cK})$ is a set of parallel compatible instruments in $\mathscr{I}(\overline{\cH},\overline{\cK})$\cite{Mitra_char_quantifying_incomp_inst,Leevi_incomp_inst}.  Thus the set of instruments $\overline{\cJ}$ is parallel compatible. Hence, the given transformation in Eq. (\ref{Eq:free_operation_para_comp}) transforms one free set of instruments to another free set of instruments \textit{i.e.} it can be considered as a free transformation of the resource theory of parallel compatibility.

   Next thing we show is that for two given arbitrary sets of parallel compatible instruments, there exists a transformation of the form given in Eq. \eqref{Eq:free_operation_para_comp} that transforms one set of the given pair to the other set of the same given pair. Consider $\cI=\{\mathbf{I}^a=\{\Phi^{a}_b\}\in \mathscr{I}_{PC}(\cH,\cK)\}$ is a given arbitrary set of parallel compatible quantum instruments. Our goal is to show that it can be transformed into another given arbitrary set of parallel compatible quantum instruments $\overline{\cJ}=\{\overline{\mathbf{J}}^j=\{\overline{\Phi}^j_c\}\in\mathscr{I}_{PC}(\overline{\cH},\overline{\cK})\}$ through a transformation of the form in Eq. \eqref{Eq:free_operation_para_comp}. Without loss of generality, we can assume the outcome set of $\overline{\mathbf{J}}^j$ to be the same for all $j$, as we can always append zero CP maps to the instruments with smaller outcome sets. Let us consider:
    \begin{align}
        \Phi^{\prime  j,1}_{l}&=\Gamma_l\circ\overline{\Phi}^j_l\nonumber\\
        \tilde\Phi^{j,b,1}_c&=(\tr_{\cK}\otimes\mathbbm{I}_{\overline{\cK}}\otimes\Upsilon_c),\qquad~\forall~j,b ,
    \end{align}
    where $\Gamma_l:\cL(\overline{\cK})\rightarrow\cL(\cH\otimes Q)$ with $\overline{\cK}\otimes\cH_{\Omega_{\overline{\mathbf{J}}^j}}=Q$ such that for all $\sigma\in\cL(\overline{\cK}),~ \Gamma_l(\sigma)= \ket{0}\bra{0}\otimes\sigma\otimes\ket{l}\bra{l}$  and clearly, $\mathbbm{I}_{Q}=\mathbbm{I}_{\overline{\cK}}\otimes\mathbbm{I}_{\Omega_{\overline{\mathbf{J}}^j}}$. Similarly $\Upsilon_c(\omega):=\bra{c}\omega\ket{c}$ for all $\omega\in\cH_{\Omega_{\overline{\mathbf{J}}^j}}$ where $\langle c \vert l \rangle=\delta_{c,l}$. Thus,or all $\sigma\in\cL(\overline{\cK})$, we have
    \begin{align}
        \sum_{l,b}\tilde{\Phi}^{j,b,1}_c&\circ(\Phi^{*l}_b\otimes\mathbbm{I}_{Q})\circ\Phi^{\prime  j,1}_{l}(\sigma)\nonumber\\
        =&\sum_{l,b}\tilde{\Phi}^{j,b,1}_c\circ(\Phi^{*l}_b\otimes\mathbbm{I}_{Q})(\ket{0}\bra{0}\otimes\overline{\Phi}^j_l(\sigma)\otimes\ket{l}\bra{l})\nonumber\\
        =&\sum_{l,b}\tilde{\Phi}^{j,b,1}_c(\Phi^{*l}_b(\ket{0}\bra{0})\otimes\overline{\Phi}^j_l(\sigma)\otimes\ket{l}\bra{l})\nonumber\\
        =&\sum_l \delta_{c,l}(\tr_{\cK}[(\Phi^{*l}(\ket{0}\bra{0})]\otimes\overline{\Phi}^j_l(\sigma))\nonumber\\
        =&\overline{\Phi}^j_c(\sigma).
    \end{align}
    Then, by choosing $q(k_1)=\delta_{k_1,1}$ with $q(k_2)=0 ~\forall~k_2$, it can be easily shown that
    \begin{align}
        =\sum_{k_1}q(k_1)\sum_{a,b}&\tilde{\Phi}^{j,b,k_1}_c\circ(\Phi^{a}_b\otimes\mathbbm{I}_{Q})\circ\Phi^{\prime  j,k_1}_{a}\nonumber\\
       &+\sum_{k_2}q(k_2)\sum_{b}\tilde{\Gamma}_c^{j,b,k_2}\circ\Phi^{j}_b\circ\tilde{\Delta}^{j,k_2} =\overline{\Phi}^j_c.
    \end{align}
    Thus, the transformation of the form given in Eq. \eqref{Eq:free_operation_para_comp} transforms a given arbitrary set of weak entanglement-breaking instruments $\cI=\{\mathbf{I}^a=\{\Phi^{a}_b\}\in\mathscr{I}_{PC}(\cH,\cK)\}$ to another given arbitrary set of weak entanglement-breaking instruments $\overline{\cJ}=\{\overline{\mathbf{J}}^j=\{\overline{\Phi}^j_c\}\in\mathscr{I}_{PC}(\overline{\cH},\overline{\cK})\}$
\end{proof}
\begin{theorem}
    $\widehat{\cD}$ is monotonically non-increasing under the free transformations of parallel compatibility.\label{Th:Dist_mon_PC}
\end{theorem}

A detailed proof is provided in Appendix \ref{App:Dist_mon_PC}. As a result of the above theorem, using Proposition \ref{Propsi:res_meas_prop} and \ref{Propsi:res_meas_ext_prop}, we can also conclude that the distance-based resource measures in Eqs. \eqref{Eq:Def_res_meas} and \eqref{Eq:Def_res_ext_meas} are valid resource measures for the resource theory of parallel incompatibility. These are denoted as $\mathbbm{R}_{PI}$ and $\overline{\mathbbm{R}}_{PI}$, respectively.

\subsubsection{Hierarchies among resource measures}
Consider two sets of objects (here, each object is a set of instruments with an arbitray input Hilbert space $\cH$ and an aribtrary output Hilbert space $\cK$) $X(\cH,\cK)$ and $Y(\cH,\cK)$. Consider two generic resource theories $\mathbf{RT}_X$ with the set of free objects $X(\cH,\cK)$ and $\mathbf{RT}_Y$ with the set of free objects $Y(\cH,\cK)$. Let $\mathbbm{R}_X$, and $\overline{\mathbbm{R}}_X$ be the resource measures corresponding to the resource theory $\mathbf{RT}_X$ and $\mathbbm{R}_Y$, and $\overline{\mathbbm{R}}_Y$ be the resource measures corresponding to the resource theory $\mathbf{RT}_Y$. Then the following result holds.

\begin{proposition} 
    Consider a set of quantum instrument $\cI\in\mathscr{I}(\cH_A,\cK_A)$. Then the following statements are true: \label{Proposi:QI_hierarchy}

    \begin{enumerate}
        \item if $X(\cH_A,\cK_A)\subseteq Y(\cH_A,\cK_A)$, then $\mathbbm{R}_X(\cI)\geq \mathbbm{R}_Y(\cI)$, and
        \item if $X(\cH_A\otimes\cH_B,\cK_A)\subseteq Y(\cH_A\otimes\cH_B,\cK_A)~\forall \cH_B$, then $\overline{\mathbbm{R}}_X(\cI)\geq\overline{\mathbbm{R}}_Y(\cI)$.
    \end{enumerate}

\end{proposition}

\begin{proof}
    Let $\cJ^*$ be the set of instruments for which the minimum occurs in Eq. \eqref{Eq:Def_res_meas} (for the set of free objects $X(\cH_A,\cK_A)$ i.e., for the resource theory $\mathbf{RT}_X$). Now, note that $X(\cH_A,\cK_A)\subseteq Y(\cH_A,\cK_A)$. Therefore, as $\cJ^*\in X(\cH_A,\cK_A)$, we have $\cJ^*\in Y(\cH_A,\cK_A)$. Therefore, we have 

    \begin{align}
        \mathbbm{R}_X(\cI)=&\widehat{\cD}(\cI,\cJ^*)\nonumber\\
         \geq&\min_{\tilde{\cJ}\in Y(\cH_A,\cK_A)}~\widehat{\cD}(\cI,\tilde{\cJ})\nonumber\\
        \geq&\mathbbm{R}_Y(\cI),
    \end{align}
    Now, we have to prove the second statement. By assumption, we have $X(\cH_A\otimes\cH_B,\cK_A)\subseteq Y(\cH_A\otimes\cH_B,\cK_A)~\forall \cH_B$. Let for an arbitrary $\cH_B$, $\cJ^*_{(\cH_{AB},\cK_A)}$ be the set of instruments for which the minimum occurs in Eq. \eqref{Eq:K_min} (for the set of free objects $X(\cH_{AB},\cK_A)$ i.e., for the resource theory $\mathbf{RT}_X$) where $\cH_{AB}=\cH_A\otimes\cH_B$. As $\cJ^*_{(\cH_{AB},\cK_A)}\in X(\cH_{AB},\cK_A)$, we have $\cJ^*_{(\cH_{AB},\cK_A)}\in Y(\cH_{AB},\cK_A)$.  Then

     \begin{align}
      \mathbbm{K}_X(\cI,\cH_B)=&\widehat{\cD}(\widehat{\cI}_{(\cH_{AB},\cK_A)},\cJ^*_{(\cH_{AB},\cK_A)})\nonumber\\
      \geq&\min_{\tilde{\cJ}_{(\cH_{AB},\cK_A)}\in Y(\cH_{AB},\cK_A)}~\widehat{\cD}(\cI,\tilde{\cJ}_{(\cH_{AB},\cK_A)})\nonumber\\
        \geq&\mathbbm{K}_Y(\cI,\cH_B)\label{Eq:K_ineq}
  \end{align}
  where we have used the subscript $X$ and $Y$ to indicate the resource theory (i.e.,$\mathbf{RT}_X$ or $\mathbf{RT}_Y$) we are talking about. Note that Eq. \eqref{Eq:K_ineq} is valid for an arbitrary $\cH_B$ and hence is valid for all $\cH_B$. Therefore,

  \begin{align}
      \mathbbm{K}_X(\cI,\cH_B)\geq&\mathbbm{K}_Y(\cI,\cH_B)~\forall \cH_B\nonumber\\
      or,~\inf_{\cH_B}\mathbbm{K}_X(\cI,\cH_B)\geq&\inf_{\cH_B}\mathbbm{K}_Y(\cI,\cH_B)\nonumber\\
      or,~\overline{\mathbbm{R}}_X(\cI)\geq&\overline{\mathbbm{R}}_Y(\cI).
  \end{align}
\end{proof}

Now, suppose $X(\cH,\cK)$ and $Y(\cH,\cK)$ are the collection of all finite subsets of $\overline{X}(\cH,\cK)$ and $\overline{Y}(\cH,\cK)$ respectively for arbitrary $\cH$ and $\cK$. Then $\overline{X}(\cH,\cK)\subseteq \overline{Y}(\cH,\cK)$ implies $X(\cH,\cK)\subseteq Y(\cH,\cK)$. Therefore, from Proposition \ref{Proposi:QI_hierarchy}, we have the following corollary.

\begin{corollary}
    Suppose $X(\cH,\cK)$ and $Y(\cH,\cK)$ are the collection of all finite subsets of $\overline{X}(\cH,\cK)$ and $\overline{Y}(\cH,\cK)$ respectively for arbitrary $\cH$ and $\cK$ and consider a set of quantum instrument $\cI\in\mathscr{I}(\cH_A,\cK_A)$. Then the following statements are true: \label{Coro:QI_hierarchy}

    \begin{enumerate}
        \item if $\overline{X}(\cH_A,\cK_A)\subseteq \overline{Y}(\cH_A,\cK_A)$, then $\mathbbm{R}_X(\cI)\geq \mathbbm{R}_Y(\cI)$, and
        \item if $\overline{X}(\cH_A\otimes\cH_B,\cK_A)\subseteq \overline{Y}(\cH_A\otimes\cH_B,\cK_A)~\forall \cH_B$, then $\overline{\mathbbm{R}}_X(\cI)\geq\overline{\mathbbm{R}}_Y(\cI)$.
    \end{enumerate}
\end{corollary}

Now, note that from Fig. \ref{fig_venn}, we have we have $\mathscr{I}_{TP}(\cH, \cK)\subseteq\mathscr{I}_{EB}(\cH,\cK)\subseteq \mathscr{I}_{WEB}(\cH,\cK)\subseteq\mathscr{I}_{WIB}(\cH,\cK) $ and $\mathscr{I}_{TP}(\cH, \cK)\subseteq\mathscr{I}_{EB}(\cH,\cK)\subseteq \mathscr{I}_{IB}(\cH,\cK)\subseteq\mathscr{I}_{WIB}(\cH,\cK) $. Therefore, from Proposition \ref{Proposi:QI_hierarchy} and Corollary \ref{Coro:QI_hierarchy}, we have the following set of inequalities.

\begin{align}
    \mathbbm{R}_{IP}\geq \mathbbm{R}_{EP} \geq \mathbbm{R}_{SEP} \geq \mathbbm{R}_{SMIP},\\
    \overline{\mathbbm{R}}_{IP}\geq \overline{\mathbbm{R}}_{EP} \geq \overline{\mathbbm{R}}_{SEP} \geq \overline{\mathbbm{R}}_{SMIP},\\
    \mathbbm{R}_{IP}\geq \mathbbm{R}_{EP} \geq \mathbbm{R}_{MIP} \geq \mathbbm{R}_{SMIP},\\
    \overline{\mathbbm{R}}_{IP}\geq \overline{\mathbbm{R}}_{EP} \geq \overline{\mathbbm{R}}_{MIP} \geq \overline{\mathbbm{R}}_{SMIP}.
\end{align}

\section{Relation of the resource measure $\mathbbm{R}$ with advantage in an information-theoretic task}\label{Sec:inf_theo_task}
Before we explain the information-theoretic task that we study here, we would like to mention the following important point. 
It is known that resource robustness $\mathscr{R}(.)$ has an information-theoretic interpretation, in the sense that the quantity $\mathscr{R}^*(.)=1+\mathscr{R}(.)$ is directly related to the advantage provided by a resourceful instrument compared to the best free instrument possible in a guessing game\cite{chitambar_PID}. From Proposition \ref{Propsi:res_rob_and_weight} we know that
\begin{align}
   \mathscr{R}^*(.)\geq\frac{2}{2- \mathbbm{R}(.)}.
\end{align}
Thus, the resource measure $\mathbbm{R}$ is also related to the advantage provided by a resourceful instrument compared to the best free instrument possible in the same. Similarly, since the quantity $\mathscr{W}^*(.)=1+\mathscr{W}(.)$, related to weight $\mathscr{W}$, is directly related to the advantage provided by a resourceful instrument compared to the best free instrument possible in exclusion task \cite{Uola_2020_Adv}, $\mathbbm{R}(.)$ can also be related to the same through Proposition \ref{Propsi:res_rob_and_weight}. Furthermore, as Proposition \ref{Propsi:res_rob_and_weight} implies $\overline{\mathbbm{R}}(.)\leq\mathbbm{R}(.)$, a similar connection can also be established for the other resource measure $\overline{\mathbbm{R}}(.)$.

Now, we are ready to explain the information-theoretic task that we study in this section. For simplicity, we consider an arbitrary resource theory where each object is a single instrument instead of being a set of instruments containing multiple instruments. Therefore, the distance defined in Eq. \eqref{Eq:def_dist_meas_of_set_inst} is just the diamond distance. Consider two parties, Alice and Bob, equipped with the Hilbert spaces $\cH_A,\cH_B$ respectively. Here we have $\cH_A=\cH_B=\cH$. Now, consider the following task

\begin{itemize}
\item Alice and Bob share a bipartite qubit state $\rho_{AB}\in\cS(\cH_A\otimes\cH_B)$.  
\item Alice has an instrument $\mathbf{I}=\{\Lambda_x\}\in\mathscr{I}(\cH,\cK)$. She applies this instrument to her portion of $\rho_{AB}$. Conditioned on the outcome $x$ of the instrument, the obtained un-normalised state can be written as $\rho^{x}_{AB}=(\Lambda_x\otimes\mathbbm{1}_{\cH_B})(\rho_{AB})$. 
\item She then sends her output subsystem to Bob. Bob's goal is to determine the index $x$ with the highest success probability.
\item In order to accomplish that, in general, Bob performs a $d_{\Omega_{\mathbf{I}}}+1$ outcome global measurement $M=\{M(x)\}\in\mathscr{M}(\cK\otimes\cH)$  such that $x=1,2,.....,d_{\Omega_{\mathbf{I}}},0$. Here $0$ corresponds to the round where Bob is inconclusive about the outcome $x$ and simply discards it. Note that one can choose $M(0)=0$ as a special case. The average success probability is given as
\begin{align}
    P_{succ}=\sum_x\tr\left[(\Lambda_x\otimes \mathbbm{1}_{\cH_B})(\rho_{AB})M(x)\right].
\end{align}
\item Alice chooses a free instrument $\mathbf{J}^*=\{\phi_x^*\}\in\cF_{(\cH,\cK)}$ such that 
\begin{align}
    P^{free}_{succ}=&\sum_x\tr\left[(\phi_x^*\otimes \mathbbm{1}_{\cH_B})(\rho_{AB})M(x)\right]\nonumber\\
    =&\max_{\mathbf{J}\in\cF_{(\cH,\cK)}}\sum_x\tr\left[(\phi_x\otimes \mathbbm{1}_{\cH_B})(\rho_{AB})M(x)\right].\label{Eq:Succ_Prob_free}
\end{align}
Thus, the quantity 
\begin{align}
P_{succ}-P^{free}_{succ}=&\sum_x\tr\left[(\Lambda_x\otimes \mathbbm{1})(\rho_{AB})M(x)\right]\nonumber\\
&-\sum_x\tr\left[(\phi_x^*\otimes \mathbbm{1}_{\cH_B})(\rho_{AB})M(x)\right]\nonumber\\
=&\min_{\mathbf{J}\in\cF_{(\cH,\cK)}}\sum_x\tr\left[((\Lambda_x-\phi_x)\otimes \mathbbm{1})(\rho_{AB})M(x)\right],\label{Eq:Game}
\end{align}
quantifies how much the given instrument provides an advantage over the best free instrument in this information-theoretic task.
\end{itemize}

For a single instrument, the resource measure in Eq. \eqref{Eq:Def_res_meas} is written as
\begin{align}
    \mathbbm{R}(\mathbf{I})=&\min_{\mathbf{J}\in\cF_{(\cH,\cK)}}~\cD_{\Diamond}(\mathbf{I},\mathbf{J})\nonumber\\
    \leq&\cD_{\Diamond}(\mathbf{I},\mathbf{J}^*)\nonumber\\
    =&||\Gamma_{\mathbf{I}}-\Gamma_{\mathbf{J}^*}||_{\Diamond}\nonumber\\
    =&||\sum_x(\Lambda_x-\phi_x^*)\otimes\ket{x}\bra{x}||_{\Diamond}\nonumber\\
    =&\max_{\rho_{AB}}||\sum_x((\Lambda_x-\phi_x^*)\otimes\mathbbm{I}_{\cH_B})(\rho_{AB})\otimes\ket{x}\bra{x}||_{1}\nonumber\\
    =&\max_{\rho_{AB}}\max_{0\le M\le \mathbbm{1}_{\cY}}2\tr[\sum_x((\Lambda_x-\phi_x^*)\otimes\mathbbm{I}_{\cH_B})(\rho_{AB})\otimes\ket{x}\bra{x}M]\nonumber\\
    =&\max_{\rho_{AB}}\max_{0\le M\le \mathbbm{1}_{\cY}}2\sum_x\tr[((\Lambda_x-\phi_x^*)\otimes\mathbbm{I}_{\cH_B})(\rho_{AB})M(x)]\nonumber\\
\end{align}
where 
\begin{align}
    0\leq M(x)=\tr_{\cH_{\Omega_{\mathbf{I}}}}[(\mathbbm{1}_{\cK\otimes\cH_{B}}\otimes |x\rangle\langle x|)M]\leq \mathbbm{1}_{\cK\otimes\cH_B}.
\end{align}
 Here, $\cY=\cK\otimes\cH_B\otimes\cH_{\Omega_{\mathbf{I}}}$, and $\mathbf{J}^*$ is the instrument for which the maximum occurs in Eq. \eqref{Eq:Succ_Prob_free}.  In the sixth line, we have used the definition of trace distance with $M\in\cL(\cK\otimes\cH_{B})$. Next, multiplying and dividing the R.H.S by $d_{\Omega_{\mathbf{I}}}=dim(\cH_{\Omega_{\mathbf{I}}})$, the dimension of the Hilbert space $\cH_{\Omega_{\mathbf{I}}}$, we get

\begin{align}
    \mathbbm{R}(\mathbf{I})\leq&\max_{\rho_{AB}}\max_{0\le M\le \mathbbm{1}_{\cY}}2d_{\Omega_{\mathbf{I}}}\sum_x\tr[((\Lambda_x-\phi_x^*)\otimes\mathbbm{I}_{\cH_B})(\rho_{AB})\frac{M(x)}{d_{\Omega_{\mathbf{I}}}}]\nonumber\\
    =&\max_{\rho_{AB}}\max_{0\le M\le \mathbbm{1}_{\cY}}2d_{\Omega_{\mathbf{I}}}\sum_x\tr[((\Lambda_x-\phi_x^*)\otimes\mathbbm{I}_{\cH_B})(\rho_{AB})M^{\prime}(x)].
\end{align}

Now as $0\le M\le \mathbbm{1}_{\cY}$ we have
\begin{align}
    0\le M^{\prime}(x)\leq \frac{\mathbbm{1}_{\cK\otimes\cH_B}}{d_{\Omega_{\mathbf{I}}}},\qquad\sum_x M^{\prime}(x) \leq \,\mathbbm{1}_{\cK\otimes\cH_B}.\label{Eq:Res_Meas}
\end{align}
Defining $M^{\prime}(0):=(\mathbbm{1}_{\cK\otimes\cH_B}-\sum_x M^{\prime}(x))$, we see that $M^{\prime}:=\{M^{\prime}(x),M^{\prime}(0)\}$ form a valid measurement. Let us denote the set of all such measurements as defined above using Eq. \eqref{Eq:Res_Meas} as $\mathscr{M}^{\prime}(\cK\otimes\cH_B)$. It is clear that $\mathscr{M}^{\prime}(\cK\otimes\cH_B)\subseteq\mathscr{M}(\cK\otimes\cH_B)$.  Thus, we have 

\begin{align}
\frac{\mathbbm{R}(\mathbf{I})}{2d_{\Omega_{\mathbf{I}}}}\leq&\max_{\rho_{AB}}\max_{0\le M\le \mathbbm{1}_{\cY}}\sum_x\tr[((\Lambda_x-\phi_x^*)\otimes\mathbbm{I}_{\cH_B})(\rho_{AB})M^{\prime}(x)]\nonumber\\
=&\max_{\rho_{AB}}\max_{M^{\prime}\in\mathscr{M}^{\prime}(\cK\otimes\cH_B)}\sum_x\tr[((\Lambda_x-\phi_x^*)\otimes\mathbbm{I}_{\cH_B})(\rho_{AB})M^{\prime}(x)]\nonumber\\
\leq&\max_{\rho_{AB}}\max_{M^{\prime}\in\mathscr{M}(\cK\otimes\cH_B)}\sum_x\tr[((\Lambda_x-\phi_x^*)\otimes\mathbbm{I}_{\cH_B})(\rho_{AB})M^{\prime}(x)]
\end{align}

This, using Eq. \eqref{Eq:Game}, implies
\begin{align}
    \frac{\mathbbm{R}(\mathbf{I})}{2d_{\Omega_{\mathbf{I}}}}\leq\max_{\rho_{AB}}\max_{M^{\prime}\in\mathscr{M}(\cK\otimes\cH_B)}P_{succ}-P^{free}_{succ}\label{Eq:res_meas_task_lb}
\end{align}.

Note that the R.H.S. of Eq. \eqref{Eq:res_meas_task_lb} can be considered as \emph{the advantage} provided by the instrument $\mathbf{I}$ compared to the best free instrument possible in our information-theoretic task, and the resource measure $\mathbbm{R}(\mathbf{I})$ is the lower bound of the advantage up to a scaling factor $\frac{1}{2d_{\Omega_{\mathbf{I}}}}$. Furthermore, note that if $\mathbf{I}$ is a free instrument, the R.H.S. of Eq. \eqref{Eq:res_meas_task_lb} is clearly zero. But if $\mathbf{I}$ is a resourceful instrument, then $\mathbbm{R}(\mathbf{I})>0$ and therefore, the R.H.S. of Eq. \eqref{Eq:res_meas_task_lb} is lower bounded by a strictly positive number and hence, the advantage provided by $\mathbf{I}$ is nonzero. Moreover, from Proposition \ref{Propsi:res_rob_and_weight}, as we have $\overline{\mathbbm{R}}(.)\leq\mathbbm{R}(.)$, the resource measure $\overline{\mathbbm{R}}(.)$ also provides the lower bound on the advantage, up to a scaling factor, in a similar way.

\section{Conclusion}
\label{Sec:conc}
In this work, we have tried to characterize and quantify some instrument-based quantum resources, to study their hierarchies, and to construct their resource theories. We provided a detailed framework for a variety of instrument-based resource theories. Our work offers a deep insight into these instrument-based resources. In the following, we pointwise summarize our results.
\begin{enumerate}
    \item At first, we have discussed the quantification and distance measures for generic instrument-based resources.
    \item We have tried to provide a method to compute our distance measure and resource measures using SDP.
    \item We then tried to motivate and develop resource theories for various instrument-based quantum resources, highlighting their significance as valuable operational resources. Detailed descriptions are provided as follows:
    \begin{enumerate}
        \item We have tried to construct the resource theory of information-preservability, considering sets of trash-and-prepare instruments as free objects.
        \item We have tried to construct the resource theory of entanglement-preservability and the resource theory of strong entanglement-preservability, considering sets of entanglement-breaking instruments and sets of weak entanglement-breaking instruments as free objects, respectively.
        \item We have tried to construct the resource theory of incompatibility-preservability and the resource theory of strong incompatibility-preservability, considering sets of incompatibility-breaking instruments and sets of weak incompatibility-breaking instruments as free objects, respectively.
        \item The resource theory of traditional incompatibility has already been constructed in Ref. \cite{chitambar_PID}. We have shown that the distance measure $\widehat{\cD}$ is non-increasing under the free transformations of the resource theory of traditional compatibility.
        \item We have tried to construct the resource theory of parallel incompatibility considering sets of parallel compatible instruments as free objects.
    \end{enumerate}
    \item While exploring the above-mentioned resource theories, we have also studied the hierarchies among the free objects of these resource theories, which implied hierarchies among the resource measures.
    \item We have also tried to establish a relationship between our resource measures and the advantage in an information-theoretic task.

\end{enumerate}
    In short, we have tried to provide a \emph{complete framework} for a \emph{wide variety} of instrument-based resource theories. Our work opens up several research avenues. Here, we enlist some of those.

    \begin{enumerate}
        \item It is important to explore the one-shot and asymptotic conversion among resourceful objects under free transformations for all of the above-mentioned instrument-based resource theories.
        \item It is interesting to study resource-assisted transformation among resourceful objects under free transformations for all of the above-mentioned instrument-based resource theories.
        \item It is important to investigate whether at least some of the above-mentioned resource theories admit the notion of catalysis.
        \item It is also interesting to investigate whether optimal resources are equivalent under free transformations for the above-mentioned instrument-based resource theories.
        \item We know that there exists a notion of "layers of classicality" in the set of all compatible pairs of measurements, and they are non-convex \cite{heinnosari_layers_incom,arindam_layers_incom}. It will be worthwhile to explore whether analogous layers of classicality exist for both traditional and parallel incompatibility of instruments, and whether convex resource theories can be formulated for these layers.
        
    \end{enumerate}

\section{Acknowledgments}
The authors would like to thank Dr. Debashis Saha, Dr. Chandan Datta, and Prof. Sibasish Ghosh for their useful comments. AM acknowledges STARS (Grant No.
STARS/STARS-2/2023-0809), Government of India, for support.

\bibliography{references}

\begin{thebibliography}{53}%
\makeatletter
\providecommand \@ifxundefined [1]{%
 \@ifx{#1\undefined}
}%
\providecommand \@ifnum [1]{%
 \ifnum #1\expandafter \@firstoftwo
 \else \expandafter \@secondoftwo
 \fi
}%
\providecommand \@ifx [1]{%
 \ifx #1\expandafter \@firstoftwo
 \else \expandafter \@secondoftwo
 \fi
}%
\providecommand \natexlab [1]{#1}%
\providecommand \enquote  [1]{``#1''}%
\providecommand \bibnamefont  [1]{#1}%
\providecommand \bibfnamefont [1]{#1}%
\providecommand \citenamefont [1]{#1}%
\providecommand \href@noop [0]{\@secondoftwo}%
\providecommand \href [0]{\begingroup \@sanitize@url \@href}%
\providecommand \@href[1]{\@@startlink{#1}\@@href}%
\providecommand \@@href[1]{\endgroup#1\@@endlink}%
\providecommand \@sanitize@url [0]{\catcode `\\12\catcode `\$12\catcode `\&12\catcode `\#12\catcode `\^12\catcode `\_12\catcode `\%12\relax}%
\providecommand \@@startlink[1]{}%
\providecommand \@@endlink[0]{}%
\providecommand \url  [0]{\begingroup\@sanitize@url \@url }%
\providecommand \@url [1]{\endgroup\@href {#1}{\urlprefix }}%
\providecommand \urlprefix  [0]{URL }%
\providecommand \Eprint [0]{\href }%
\providecommand \doibase [0]{https://doi.org/}%
\providecommand \selectlanguage [0]{\@gobble}%
\providecommand \bibinfo  [0]{\@secondoftwo}%
\providecommand \bibfield  [0]{\@secondoftwo}%
\providecommand \translation [1]{[#1]}%
\providecommand \BibitemOpen [0]{}%
\providecommand \bibitemStop [0]{}%
\providecommand \bibitemNoStop [0]{.\EOS\space}%
\providecommand \EOS [0]{\spacefactor3000\relax}%
\providecommand \BibitemShut  [1]{\csname bibitem#1\endcsname}%
\let\auto@bib@innerbib\@empty
\bibitem [{\citenamefont {Skrzypczyk}\ \emph {et~al.}(2019)\citenamefont {Skrzypczyk}, \citenamefont {\ifmmode \check{S}\else \v{S}\fi{}upi\ifmmode~\acute{c}\else \'{c}\fi{}},\ and\ \citenamefont {Cavalcanti}}]{Skrzypczyk_incomp_state_disc}%
  \BibitemOpen
  \bibfield  {author} {\bibinfo {author} {\bibfnamefont {P.}~\bibnamefont {Skrzypczyk}}, \bibinfo {author} {\bibfnamefont {I.}~\bibnamefont {\ifmmode \check{S}\else \v{S}\fi{}upi\ifmmode~\acute{c}\else \'{c}\fi{}}},\ and\ \bibinfo {author} {\bibfnamefont {D.}~\bibnamefont {Cavalcanti}},\ }\bibfield  {title} {\bibinfo {title} {All sets of incompatible measurements give an advantage in quantum state discrimination},\ }\href {https://doi.org/10.1103/PhysRevLett.122.130403} {\bibfield  {journal} {\bibinfo  {journal} {Phys. Rev. Lett.}\ }\textbf {\bibinfo {volume} {122}},\ \bibinfo {pages} {130403} (\bibinfo {year} {2019})}\BibitemShut {NoStop}%
\bibitem [{\citenamefont {Bennett}\ \emph {et~al.}(1993)\citenamefont {Bennett}, \citenamefont {Brassard}, \citenamefont {Cr\'epeau}, \citenamefont {Jozsa}, \citenamefont {Peres},\ and\ \citenamefont {Wootters}}]{bennet_teleportation}%
  \BibitemOpen
  \bibfield  {author} {\bibinfo {author} {\bibfnamefont {C.~H.}\ \bibnamefont {Bennett}}, \bibinfo {author} {\bibfnamefont {G.}~\bibnamefont {Brassard}}, \bibinfo {author} {\bibfnamefont {C.}~\bibnamefont {Cr\'epeau}}, \bibinfo {author} {\bibfnamefont {R.}~\bibnamefont {Jozsa}}, \bibinfo {author} {\bibfnamefont {A.}~\bibnamefont {Peres}},\ and\ \bibinfo {author} {\bibfnamefont {W.~K.}\ \bibnamefont {Wootters}},\ }\bibfield  {title} {\bibinfo {title} {Teleporting an unknown quantum state via dual classical and einstein-podolsky-rosen channels},\ }\href {https://doi.org/10.1103/PhysRevLett.70.1895} {\bibfield  {journal} {\bibinfo  {journal} {Phys. Rev. Lett.}\ }\textbf {\bibinfo {volume} {70}},\ \bibinfo {pages} {1895} (\bibinfo {year} {1993})}\BibitemShut {NoStop}%
\bibitem [{\citenamefont {Uola}\ \emph {et~al.}(2020{\natexlab{a}})\citenamefont {Uola}, \citenamefont {Bullock}, \citenamefont {Kraft}, \citenamefont {Pellonp\"a\"a},\ and\ \citenamefont {Brunner}}]{Uola_q_res_exclus}%
  \BibitemOpen
  \bibfield  {author} {\bibinfo {author} {\bibfnamefont {R.}~\bibnamefont {Uola}}, \bibinfo {author} {\bibfnamefont {T.}~\bibnamefont {Bullock}}, \bibinfo {author} {\bibfnamefont {T.}~\bibnamefont {Kraft}}, \bibinfo {author} {\bibfnamefont {J.-P.}\ \bibnamefont {Pellonp\"a\"a}},\ and\ \bibinfo {author} {\bibfnamefont {N.}~\bibnamefont {Brunner}},\ }\bibfield  {title} {\bibinfo {title} {All quantum resources provide an advantage in exclusion tasks},\ }\href {https://doi.org/10.1103/PhysRevLett.125.110402} {\bibfield  {journal} {\bibinfo  {journal} {Phys. Rev. Lett.}\ }\textbf {\bibinfo {volume} {125}},\ \bibinfo {pages} {110402} (\bibinfo {year} {2020}{\natexlab{a}})}\BibitemShut {NoStop}%
\bibitem [{\citenamefont {Carmeli}\ \emph {et~al.}(2020)\citenamefont {Carmeli}, \citenamefont {Heinosaari},\ and\ \citenamefont {Toigo}}]{Carmeli_incomp_meas_qrac}%
  \BibitemOpen
  \bibfield  {author} {\bibinfo {author} {\bibfnamefont {C.}~\bibnamefont {Carmeli}}, \bibinfo {author} {\bibfnamefont {T.}~\bibnamefont {Heinosaari}},\ and\ \bibinfo {author} {\bibfnamefont {A.}~\bibnamefont {Toigo}},\ }\bibfield  {title} {\bibinfo {title} {Quantum random access codes and incompatibility of measurements},\ }\href {https://doi.org/10.1209/0295-5075/130/50001} {\bibfield  {journal} {\bibinfo  {journal} {Europhysics Letters}\ }\textbf {\bibinfo {volume} {130}},\ \bibinfo {pages} {50001} (\bibinfo {year} {2020})}\BibitemShut {NoStop}%
\bibitem [{\citenamefont {Mori}(2020)}]{Mori_incomp_chan_state_disc}%
  \BibitemOpen
  \bibfield  {author} {\bibinfo {author} {\bibfnamefont {J.}~\bibnamefont {Mori}},\ }\bibfield  {title} {\bibinfo {title} {Operational characterization of incompatibility of quantum channels with quantum state discrimination},\ }\href {https://doi.org/10.1103/PhysRevA.101.032331} {\bibfield  {journal} {\bibinfo  {journal} {Phys. Rev. A}\ }\textbf {\bibinfo {volume} {101}},\ \bibinfo {pages} {032331} (\bibinfo {year} {2020})}\BibitemShut {NoStop}%
\bibitem [{\citenamefont {Saha}\ \emph {et~al.}(2023)\citenamefont {Saha}, \citenamefont {Das}, \citenamefont {Das}, \citenamefont {Bhattacharya},\ and\ \citenamefont {Majumdar}}]{debasish_incom_random_acces_code}%
  \BibitemOpen
  \bibfield  {author} {\bibinfo {author} {\bibfnamefont {D.}~\bibnamefont {Saha}}, \bibinfo {author} {\bibfnamefont {D.}~\bibnamefont {Das}}, \bibinfo {author} {\bibfnamefont {A.~K.}\ \bibnamefont {Das}}, \bibinfo {author} {\bibfnamefont {B.}~\bibnamefont {Bhattacharya}},\ and\ \bibinfo {author} {\bibfnamefont {A.~S.}\ \bibnamefont {Majumdar}},\ }\bibfield  {title} {\bibinfo {title} {Measurement incompatibility and quantum advantage in communication},\ }\href {https://doi.org/10.1103/PhysRevA.107.062210} {\bibfield  {journal} {\bibinfo  {journal} {Phys. Rev. A}\ }\textbf {\bibinfo {volume} {107}},\ \bibinfo {pages} {062210} (\bibinfo {year} {2023})}\BibitemShut {NoStop}%
\bibitem [{\citenamefont {Aimet}\ and\ \citenamefont {Kwon}(2023)}]{kwon_coh_heat_engine}%
  \BibitemOpen
  \bibfield  {author} {\bibinfo {author} {\bibfnamefont {S.}~\bibnamefont {Aimet}}\ and\ \bibinfo {author} {\bibfnamefont {H.}~\bibnamefont {Kwon}},\ }\bibfield  {title} {\bibinfo {title} {Engineering a heat engine purely driven by quantum coherence},\ }\href {https://doi.org/10.1103/PhysRevA.107.012221} {\bibfield  {journal} {\bibinfo  {journal} {Phys. Rev. A}\ }\textbf {\bibinfo {volume} {107}},\ \bibinfo {pages} {012221} (\bibinfo {year} {2023})}\BibitemShut {NoStop}%
\bibitem [{\citenamefont {Bresque}\ \emph {et~al.}(2021)\citenamefont {Bresque}, \citenamefont {Camati}, \citenamefont {Rogers}, \citenamefont {Murch}, \citenamefont {Jordan},\ and\ \citenamefont {Auff\`eves}}]{bresque_eng_ent}%
  \BibitemOpen
  \bibfield  {author} {\bibinfo {author} {\bibfnamefont {L.}~\bibnamefont {Bresque}}, \bibinfo {author} {\bibfnamefont {P.~A.}\ \bibnamefont {Camati}}, \bibinfo {author} {\bibfnamefont {S.}~\bibnamefont {Rogers}}, \bibinfo {author} {\bibfnamefont {K.}~\bibnamefont {Murch}}, \bibinfo {author} {\bibfnamefont {A.~N.}\ \bibnamefont {Jordan}},\ and\ \bibinfo {author} {\bibfnamefont {A.}~\bibnamefont {Auff\`eves}},\ }\bibfield  {title} {\bibinfo {title} {Two-qubit engine fueled by entanglement and local measurements},\ }\href {https://doi.org/10.1103/PhysRevLett.126.120605} {\bibfield  {journal} {\bibinfo  {journal} {Phys. Rev. Lett.}\ }\textbf {\bibinfo {volume} {126}},\ \bibinfo {pages} {120605} (\bibinfo {year} {2021})}\BibitemShut {NoStop}%
\bibitem [{\citenamefont {Wang}\ \emph {et~al.}(2022)\citenamefont {Wang}, \citenamefont {Xia}, \citenamefont {Bresque},\ and\ \citenamefont {Xue}}]{Wang_exp_eng_ent}%
  \BibitemOpen
  \bibfield  {author} {\bibinfo {author} {\bibfnamefont {K.}~\bibnamefont {Wang}}, \bibinfo {author} {\bibfnamefont {R.}~\bibnamefont {Xia}}, \bibinfo {author} {\bibfnamefont {L.}~\bibnamefont {Bresque}},\ and\ \bibinfo {author} {\bibfnamefont {P.}~\bibnamefont {Xue}},\ }\bibfield  {title} {\bibinfo {title} {Experimental demonstration of a quantum engine driven by entanglement and local measurements},\ }\href {https://doi.org/10.1103/PhysRevResearch.4.L032042} {\bibfield  {journal} {\bibinfo  {journal} {Phys. Rev. Res.}\ }\textbf {\bibinfo {volume} {4}},\ \bibinfo {pages} {L032042} (\bibinfo {year} {2022})}\BibitemShut {NoStop}%
\bibitem [{\citenamefont {Shi}\ \emph {et~al.}(2020)\citenamefont {Shi}, \citenamefont {Shi}, \citenamefont {Wang}, \citenamefont {Hu}, \citenamefont {Liu}, \citenamefont {Yang},\ and\ \citenamefont {Fan}}]{shi_coh_heat_eng}%
  \BibitemOpen
  \bibfield  {author} {\bibinfo {author} {\bibfnamefont {Y.-H.}\ \bibnamefont {Shi}}, \bibinfo {author} {\bibfnamefont {H.-L.}\ \bibnamefont {Shi}}, \bibinfo {author} {\bibfnamefont {X.-H.}\ \bibnamefont {Wang}}, \bibinfo {author} {\bibfnamefont {M.-L.}\ \bibnamefont {Hu}}, \bibinfo {author} {\bibfnamefont {S.-Y.}\ \bibnamefont {Liu}}, \bibinfo {author} {\bibfnamefont {W.-L.}\ \bibnamefont {Yang}},\ and\ \bibinfo {author} {\bibfnamefont {H.}~\bibnamefont {Fan}},\ }\bibfield  {title} {\bibinfo {title} {Quantum coherence in a quantum heat engine},\ }\href {https://doi.org/10.1088/1751-8121/ab6a6b} {\bibfield  {journal} {\bibinfo  {journal} {Journal of Physics A: Mathematical and Theoretical}\ }\textbf {\bibinfo {volume} {53}},\ \bibinfo {pages} {085301} (\bibinfo {year} {2020})}\BibitemShut {NoStop}%
\bibitem [{\citenamefont {Chitambar}\ and\ \citenamefont {Gour}(2019)}]{Chitambar_QRT_review}%
  \BibitemOpen
  \bibfield  {author} {\bibinfo {author} {\bibfnamefont {E.}~\bibnamefont {Chitambar}}\ and\ \bibinfo {author} {\bibfnamefont {G.}~\bibnamefont {Gour}},\ }\bibfield  {title} {\bibinfo {title} {Quantum resource theories},\ }\href {https://doi.org/10.1103/RevModPhys.91.025001} {\bibfield  {journal} {\bibinfo  {journal} {Rev. Mod. Phys.}\ }\textbf {\bibinfo {volume} {91}},\ \bibinfo {pages} {025001} (\bibinfo {year} {2019})}\BibitemShut {NoStop}%
\bibitem [{\citenamefont {Gour}(2024)}]{Gour_QRT_book}%
  \BibitemOpen
  \bibfield  {author} {\bibinfo {author} {\bibfnamefont {G.}~\bibnamefont {Gour}},\ }\href {https://doi.org/10.22331/q-2024-01-25-1235} {\emph {\bibinfo {title} {Resources of the quantum world}}}\ (\bibinfo  {publisher} {arXiv preprint arXiv:2402.05474},\ \bibinfo {year} {2024})\BibitemShut {NoStop}%
\bibitem [{\citenamefont {Horodecki}\ \emph {et~al.}(2009)\citenamefont {Horodecki}, \citenamefont {Horodecki}, \citenamefont {Horodecki},\ and\ \citenamefont {Horodecki}}]{Horodecki_review_entang}%
  \BibitemOpen
  \bibfield  {author} {\bibinfo {author} {\bibfnamefont {R.}~\bibnamefont {Horodecki}}, \bibinfo {author} {\bibfnamefont {P.}~\bibnamefont {Horodecki}}, \bibinfo {author} {\bibfnamefont {M.}~\bibnamefont {Horodecki}},\ and\ \bibinfo {author} {\bibfnamefont {K.}~\bibnamefont {Horodecki}},\ }\bibfield  {title} {\bibinfo {title} {Quantum entanglement},\ }\href {https://doi.org/10.1103/RevModPhys.81.865} {\bibfield  {journal} {\bibinfo  {journal} {Rev. Mod. Phys.}\ }\textbf {\bibinfo {volume} {81}},\ \bibinfo {pages} {865} (\bibinfo {year} {2009})}\BibitemShut {NoStop}%
\bibitem [{\citenamefont {Baumgratz}\ \emph {et~al.}(2014)\citenamefont {Baumgratz}, \citenamefont {Cramer},\ and\ \citenamefont {Plenio}}]{Baumgratz_coh_RT}%
  \BibitemOpen
  \bibfield  {author} {\bibinfo {author} {\bibfnamefont {T.}~\bibnamefont {Baumgratz}}, \bibinfo {author} {\bibfnamefont {M.}~\bibnamefont {Cramer}},\ and\ \bibinfo {author} {\bibfnamefont {M.~B.}\ \bibnamefont {Plenio}},\ }\bibfield  {title} {\bibinfo {title} {Quantifying coherence},\ }\href {https://doi.org/10.1103/PhysRevLett.113.140401} {\bibfield  {journal} {\bibinfo  {journal} {Phys. Rev. Lett.}\ }\textbf {\bibinfo {volume} {113}},\ \bibinfo {pages} {140401} (\bibinfo {year} {2014})}\BibitemShut {NoStop}%
\bibitem [{\citenamefont {Winter}\ and\ \citenamefont {Yang}(2016)}]{Winter_coh_RT}%
  \BibitemOpen
  \bibfield  {author} {\bibinfo {author} {\bibfnamefont {A.}~\bibnamefont {Winter}}\ and\ \bibinfo {author} {\bibfnamefont {D.}~\bibnamefont {Yang}},\ }\bibfield  {title} {\bibinfo {title} {Operational resource theory of coherence},\ }\href {https://doi.org/10.1103/PhysRevLett.116.120404} {\bibfield  {journal} {\bibinfo  {journal} {Phys. Rev. Lett.}\ }\textbf {\bibinfo {volume} {116}},\ \bibinfo {pages} {120404} (\bibinfo {year} {2016})}\BibitemShut {NoStop}%
\bibitem [{\citenamefont {Bischof}\ \emph {et~al.}(2019)\citenamefont {Bischof}, \citenamefont {Kampermann},\ and\ \citenamefont {Bru\ss{}}}]{Bischof_coh_RT}%
  \BibitemOpen
  \bibfield  {author} {\bibinfo {author} {\bibfnamefont {F.}~\bibnamefont {Bischof}}, \bibinfo {author} {\bibfnamefont {H.}~\bibnamefont {Kampermann}},\ and\ \bibinfo {author} {\bibfnamefont {D.}~\bibnamefont {Bru\ss{}}},\ }\bibfield  {title} {\bibinfo {title} {Resource theory of coherence based on positive-operator-valued measures},\ }\href {https://doi.org/10.1103/PhysRevLett.123.110402} {\bibfield  {journal} {\bibinfo  {journal} {Phys. Rev. Lett.}\ }\textbf {\bibinfo {volume} {123}},\ \bibinfo {pages} {110402} (\bibinfo {year} {2019})}\BibitemShut {NoStop}%
\bibitem [{\citenamefont {Buscemi}\ \emph {et~al.}(2020)\citenamefont {Buscemi}, \citenamefont {Chitambar},\ and\ \citenamefont {Zhou}}]{Buscemi_meas_incomp}%
  \BibitemOpen
  \bibfield  {author} {\bibinfo {author} {\bibfnamefont {F.}~\bibnamefont {Buscemi}}, \bibinfo {author} {\bibfnamefont {E.}~\bibnamefont {Chitambar}},\ and\ \bibinfo {author} {\bibfnamefont {W.}~\bibnamefont {Zhou}},\ }\bibfield  {title} {\bibinfo {title} {Complete resource theory of quantum incompatibility as quantum programmability},\ }\href {https://doi.org/10.1103/PhysRevLett.124.120401} {\bibfield  {journal} {\bibinfo  {journal} {Phys. Rev. Lett.}\ }\textbf {\bibinfo {volume} {124}},\ \bibinfo {pages} {120401} (\bibinfo {year} {2020})}\BibitemShut {NoStop}%
\bibitem [{\citenamefont {Baek}\ \emph {et~al.}(2020)\citenamefont {Baek}, \citenamefont {Sohbi}, \citenamefont {Lee}, \citenamefont {Kim},\ and\ \citenamefont {Nha}}]{baek_quantifying}%
  \BibitemOpen
  \bibfield  {author} {\bibinfo {author} {\bibfnamefont {K.}~\bibnamefont {Baek}}, \bibinfo {author} {\bibfnamefont {A.}~\bibnamefont {Sohbi}}, \bibinfo {author} {\bibfnamefont {J.}~\bibnamefont {Lee}}, \bibinfo {author} {\bibfnamefont {J.}~\bibnamefont {Kim}},\ and\ \bibinfo {author} {\bibfnamefont {H.}~\bibnamefont {Nha}},\ }\bibfield  {title} {\bibinfo {title} {Quantifying coherence of quantum measurements},\ }\href {https://doi.org/10.1088/1367-2630/abad7e} {\bibfield  {journal} {\bibinfo  {journal} {New Journal of Physics}\ }\textbf {\bibinfo {volume} {22}},\ \bibinfo {pages} {093019} (\bibinfo {year} {2020})}\BibitemShut {NoStop}%
\bibitem [{\citenamefont {Mitra}(2022)}]{Mitra_sharp_meas}%
  \BibitemOpen
  \bibfield  {author} {\bibinfo {author} {\bibfnamefont {A.}~\bibnamefont {Mitra}},\ }\bibfield  {title} {\bibinfo {title} {Quantifying unsharpness of observables in an outcome-independent way},\ }\href {https://doi.org/10.1007/s10773-022-05219-2} {\bibfield  {journal} {\bibinfo  {journal} {International Journal of Theoretical Physics}\ }\textbf {\bibinfo {volume} {61}},\ \bibinfo {pages} {236} (\bibinfo {year} {2022})}\BibitemShut {NoStop}%
\bibitem [{\citenamefont {Buscemi}\ \emph {et~al.}(2024)\citenamefont {Buscemi}, \citenamefont {Kobayashi},\ and\ \citenamefont {Minagawa}}]{Buscemi_sharp_meas}%
  \BibitemOpen
  \bibfield  {author} {\bibinfo {author} {\bibfnamefont {F.}~\bibnamefont {Buscemi}}, \bibinfo {author} {\bibfnamefont {K.}~\bibnamefont {Kobayashi}},\ and\ \bibinfo {author} {\bibfnamefont {S.}~\bibnamefont {Minagawa}},\ }\bibfield  {title} {\bibinfo {title} {A complete and operational resource theory of measurement sharpness},\ }\href {https://doi.org/10.22331/q-2024-01-25-1235} {\bibfield  {journal} {\bibinfo  {journal} {Quantum}\ }\textbf {\bibinfo {volume} {8}},\ \bibinfo {pages} {1235} (\bibinfo {year} {2024})}\BibitemShut {NoStop}%
\bibitem [{\citenamefont {Ji}\ and\ \citenamefont {Chitambar}(2024)}]{chitambar_PID}%
  \BibitemOpen
  \bibfield  {author} {\bibinfo {author} {\bibfnamefont {K.}~\bibnamefont {Ji}}\ and\ \bibinfo {author} {\bibfnamefont {E.}~\bibnamefont {Chitambar}},\ }\bibfield  {title} {\bibinfo {title} {Incompatibility as a resource for programmable quantum instruments},\ }\href {https://doi.org/10.1103/PRXQuantum.5.010340} {\bibfield  {journal} {\bibinfo  {journal} {PRX Quantum}\ }\textbf {\bibinfo {volume} {5}},\ \bibinfo {pages} {010340} (\bibinfo {year} {2024})}\BibitemShut {NoStop}%
\bibitem [{\citenamefont {Ozawa}(1984)}]{Ozawa_1984_inst}%
  \BibitemOpen
  \bibfield  {author} {\bibinfo {author} {\bibfnamefont {M.}~\bibnamefont {Ozawa}},\ }\bibfield  {title} {\bibinfo {title} {Quantum measuring processes of continuous observables},\ }\href {https://doi.org/10.1063/1.526000} {\bibfield  {journal} {\bibinfo  {journal} {Journal of Mathematical Physics}\ }\textbf {\bibinfo {volume} {25}},\ \bibinfo {pages} {79} (\bibinfo {year} {1984})}\BibitemShut {NoStop}%
\bibitem [{\citenamefont {Heinosaari}\ \emph {et~al.}(2018)\citenamefont {Heinosaari}, \citenamefont {Reitzner}, \citenamefont {Ryb\'ar},\ and\ \citenamefont {Ziman}}]{Heinossari_2018_Qubit}%
  \BibitemOpen
  \bibfield  {author} {\bibinfo {author} {\bibfnamefont {T.}~\bibnamefont {Heinosaari}}, \bibinfo {author} {\bibfnamefont {D.}~\bibnamefont {Reitzner}}, \bibinfo {author} {\bibfnamefont {T.~c.~v.}\ \bibnamefont {Ryb\'ar}},\ and\ \bibinfo {author} {\bibfnamefont {M.}~\bibnamefont {Ziman}},\ }\bibfield  {title} {\bibinfo {title} {Incompatibility of unbiased qubit observables and pauli channels},\ }\href {https://doi.org/10.1103/PhysRevA.97.022112} {\bibfield  {journal} {\bibinfo  {journal} {Phys. Rev. A}\ }\textbf {\bibinfo {volume} {97}},\ \bibinfo {pages} {022112} (\bibinfo {year} {2018})}\BibitemShut {NoStop}%
\bibitem [{\citenamefont {Brown}\ and\ \citenamefont {Colbeck}(2020)}]{Colbeck_2020_NonLocal}%
  \BibitemOpen
  \bibfield  {author} {\bibinfo {author} {\bibfnamefont {P.~J.}\ \bibnamefont {Brown}}\ and\ \bibinfo {author} {\bibfnamefont {R.}~\bibnamefont {Colbeck}},\ }\bibfield  {title} {\bibinfo {title} {Arbitrarily many independent observers can share the nonlocality of a single maximally entangled qubit pair},\ }\href {https://doi.org/10.1103/PhysRevLett.125.090401} {\bibfield  {journal} {\bibinfo  {journal} {Phys. Rev. Lett.}\ }\textbf {\bibinfo {volume} {125}},\ \bibinfo {pages} {090401} (\bibinfo {year} {2020})}\BibitemShut {NoStop}%
\bibitem [{\citenamefont {Sasmal}\ \emph {et~al.}(2018)\citenamefont {Sasmal}, \citenamefont {Das}, \citenamefont {Mal},\ and\ \citenamefont {Majumdar}}]{Sasmal_2018_Steering}%
  \BibitemOpen
  \bibfield  {author} {\bibinfo {author} {\bibfnamefont {S.}~\bibnamefont {Sasmal}}, \bibinfo {author} {\bibfnamefont {D.}~\bibnamefont {Das}}, \bibinfo {author} {\bibfnamefont {S.}~\bibnamefont {Mal}},\ and\ \bibinfo {author} {\bibfnamefont {A.~S.}\ \bibnamefont {Majumdar}},\ }\bibfield  {title} {\bibinfo {title} {Steering a single system sequentially by multiple observers},\ }\href {https://doi.org/10.1103/PhysRevA.98.012305} {\bibfield  {journal} {\bibinfo  {journal} {Phys. Rev. A}\ }\textbf {\bibinfo {volume} {98}},\ \bibinfo {pages} {012305} (\bibinfo {year} {2018})}\BibitemShut {NoStop}%
\bibitem [{\citenamefont {Mohan}\ \emph {et~al.}(2019)\citenamefont {Mohan}, \citenamefont {Tavakoli},\ and\ \citenamefont {Brunner}}]{Mohan_2019}%
  \BibitemOpen
  \bibfield  {author} {\bibinfo {author} {\bibfnamefont {K.}~\bibnamefont {Mohan}}, \bibinfo {author} {\bibfnamefont {A.}~\bibnamefont {Tavakoli}},\ and\ \bibinfo {author} {\bibfnamefont {N.}~\bibnamefont {Brunner}},\ }\bibfield  {title} {\bibinfo {title} {Sequential random access codes and self-testing of quantum measurement instruments},\ }\href {https://doi.org/10.1088/1367-2630/ab3773} {\bibfield  {journal} {\bibinfo  {journal} {New Journal of Physics}\ }\textbf {\bibinfo {volume} {21}},\ \bibinfo {pages} {083034} (\bibinfo {year} {2019})}\BibitemShut {NoStop}%
\bibitem [{\citenamefont {Heinosaari}\ and\ \citenamefont {Ziman}(2011)}]{Heinosaari_book_QF}%
  \BibitemOpen
  \bibfield  {author} {\bibinfo {author} {\bibfnamefont {T.}~\bibnamefont {Heinosaari}}\ and\ \bibinfo {author} {\bibfnamefont {M.}~\bibnamefont {Ziman}},\ }\href {https://doi.org/10.1017/CBO9781139031103} {\emph {\bibinfo {title} {The mathematical language of quantum theory: from uncertainty to entanglement}}}\ (\bibinfo  {publisher} {Cambridge University Press},\ \bibinfo {year} {2011})\BibitemShut {NoStop}%
\bibitem [{\citenamefont {Heinosaari}\ \emph {et~al.}(2016)\citenamefont {Heinosaari}, \citenamefont {Miyadera},\ and\ \citenamefont {Ziman}}]{Heinosaari_incomp_review}%
  \BibitemOpen
  \bibfield  {author} {\bibinfo {author} {\bibfnamefont {T.}~\bibnamefont {Heinosaari}}, \bibinfo {author} {\bibfnamefont {T.}~\bibnamefont {Miyadera}},\ and\ \bibinfo {author} {\bibfnamefont {M.}~\bibnamefont {Ziman}},\ }\bibfield  {title} {\bibinfo {title} {An invitation to quantum incompatibility},\ }\href {https://doi.org/10.1088/1751-8113/49/12/123001} {\bibfield  {journal} {\bibinfo  {journal} {Journal of Physics A: Mathematical and Theoretical}\ }\textbf {\bibinfo {volume} {49}},\ \bibinfo {pages} {123001} (\bibinfo {year} {2016})}\BibitemShut {NoStop}%
\bibitem [{\citenamefont {Horodecki}\ \emph {et~al.}(2003)\citenamefont {Horodecki}, \citenamefont {Shor},\ and\ \citenamefont {Ruskai}}]{Horodecki_gen_EBC}%
  \BibitemOpen
  \bibfield  {author} {\bibinfo {author} {\bibfnamefont {M.}~\bibnamefont {Horodecki}}, \bibinfo {author} {\bibfnamefont {P.~W.}\ \bibnamefont {Shor}},\ and\ \bibinfo {author} {\bibfnamefont {M.~B.}\ \bibnamefont {Ruskai}},\ }\bibfield  {title} {\bibinfo {title} {Entanglement breaking channels},\ }\href {https://doi.org/10.1142/S0129055X03001709} {\bibfield  {journal} {\bibinfo  {journal} {Reviews in Mathematical Physics}\ }\textbf {\bibinfo {volume} {15}},\ \bibinfo {pages} {629} (\bibinfo {year} {2003})}\BibitemShut {NoStop}%
\bibitem [{\citenamefont {Ruskai}(2003)}]{Ruskai_qubit_EBC}%
  \BibitemOpen
  \bibfield  {author} {\bibinfo {author} {\bibfnamefont {M.~B.}\ \bibnamefont {Ruskai}},\ }\bibfield  {title} {\bibinfo {title} {Qubit entanglement breaking channels},\ }\href {https://doi.org/10.1142/S0129055X03001710} {\bibfield  {journal} {\bibinfo  {journal} {Reviews in Mathematical Physics}\ }\textbf {\bibinfo {volume} {15}},\ \bibinfo {pages} {643} (\bibinfo {year} {2003})}\BibitemShut {NoStop}%
\bibitem [{\citenamefont {Heinosaari}\ \emph {et~al.}(2015)\citenamefont {Heinosaari}, \citenamefont {Kiukas}, \citenamefont {Reitzner},\ and\ \citenamefont {Schultz}}]{Heinosaari_incomp_break_chan}%
  \BibitemOpen
  \bibfield  {author} {\bibinfo {author} {\bibfnamefont {T.}~\bibnamefont {Heinosaari}}, \bibinfo {author} {\bibfnamefont {J.}~\bibnamefont {Kiukas}}, \bibinfo {author} {\bibfnamefont {D.}~\bibnamefont {Reitzner}},\ and\ \bibinfo {author} {\bibfnamefont {J.}~\bibnamefont {Schultz}},\ }\bibfield  {title} {\bibinfo {title} {Incompatibility breaking quantum channels},\ }\href {https://doi.org/10.1088/1751-8113/48/43/435301} {\bibfield  {journal} {\bibinfo  {journal} {Journal of Physics A: Mathematical and Theoretical}\ }\textbf {\bibinfo {volume} {48}},\ \bibinfo {pages} {435301} (\bibinfo {year} {2015})}\BibitemShut {NoStop}%
\bibitem [{\citenamefont {Heinosaari}\ and\ \citenamefont {Miyadera}(2017)}]{Heinosaari_incomp_chan}%
  \BibitemOpen
  \bibfield  {author} {\bibinfo {author} {\bibfnamefont {T.}~\bibnamefont {Heinosaari}}\ and\ \bibinfo {author} {\bibfnamefont {T.}~\bibnamefont {Miyadera}},\ }\bibfield  {title} {\bibinfo {title} {Incompatibility of quantum channels},\ }\href {https://doi.org/10.1088/1751-8121/aa5f6b} {\bibfield  {journal} {\bibinfo  {journal} {Journal of Physics A: Mathematical and Theoretical}\ }\textbf {\bibinfo {volume} {50}},\ \bibinfo {pages} {135302} (\bibinfo {year} {2017})}\BibitemShut {NoStop}%
\bibitem [{\citenamefont {Chiribella}\ \emph {et~al.}(2008)\citenamefont {Chiribella}, \citenamefont {D'Ariano},\ and\ \citenamefont {Perinotti}}]{Chiribella_sup_chan}%
  \BibitemOpen
  \bibfield  {author} {\bibinfo {author} {\bibfnamefont {G.}~\bibnamefont {Chiribella}}, \bibinfo {author} {\bibfnamefont {G.~M.}\ \bibnamefont {D'Ariano}},\ and\ \bibinfo {author} {\bibfnamefont {P.}~\bibnamefont {Perinotti}},\ }\bibfield  {title} {\bibinfo {title} {Transforming quantum operations: Quantum supermaps},\ }\href {https://doi.org/10.1209/0295-5075/83/30004} {\bibfield  {journal} {\bibinfo  {journal} {Europhysics Letters}\ }\textbf {\bibinfo {volume} {83}},\ \bibinfo {pages} {30004} (\bibinfo {year} {2008})}\BibitemShut {NoStop}%
\bibitem [{\citenamefont {Gour}(2019)}]{Gour_compar_sup_chan}%
  \BibitemOpen
  \bibfield  {author} {\bibinfo {author} {\bibfnamefont {G.}~\bibnamefont {Gour}},\ }\bibfield  {title} {\bibinfo {title} {Comparison of quantum channels by superchannels},\ }\href {https://doi.org/10.1109/TIT.2019.2907989} {\bibfield  {journal} {\bibinfo  {journal} {IEEE Transactions on Information Theory}\ }\textbf {\bibinfo {volume} {65}},\ \bibinfo {pages} {5880} (\bibinfo {year} {2019})}\BibitemShut {NoStop}%
\bibitem [{\citenamefont {Lepp\"aj\"arvi}\ and\ \citenamefont {Sedl\'ak}(2021)}]{Leevi_postproc_instrument}%
  \BibitemOpen
  \bibfield  {author} {\bibinfo {author} {\bibfnamefont {L.}~\bibnamefont {Lepp\"aj\"arvi}}\ and\ \bibinfo {author} {\bibfnamefont {M.}~\bibnamefont {Sedl\'ak}},\ }\bibfield  {title} {\bibinfo {title} {Postprocessing of quantum instruments},\ }\href {https://doi.org/10.1103/PhysRevA.103.022615} {\bibfield  {journal} {\bibinfo  {journal} {Phys. Rev. A}\ }\textbf {\bibinfo {volume} {103}},\ \bibinfo {pages} {022615} (\bibinfo {year} {2021})}\BibitemShut {NoStop}%
\bibitem [{\citenamefont {Heinosaari}\ \emph {et~al.}(2014)\citenamefont {Heinosaari}, \citenamefont {Miyadera},\ and\ \citenamefont {Reitzner}}]{heinosaari_strong_incomp_dev}%
  \BibitemOpen
  \bibfield  {author} {\bibinfo {author} {\bibfnamefont {T.}~\bibnamefont {Heinosaari}}, \bibinfo {author} {\bibfnamefont {T.}~\bibnamefont {Miyadera}},\ and\ \bibinfo {author} {\bibfnamefont {D.}~\bibnamefont {Reitzner}},\ }\bibfield  {title} {\bibinfo {title} {Strongly incompatible quantum devices},\ }\href {https://doi.org/10.1007/s10701-013-9761-1} {\bibfield  {journal} {\bibinfo  {journal} {Foundations of Physics}\ }\textbf {\bibinfo {volume} {44}},\ \bibinfo {pages} {34} (\bibinfo {year} {2014})}\BibitemShut {NoStop}%
\bibitem [{\citenamefont {Lever}\ \emph {et~al.}(2016)\citenamefont {Lever}, \citenamefont {G{\"u}hne}, \citenamefont {Carusotto},\ and\ \citenamefont {Uola}}]{Lever_Ph.D._Thesis}%
  \BibitemOpen
  \bibfield  {author} {\bibinfo {author} {\bibfnamefont {F.}~\bibnamefont {Lever}}, \bibinfo {author} {\bibfnamefont {O.}~\bibnamefont {G{\"u}hne}}, \bibinfo {author} {\bibfnamefont {I.}~\bibnamefont {Carusotto}},\ and\ \bibinfo {author} {\bibfnamefont {R.}~\bibnamefont {Uola}},\ }\bibfield  {title} {\bibinfo {title} {Measurement incompatibility: a resource for quantum steering},\ }\href {https://www.physik.uni-siegen.de/tqo/publications/theses/fabiano-lever-master-thesis.pdf} {\bibfield  {journal} {\bibinfo  {journal} {Master’s degree Thesis, physik. uni-siegen. de}\ } (\bibinfo {year} {2016})}\BibitemShut {NoStop}%
\bibitem [{\citenamefont {Uola}\ \emph {et~al.}(2020{\natexlab{b}})\citenamefont {Uola}, \citenamefont {Kraft},\ and\ \citenamefont {Abbott}}]{Uola_quant_qd_inp_outp_games}%
  \BibitemOpen
  \bibfield  {author} {\bibinfo {author} {\bibfnamefont {R.}~\bibnamefont {Uola}}, \bibinfo {author} {\bibfnamefont {T.}~\bibnamefont {Kraft}},\ and\ \bibinfo {author} {\bibfnamefont {A.~A.}\ \bibnamefont {Abbott}},\ }\bibfield  {title} {\bibinfo {title} {Quantification of quantum dynamics with input-output games},\ }\href {https://doi.org/10.1103/PhysRevA.101.052306} {\bibfield  {journal} {\bibinfo  {journal} {Phys. Rev. A}\ }\textbf {\bibinfo {volume} {101}},\ \bibinfo {pages} {052306} (\bibinfo {year} {2020}{\natexlab{b}})}\BibitemShut {NoStop}%
\bibitem [{\citenamefont {Mitra}\ and\ \citenamefont {Farkas}(2022)}]{Mitra_comp_inst}%
  \BibitemOpen
  \bibfield  {author} {\bibinfo {author} {\bibfnamefont {A.}~\bibnamefont {Mitra}}\ and\ \bibinfo {author} {\bibfnamefont {M.}~\bibnamefont {Farkas}},\ }\bibfield  {title} {\bibinfo {title} {Compatibility of quantum instruments},\ }\href {https://doi.org/10.1103/PhysRevA.105.052202} {\bibfield  {journal} {\bibinfo  {journal} {Phys. Rev. A}\ }\textbf {\bibinfo {volume} {105}},\ \bibinfo {pages} {052202} (\bibinfo {year} {2022})}\BibitemShut {NoStop}%
\bibitem [{\citenamefont {Mitra}\ and\ \citenamefont {Farkas}(2023)}]{Mitra_char_quantifying_incomp_inst}%
  \BibitemOpen
  \bibfield  {author} {\bibinfo {author} {\bibfnamefont {A.}~\bibnamefont {Mitra}}\ and\ \bibinfo {author} {\bibfnamefont {M.}~\bibnamefont {Farkas}},\ }\bibfield  {title} {\bibinfo {title} {Characterizing and quantifying the incompatibility of quantum instruments},\ }\href {https://doi.org/10.1103/PhysRevA.107.032217} {\bibfield  {journal} {\bibinfo  {journal} {Phys. Rev. A}\ }\textbf {\bibinfo {volume} {107}},\ \bibinfo {pages} {032217} (\bibinfo {year} {2023})}\BibitemShut {NoStop}%
\bibitem [{\citenamefont {Lepp{\"a}j{\"a}rvi}\ and\ \citenamefont {Sedl{\'a}k}(2024)}]{Leevi_incomp_inst}%
  \BibitemOpen
  \bibfield  {author} {\bibinfo {author} {\bibfnamefont {L.}~\bibnamefont {Lepp{\"a}j{\"a}rvi}}\ and\ \bibinfo {author} {\bibfnamefont {M.}~\bibnamefont {Sedl{\'a}k}},\ }\bibfield  {title} {\bibinfo {title} {Incompatibility of quantum instruments},\ }\href {https://doi.org/10.22331/q-2024-02-12-1246} {\bibfield  {journal} {\bibinfo  {journal} {Quantum}\ }\textbf {\bibinfo {volume} {8}},\ \bibinfo {pages} {1246} (\bibinfo {year} {2024})}\BibitemShut {NoStop}%
\bibitem [{\citenamefont {Watrous}(2018)}]{Watrous_book_TQI}%
  \BibitemOpen
  \bibfield  {author} {\bibinfo {author} {\bibfnamefont {J.}~\bibnamefont {Watrous}},\ }\href {https://doi.org/10.1017/9781316848142} {\emph {\bibinfo {title} {The Theory of Quantum Information}}}\ (\bibinfo  {publisher} {Cambridge University Press},\ \bibinfo {year} {2018})\BibitemShut {NoStop}%
\bibitem [{\citenamefont {Mitra}\ \emph {et~al.}(2025)\citenamefont {Mitra}, \citenamefont {Mukherjee},\ and\ \citenamefont {Lee}}]{Mitra_distance_measure_MBQR}%
  \BibitemOpen
  \bibfield  {author} {\bibinfo {author} {\bibfnamefont {A.}~\bibnamefont {Mitra}}, \bibinfo {author} {\bibfnamefont {S.}~\bibnamefont {Mukherjee}},\ and\ \bibinfo {author} {\bibfnamefont {C.}~\bibnamefont {Lee}},\ }\bibfield  {title} {\bibinfo {title} {Distance-based measures and epsilon-measures for measurement-based quantum resources},\ }\href {https://arxiv.org/abs/2505.11331} {\bibfield  {journal} {\bibinfo  {journal} {arXiv:2505.11331}\ } (\bibinfo {year} {2025})}\BibitemShut {NoStop}%
\bibitem [{\citenamefont {Tendick}\ \emph {et~al.}(2023)\citenamefont {Tendick}, \citenamefont {Kliesch}, \citenamefont {Kampermann},\ and\ \citenamefont {Bru{\ss}}}]{Tendick_dist_res_meas}%
  \BibitemOpen
  \bibfield  {author} {\bibinfo {author} {\bibfnamefont {L.}~\bibnamefont {Tendick}}, \bibinfo {author} {\bibfnamefont {M.}~\bibnamefont {Kliesch}}, \bibinfo {author} {\bibfnamefont {H.}~\bibnamefont {Kampermann}},\ and\ \bibinfo {author} {\bibfnamefont {D.}~\bibnamefont {Bru{\ss}}},\ }\bibfield  {title} {\bibinfo {title} {Distance-based resource quantification for sets of quantum measurements},\ }\href {https://doi.org/10.22331/q-2023-05-15-1003} {\bibfield  {journal} {\bibinfo  {journal} {Quantum}\ }\textbf {\bibinfo {volume} {7}},\ \bibinfo {pages} {1003} (\bibinfo {year} {2023})}\BibitemShut {NoStop}%
\bibitem [{\citenamefont {Watrous}(2009)}]{Watrous_SDP}%
  \BibitemOpen
  \bibfield  {author} {\bibinfo {author} {\bibfnamefont {J.}~\bibnamefont {Watrous}},\ }\bibfield  {title} {\bibinfo {title} {Semidefinite programs for completely bounded norms},\ }\href {https://doi.org/10.4086/toc.2009.v005a011} {\bibfield  {journal} {\bibinfo  {journal} {Theory of Computing}\ }\textbf {\bibinfo {volume} {5}},\ \bibinfo {pages} {217} (\bibinfo {year} {2009})}\BibitemShut {NoStop}%
\bibitem [{\citenamefont {Boyd}\ and\ \citenamefont {Vandenberghe}(2004)}]{boyd_convex_opt}%
  \BibitemOpen
  \bibfield  {author} {\bibinfo {author} {\bibfnamefont {S.}~\bibnamefont {Boyd}}\ and\ \bibinfo {author} {\bibfnamefont {L.}~\bibnamefont {Vandenberghe}},\ }\href {https://web.stanford.edu/~boyd/cvxbook/bv_cvxbook.pdf} {\emph {\bibinfo {title} {Convex optimization}}}\ (\bibinfo  {publisher} {Cambridge university press},\ \bibinfo {year} {2004})\BibitemShut {NoStop}%
\bibitem [{\citenamefont {Cormen}\ \emph {et~al.}(2022)\citenamefont {Cormen}, \citenamefont {Leiserson}, \citenamefont {Rivest},\ and\ \citenamefont {Stein}}]{cormen_algo}%
  \BibitemOpen
  \bibfield  {author} {\bibinfo {author} {\bibfnamefont {T.~H.}\ \bibnamefont {Cormen}}, \bibinfo {author} {\bibfnamefont {C.~E.}\ \bibnamefont {Leiserson}}, \bibinfo {author} {\bibfnamefont {R.~L.}\ \bibnamefont {Rivest}},\ and\ \bibinfo {author} {\bibfnamefont {C.}~\bibnamefont {Stein}},\ }\href {https://mitpress.mit.edu/9780262046305/introduction-to-algorithms/} {\emph {\bibinfo {title} {Introduction to algorithms}}}\ (\bibinfo  {publisher} {MIT press},\ \bibinfo {year} {2022})\BibitemShut {NoStop}%
\bibitem [{\citenamefont {Harrow}\ \emph {et~al.}(2004)\citenamefont {Harrow}, \citenamefont {Hayden},\ and\ \citenamefont {Leung}}]{harrow_superdense_coding}%
  \BibitemOpen
  \bibfield  {author} {\bibinfo {author} {\bibfnamefont {A.}~\bibnamefont {Harrow}}, \bibinfo {author} {\bibfnamefont {P.}~\bibnamefont {Hayden}},\ and\ \bibinfo {author} {\bibfnamefont {D.}~\bibnamefont {Leung}},\ }\bibfield  {title} {\bibinfo {title} {Superdense coding of quantum states},\ }\href {https://doi.org/10.1103/PhysRevLett.92.187901} {\bibfield  {journal} {\bibinfo  {journal} {Phys. Rev. Lett.}\ }\textbf {\bibinfo {volume} {92}},\ \bibinfo {pages} {187901} (\bibinfo {year} {2004})}\BibitemShut {NoStop}%
\bibitem [{\citenamefont {Shor}\ and\ \citenamefont {Preskill}(2000)}]{Shor_BB84}%
  \BibitemOpen
  \bibfield  {author} {\bibinfo {author} {\bibfnamefont {P.~W.}\ \bibnamefont {Shor}}\ and\ \bibinfo {author} {\bibfnamefont {J.}~\bibnamefont {Preskill}},\ }\bibfield  {title} {\bibinfo {title} {Simple proof of security of the bb84 quantum key distribution protocol},\ }\href {https://doi.org/10.1103/PhysRevLett.85.441} {\bibfield  {journal} {\bibinfo  {journal} {Phys. Rev. Lett.}\ }\textbf {\bibinfo {volume} {85}},\ \bibinfo {pages} {441} (\bibinfo {year} {2000})}\BibitemShut {NoStop}%
\bibitem [{\citenamefont {Hsieh}\ \emph {et~al.}(2025)\citenamefont {Hsieh}, \citenamefont {Stratton}, \citenamefont {Wu},\ and\ \citenamefont {Ku}}]{hsieh_incom_preservability}%
  \BibitemOpen
  \bibfield  {author} {\bibinfo {author} {\bibfnamefont {C.-Y.}\ \bibnamefont {Hsieh}}, \bibinfo {author} {\bibfnamefont {B.}~\bibnamefont {Stratton}}, \bibinfo {author} {\bibfnamefont {C.-H.}\ \bibnamefont {Wu}},\ and\ \bibinfo {author} {\bibfnamefont {H.-Y.}\ \bibnamefont {Ku}},\ }\bibfield  {title} {\bibinfo {title} {Dynamical resource theory of incompatibility preservability},\ }\href {https://doi.org/10.1103/PhysRevA.111.022422} {\bibfield  {journal} {\bibinfo  {journal} {Phys. Rev. A}\ }\textbf {\bibinfo {volume} {111}},\ \bibinfo {pages} {022422} (\bibinfo {year} {2025})}\BibitemShut {NoStop}%
\bibitem [{\citenamefont {Uola}\ \emph {et~al.}(2020{\natexlab{c}})\citenamefont {Uola}, \citenamefont {Bullock}, \citenamefont {Kraft}, \citenamefont {Pellonp\"a\"a},\ and\ \citenamefont {Brunner}}]{Uola_2020_Adv}%
  \BibitemOpen
  \bibfield  {author} {\bibinfo {author} {\bibfnamefont {R.}~\bibnamefont {Uola}}, \bibinfo {author} {\bibfnamefont {T.}~\bibnamefont {Bullock}}, \bibinfo {author} {\bibfnamefont {T.}~\bibnamefont {Kraft}}, \bibinfo {author} {\bibfnamefont {J.-P.}\ \bibnamefont {Pellonp\"a\"a}},\ and\ \bibinfo {author} {\bibfnamefont {N.}~\bibnamefont {Brunner}},\ }\bibfield  {title} {\bibinfo {title} {All quantum resources provide an advantage in exclusion tasks},\ }\href {https://doi.org/10.1103/PhysRevLett.125.110402} {\bibfield  {journal} {\bibinfo  {journal} {Phys. Rev. Lett.}\ }\textbf {\bibinfo {volume} {125}},\ \bibinfo {pages} {110402} (\bibinfo {year} {2020}{\natexlab{c}})}\BibitemShut {NoStop}%
\bibitem [{\citenamefont {Heinosaari}(2016)}]{heinnosari_layers_incom}%
  \BibitemOpen
  \bibfield  {author} {\bibinfo {author} {\bibfnamefont {T.}~\bibnamefont {Heinosaari}},\ }\bibfield  {title} {\bibinfo {title} {Simultaneous measurement of two quantum observables: Compatibility, broadcasting, and in-between},\ }\href {https://doi.org/10.1103/PhysRevA.93.042118} {\bibfield  {journal} {\bibinfo  {journal} {Phys. Rev. A}\ }\textbf {\bibinfo {volume} {93}},\ \bibinfo {pages} {042118} (\bibinfo {year} {2016})}\BibitemShut {NoStop}%
\bibitem [{\citenamefont {Mitra}(2021)}]{arindam_layers_incom}%
  \BibitemOpen
  \bibfield  {author} {\bibinfo {author} {\bibfnamefont {A.}~\bibnamefont {Mitra}},\ }\bibfield  {title} {\bibinfo {title} {Layers of classicality in the compatibility of measurements},\ }\href {https://doi.org/10.1103/PhysRevA.104.022206} {\bibfield  {journal} {\bibinfo  {journal} {Phys. Rev. A}\ }\textbf {\bibinfo {volume} {104}},\ \bibinfo {pages} {022206} (\bibinfo {year} {2021})}\BibitemShut {NoStop}%
\end{thebibliography}%

\appendix
\section{}
\subsection{Proof of Proposition \ref{Prop:Meas_Chan_Inst_Dist_order}}\label{App:Meas_Chan_Inst_Dist_order}
\begin{proof}
If we suppose that the maximum in Eq. \eqref{Eq:def_dist_meas_of_set_meas} occurs for $i=i^*$ then we have
    \begin{align}
        \widetilde{\cD}(\cM, \overline{\cM}):=&\overline{\cD}(\cG_{\cM},\cG_{\overline{\cM}})\nonumber\\
    =&\max_{i\in\{1,\ldots,n\}}\cD_{\Diamond}(\Gamma_{M_i},\Gamma_{\overline{M}_i})\nonumber\\
    =&\cD_{\Diamond}(\Gamma_{M_{i^*}},\Gamma_{\overline{M}_{i^*}})\nonumber\\
    =&\cD_{\Diamond}(M_{i^*},\overline{M}_{i^*})\nonumber\\
    \leq&\cD_{\Diamond}(\mathbf{I}_{i^*}, \overline{\mathbf{I}}_{i^*})\nonumber\\
    =&\cD_{\Diamond}(\hat{\Gamma}_{\mathbf{I}_{i^*}}, \hat{\Gamma}_{\overline{\mathbf{I}}_{i^*}})\nonumber\\
    \leq&\max_{i\in\{1,\ldots,n\}}\cD_{\Diamond}(\hat{\Gamma}_{\mathbf{I}_{i}}, \hat{\Gamma}_{\overline{\mathbf{I}}_{i}})\nonumber\\
    =&\widehat{\cD}(\cI, \overline{\cI}),
    \end{align}
where in the fourth line we have used Lemma \ref{Le:ins_chan_meas_dist_ineq}. In a similar way, using Eq. \eqref{Eq:def_dist_meas_of_set_chan} we can write
    \begin{align}
        \overline{\cD}(\cC,\overline{\cC}):=&\max_{i\in\{1,\ldots,n\}}\cD_{\Diamond}(\Phi_i,\overline{\Phi}_i)\nonumber\\
    =&\cD_{\Diamond}(\Phi_{i^*},\overline{\Phi}_{i^*})\nonumber\\
    \leq&\cD_{\Diamond}(\hat{\Gamma}_{\mathbf{I}_{i^*}}, \hat{\Gamma}_{\overline{\mathbf{I}}_{i^*}})\nonumber\\
    \leq&\max_{i\in\{1,\ldots,n\}}\cD_{\Diamond}(\hat{\Gamma}_{\mathbf{I}_{i}}, \hat{\Gamma}_{\overline{\mathbf{I}}_{i}})\nonumber\\
    =&\widehat{\cD}(\cI, \overline{\cI}).
    \end{align}
\end{proof}
\subsection{Proof of Theorem \ref{Th:Dist_cont_post_process}}\label{App:Dist_cont_post_process}
\begin{proof}
Consider two sets of instruments $\tilde{\cI}=\{\tilde{\mathbf{I}}^i=\{\tilde\Lambda^{i}_{y^i}\}\in\mathscr{I}(\cH,\tilde{\cK})\}^n_{i=1}$ and $\tilde{\cJ}=\{\tilde{\mathbf{J}}^i=\{\tilde\Phi^{i}_{y^i}\}\in\mathscr{I}(\cH,\tilde{\cK})\}^n_{i=1}$ and suppose that in Eq. (\ref{Eq:def_dist_meas_of_set_inst}) the maximum occurs for $i=i^*$. Then we can write
     \begin{align}
\widehat{\cD}(\tilde\cI, \tilde\cJ)=&\cD_{\Diamond}(\tilde{\mathbf{I}}^{i^*}, \tilde{\mathbf{J}}^{i^*})\nonumber\\
       =&\cD_{\Diamond}(\hat{\Gamma}_{\tilde{\mathbf{I}}^{i^*}}, \hat{\Gamma}_{\tilde{\mathbf{J}}^{i^*}}),\nonumber
    \end{align}
where 
\begin{align}
\hat{\Gamma}_{\tilde{\mathbf{I}}^{i^*}}=\sum_{y^{i^*}}\tilde{\Lambda}_{y^{i^*}}^{i^*}\otimes\ket{y^{i^*}}\bra{y^{i^*}},\label{Eq:chan_for_instr}\\
\hat{\Gamma}_{\tilde{\mathbf{J}}^{i^*}}=\sum_{y^{i^*}}\tilde{\Phi}_{y^{i^*}}^{i^*}\otimes\ket{y^{i^*}}\bra{y^{i^*}},\label{Eq:chan_for_instr_2}
\end{align}
with $j=1,2$ and $\{\ket{y^{i^*}}\}$ is an orthonormal basis in the Hilbert space $\cH_{\Omega_{\tilde{\mathbf{I}}^{i^*}}}$ and $\cH_{\Omega_{\tilde{\mathbf{J}}^{i^*}}}$. Suppose the instruments from sets $\tilde\cI$ and $\tilde\cJ$ can be post-processed from the sets $\cI=\{\mathbf{I}^i=\{\Lambda^{i}_{x^i}\}_{x^i}\in\mathscr{I}(\cH,\cK)\}_i$ and $\cJ=\{\mathbf{J}^i=\{\Phi^{i}_{x^i}\}_{x^i}\in\mathscr{I}(\cH,\cK)\}_i$ using the same sets of sets of instruments $\{\cP^i=\{\mathbf{P}^{i,x^{i}}=\{P^{i,x^{i}}_{y^i}\}\in\mathscr{I}(\cK,\tilde{\cK})\}\}$. Then
\begin{align}
\tilde\Lambda^{i}_{y^i}=\sum_{x^i}P^{i,x^{i}}_{y^i}\circ\Lambda^{i}_{x^i},\\
\tilde\Phi^{i}_{y^i}=\sum_{x^i}P^{i,x^{i}}_{y^i}\circ\Phi^{i}_{x^i}.
\end{align} 

Eq. (\ref{Eq:chan_for_instr}) and Eq. \ref{Eq:chan_for_instr_2}  can be rewritten as
\begin{align}
\hat{\Gamma}_{\tilde{\mathbf{I}}^{i^*}}=\Theta\circ\hat{\Gamma}_{\mathbf{I}^{i^*}},\\
\hat{\Gamma}_{\tilde{\mathbf{I}}^{i^*}}=\Theta\circ\hat{\Gamma}_{\mathbf{J}^{i^*}},
\end{align}

with 
\begin{align}
\hat{\Gamma}_{\mathbf{I}^{i^*}}=\sum_{x^{i^*}}\Lambda_{x^{i^*}}^{i^*}\otimes\ket{x^{i^*}}\bra{x^{i^*}},\\
\hat{\Gamma}_{\mathbf{J}^{i^*}}=\sum_{x^{i^*}}\Lambda_{x^{i^*}}^{i^*}\otimes\ket{x^{i^*}}\bra{x^{i^*}},
\end{align}
and $\{\ket{x^{i^*}}\}$ is an orthonormal basis in the Hilbert space $\cH_{\Omega_{\mathbf{I}_j^{i^*}}}$ and $\cH_{\Omega_{\mathbf{J}_j^{i^*}}}$. Here for $\Theta\in\mathscr{C}(\cK\otimes\cH_{\Omega_{\mathbf{I}^{i^*}}},\tilde{\cK}\otimes\cH_{\Omega_{\tilde{\mathbf{I}}^{i^*}}})$ and all $\rho\in\cL(\cK\otimes\cH_{\Omega_{\mathbf{I}^{i^*}}})$ we have
\begin{align}
    \Theta(\rho)=\sum_{x^{i^*},y^{i^*}}P^{i,x^{i^*}}_{y^{i^*}}(\bra{x^{i^*}}\rho\ket{x^{i^*}})\otimes\ket{y^{i^*}}\bra{y^{i^*}}.
\end{align}
We see that it is a special case of Eq. (\ref{gensupermap}). Thus using Eq. (\ref{Eq:supermap_contractive_distance}) we get
\begin{align}
    \widehat{\cD}(\tilde\cI, \tilde\cJ)=&\cD_{\Diamond}(\tilde{\mathbf{I}}^{i^*}, \tilde{\mathbf{J}}^{i^*})\nonumber\\
       =&\cD_{\Diamond}(\hat{\Gamma}_{\tilde{\mathbf{I}}^{i^*}}, \hat{\Gamma}_{\tilde{\mathbf{J}}^{i^*}})\nonumber\\
       \leq&\cD_{\Diamond}(\hat{\Gamma}_{\mathbf{I}^{i^*}}, \hat{\Gamma}_{\mathbf{J}^{i^*}})\nonumber\\
       \leq&\max_{i\in\{1,\ldots,n\}}\cD_{\Diamond}(\hat{\Gamma}_{\mathbf{I}_{i}}, \hat{\Gamma}_{\mathbf{J}_{i}})\nonumber\\
       =&\widehat{\cD}(\cI, \cJ)
\end{align}
\end{proof}
\subsection{Proof of Proposition \ref{Propsi:res_rob_and_weight}}\label{App:res_rob_and_weight}
\begin{proof}
    Let us assume that the minimum in Eq. (\ref{Eq:robustness}) occurs for $\tilde\cI=\{\tilde{\mathbf{I}}^{*i}\}$ and $\mathscr{R}(\cI)=r^*$. Then we can write
    \begin{align} &\Phi^i_a=\frac{\Lambda^i_a+r^*\tilde{\Lambda}^{*i}_a}{1+r^*}~\forall~i\nonumber\\ 
     &\hat{\Gamma}_{\mathbf{J}^i}=\frac{\hat{\Gamma}_{\mathbf{I}^i}+r^*\hat{\Gamma}_{\tilde{\mathbf{I}}^{*i}}}{1+r^*}~\forall~i\nonumber\\ &\hat{\Gamma}_{\mathbf{I}^i}-\hat{\Gamma}_{\mathbf{J}^i}=\frac{r^*(\hat{\Gamma}_{\mathbf{I}^i}-\hat{\Gamma}_{\tilde{\mathbf{I}}^{*i}})}{1+r^*} ~\forall~i\nonumber\\ &\vert\vert\hat{\Gamma}_{\mathbf{I}^i}-\hat{\Gamma}_{\mathbf{J}^i}\vert\vert_{\Diamond}=\frac{r^*}{1+r^*}\vert\vert\hat{\Gamma}_{\mathbf{I}^i}-\hat{\Gamma}_{\tilde{\mathbf{I}}^{*i}}\vert\vert_{\Diamond}~\forall~i\nonumber\\
    &\cD_{\Diamond}(\mathbf{I}^i, \mathbf{J}^i)\leq\frac{2r^*}{1+r^*}~\forall~i
    \end{align}
   where we have used the property that the diamond norm is upper-bounded by 2. Now we know
    \begin{align}
        \widehat{\cD}(\cI,\cJ)=&\max_i \cD_{\Diamond}(\mathbf{I}^i, \mathbf{J}^i)\nonumber\\
        =&\cD_{\Diamond}(\mathbf{I}^{i^*}, \mathbf{J}^{i^*})\nonumber\\
        \leq&\frac{2r^*}{1+r^*}.
    \end{align}
    where we have assumed that the maximum happens for $i=i^*$.
    
    Recalling the definition of resource measure in Eq. \eqref{Eq:Def_res_meas}
    \begin{align}
        \mathbbm{R}(\cI)=&\min_{\cJ\in\cF}~\widehat{\cD}(\cI,\cJ)\nonumber\\ =&\widehat{\cD}(\cI,\cJ^*),\nonumber\\
        \leq&\frac{2r^*}{1+r^*},
    \end{align}
    where $\cF$ is the set of free sets of quantum instruments. Remembering $\mathscr{R}(\cI)=r^*$ we get
    \begin{align}
         \mathbbm{R}(\cI)\leq\frac{2\mathscr{R}(\cI)}{1+\mathscr{R}(\cI)}.
    \end{align}

    Similarly, let us suppose in Eq. (\ref{Eq:weight}) the minimum occurs for $\tilde\cI=\{\tilde{\mathbf{I}}^{*i}\}, \tilde\cJ=\{\tilde{\mathbf{J}}^{*i}\}$ and $\mathscr{W}(\cI)=r^*$. Then we get
    \begin{align} &\Lambda^i_a=\frac{\Phi^{*i}_a+r^*\tilde{\Lambda}^{*i}_a}{1+r^*}~\forall~i\nonumber\\ &\hat{\Gamma}_{\mathbf{I}^i}=\frac{\hat{\Gamma}_{\mathbf{J}^{*i}}+r^*\hat{\Gamma}_{\tilde{\mathbf{I}}^{*i}}}{1+r^*}~\forall~i\nonumber\\ &\hat{\Gamma}_{\mathbf{I}^i}-\hat{\Gamma}_{\mathbf{J}^i}=\frac{r^*(\hat{\Gamma}_{\tilde{\mathbf{I}}^{*i}}-\hat{\Gamma}_{\mathbf{J}^{*i}})}{1+r^*}~\forall~i \nonumber\\ &\vert\vert\hat{\Gamma}_{\mathbf{I}^i}-\hat{\Gamma}_{\mathbf{J}^i}\vert\vert_{\Diamond}=\frac{r^*}{1+r^*}\vert\vert\hat{\Gamma}_{\tilde{\mathbf{I}}^{*i}}-\hat{\Gamma}_{\mathbf{J}^{*i}}\vert\vert_{\Diamond}~\forall~i\nonumber\\
    &\cD_{\Diamond}(\mathbf{I}^i, \mathbf{J}^i)\leq\frac{2r^*}{1+r^*}~\forall~i.
    \end{align}
    Following a similar procedure to the previous one, we can conclude
        \begin{align}
         \mathbbm{R}(\cI)\leq\frac{2\mathscr{W}(\cI)}{1+\mathscr{W}(\cI)}.
    \end{align}
Now from Eqs. \eqref{Eq:Def_res_meas} and \eqref{Eq:Def_res_ext_meas} it is clear that
\begin{align}
    \overline{\mathbbm{R}}(\cI)\leq\mathbbm{R}(\cI)
\end{align}
Therefore, we can conclude,
 \begin{align}
         \overline{\mathbbm{R}}(\cI)\leq\mathbbm{R}(\cI)\leq\min\Big\{\frac{2\mathscr{R}(\cI)}{1+\mathscr{R}(\cI)},\frac{2\mathscr{W}(\cI)}{1+\mathscr{W}(\cI)}\Big\}.
    \end{align}
\end{proof}
\section{}
\subsection{Proof of Theorem \ref{Th:Dist_monotone_EP}}\label{App:Dist_monotone_EP}
\begin{proof}
    According to Theorem \ref{Th:free_op_ent_break}, the free transformation of the resource theory of entanglement preservability is
     \begin{align}
       \overline{\Phi}^j_c=\sum_{k_1}q(k_1)\sum_{a,b}&\tilde{\Phi}^{j,b,k_1}_c\circ(\Phi^{a}_b\otimes\mathbbm{I}_{Q^{\prime}})\circ\Phi^{\prime * j,k_1}_{a}\nonumber\\
       +&\sum_{k_2}q(k_2)\sum_{a,b}\tilde{\Phi}^{*j,b,k_2}_c\circ(\Phi^{a}_b\otimes\mathbbm{I}_{Q})\circ\Phi^{\prime j,k_2}_{a}\nonumber\\
       +&\sum_{k_3}q(k_3)\sum_{a,b}\tilde{\Gamma}^{ j,b,k_3}_c\circ\Phi^{a}_b\circ\tilde{\Delta}^{ j,k_3}_a.
    \end{align}
    Symbolically, it can be written as $\overline{\cJ}=\sum_{k_1}q(k_1)\cV^{k_1}[\cI]+\sum_{k_2}q(k_2)\cV^{k_2}[\cI]+\sum_{k_3}q(k_3)\cV^{k_3}[\cI]:=\cW[\cI]$ where the set of instruments $\cV^{k_1}[\cI]=\{\cV^{k_1}[\cI]^{j,b}=\{\cV^{k_1}[\cI]^{j,b}_c=\sum_{a,b}\tilde{\Phi}^{j,b,k_1}_c\circ(\Phi^{a}_b\otimes\mathbbm{I}_{Q^{\prime}})\circ\Phi^{\prime * j,k_1}_{a}\}\}$, $\cV^{k_2}[\cI]=\{\cV^{k_2}[\cI]^{j,b}=\{\cV^{k_2}[\cI]^{j,b}_c=\sum_{a,b}\tilde{\Phi}^{*j,b,k_2}_c\circ(\Phi^{a}_b\otimes\mathbbm{I}_{Q})\circ\Phi^{\prime j,k_2}_{a}\}\}$ and $\cV^{k_3}[\cI]=\{\cV^{k_3}[\cI]^{j,b}=\{\cV^{k_3}[\cI]_c^{j,b}=\sum_{a,b}\tilde{\Gamma}^{ j,b,k_3}_c\circ\Phi^{a}_b\circ\tilde{\Delta}^{ j,k_3}_a\}\}$ . We can also write 
    \begin{align}
        \hat\Gamma_{\overline{\mathbf{J}}^j}(\rho)=&\sum_c\overline{\Phi}_c^j(\rho)\otimes\ket{c}\bra{c}\nonumber\\ =&\sum_c\Big(\sum_{k_1}q(k_1)\sum_{a,b}\tilde{\Phi}^{j,b,k_1}_c\circ(\Phi^{a}_b\otimes\mathbbm{I}_{Q^{\prime}})\circ\Phi^{\prime * j,k_1}_{a}\nonumber\\
       &\quad+\sum_{k_2}q(k_2)\sum_{a,b}\tilde{\Phi}^{*j,b,k_2}_c\circ(\Phi^{a}_b\otimes\mathbbm{I}_{Q})\circ\Phi^{\prime j,k_2}_{a}\nonumber\\
       &\quad+\sum_{k_3}q(k_3)\sum_{a,b}\tilde{\Gamma}^{ j,b,k_3}_c\circ\Phi^{a}_b\circ\tilde{\Delta}^{ j,k_3}_a\Big)\otimes\ket{c}\bra{c}\nonumber\\
       =&\sum_{k_1}q(k_1)\hat\Gamma_{\cV^{k_1}[\cI]^j}+\sum_{k_2}q(k_2)\hat\Gamma_{\cV^{k_2}[\cI]^j}+\sum_{k_3}q(k_3)\hat\Gamma_{\cV^{k_3}[\cI]^j}.
    \end{align}
     Consider the following quantum channels $\Theta_{pre}^{j,k_1}:\cL(\overline{\cH})\rightarrow\cL(\cH\otimes\cH_{\Omega_{\mathbf{J}^{j,k_1}}}\otimes Q^{\prime})$, $\Sigma_{\hat\Gamma_{\cI}}:\cL(\cH\otimes\cH_{\Omega_{\mathbf{J}^{j,k_1}}})\rightarrow\cL(\cK\otimes\cH_{\Omega_{\mathbf{I}^a}}\otimes\cH_{\Omega_{\mathbf{J}^{j,k_1}}})$, and $\Theta_{post}^{j,k_1}:\cL(\cK\otimes\cH_{\Omega_{\mathbf{I}^a}}\otimes\cH_{\Omega_{\mathbf{J}^{j,k_1}}}\otimes Q^{\prime})\rightarrow\cL(\overline{\cK}\otimes\cH_{\Omega_{\overline{\mathbf{J}}^{j,k_1}}})$ such that for all $\rho\in\cL(\overline{\cH})$, for all $\sigma\in\cL(\cH\otimes\cH_{\Omega_{\mathbf{J}^{j,k_1}}})$, and for all $\omega\in\cL(\cK\otimes\cH_{\Omega_{\mathbf{I}^a}}\otimes\cH_{\Omega_{\mathbf{J}^{j,k_1}}}\otimes Q^{\prime})$ we have
     \begin{align}
      \Theta_{pre}^{j,k_1}(\rho)=\mathbbm{I}_{\cH}\otimes\mathbf{SWAP}_{Q^{\prime}\leftrightarrow\cH_{\Omega_{\mathbf{J}^{j,k_1}}}}(\sum_{a''}\Phi^{\prime*j,k_1}_{a''}(\rho)\otimes\ket{a''}\bra{a''}),
        \end{align}
        where $\{\ket{a''}\}$ is the orthonormal basis of Hilbert space $\cH_{\Omega_{\mathbf{J}^{j,k_1}}}$.
         \begin{align}
            \Sigma_{\hat\Gamma_{\cI}}(\sigma)=\sum_{a'}\hat{\Gamma}_{\mathbf{I}^{a^\prime}}(\bra{a^\prime}\sigma\ket{a^\prime})\otimes\ket{a^\prime}\bra{a^\prime},
        \end{align}
        where $\{\ket{a^\prime}\}$ is the orthonormal basis of Hilbert space $\cH_{\Omega_{\mathbf{J}^{j,k_1}}}$, $\hat{\Gamma}_{\mathbf{I}^{a^\prime}}=\sum_{b^\prime}\Phi^{a^\prime}_{b^\prime}\otimes\ket{b^{\prime}}\bra{b^{\prime}}$, and $\{\ket{b^\prime}\}$ is the orthonormal basis of Hilbert space $\cH_{\Omega_{\mathbf{I}^{a^\prime}}}$.
        \begin{align}
            \Theta_{post}^{j,k_1}(\omega)=\sum_{c,a,b}\tilde\Phi^{j,b,k_1}_c(\bra{b,a}\omega\ket{b,a})\otimes\ket{c}\bra{c}.
        \end{align}
        Here $\{\ket{a}\}$ is the orthonormal basis of Hilbert space $\cH_{\Omega_{\mathbf{J}^{j,k_1}}}$, 
        $\{\ket{b}\}$ is the orthonormal basis of Hilbert space $\cH_{\Omega_{\mathbf{I}^a}}$, and $\{\ket{c}\}$ is the orthonormal basis of Hilbert space $\cH_{\Omega_{\overline{\mathbf{J}}^{j}}}$. Similarly consider
        \begin{align}
      \Theta_{pre}^{j,k_2}(\rho)=&\mathbbm{I}_{\cH}\otimes\mathbf{SWAP}_{Q\leftrightarrow\cH_{\Omega_{\mathbf{J}^{j,k_2}}}}(\sum_{a''}\Phi^{\prime j,k_2}_{a''}(\rho)\otimes\ket{a''}\bra{a''}),\nonumber\\
       \Theta_{post}^{j,k_2}(\omega)=&\sum_{c,a,b}\tilde\Phi^{*j,b,k_2}_c(\bra{b,a}\omega\ket{b,a})\otimes\ket{c}\bra{c},\nonumber\\
       \Theta_{pre}^{j,k_3}(\rho)=&\sum_{a''}\tilde{\Delta}^{ j,k_3}_{a''}(\rho)\otimes\ket{a''}\bra{a''},\nonumber\\
        \Theta_{post}^{j,k_3}(\omega)=&\sum_{c,a,b}\tilde{\Gamma}^{ j,b,k_3}_c(\bra{b,a}\omega\ket{b,a})\otimes\ket{c}\bra{c}.
        \end{align}
        
         Here we would like to mention that for the rest of the calculations, we have used the same variables $a'',a'$ and $a$ for the orthonormal bases of the Hilbert spaces $\cH_{\Omega_{\mathbf{J}^{j,k_2}}}$ and $\cH_{\Omega_{\mathbf{J}^{j,k_1}}}$ wherever necessary so that there is no confusion. It can then be easily verified that       
    \begin{align}
        \hat\Gamma_{\overline{\mathbf{J}}^j}(\rho)=\sum_{k_1} q(k_1)~\Theta^{j,k_1}_{post} &\circ(\Sigma_{\cI}\otimes\mathbbm{I}_{Q^{\prime}})\circ\Theta^{j,k_1}_{pre}\nonumber\\
       +&\sum_{k_2}q(k_2)~\Theta^{j,k_2}_{post} \circ(\Sigma_{\cI}\otimes\mathbbm{I}_{Q})\circ\Theta^{j,k_2}_{pre}\nonumber\\
       &\qquad\quad+\sum_{k_3}q(k_3)~\Theta^{j,k_3}_{post} \circ\Sigma_{\cI}\circ\Theta^{j,k_3}_{pre}.
    \end{align}
    
    Clearly, the quantum channels $\hat\Gamma_{\tilde{\cV}^{k_1}[\cI]^j}=\Theta^{j,k_1}_{post} \circ(\Sigma_{\cI}\otimes\mathbbm{I}_{Q^{\prime}})\circ\Theta^{ j,k_1}_{pre}$, $\hat\Gamma_{\cV^{k_2}[\cI]^j}=\Theta^{j,k_2}_{post} \circ(\Sigma_{\cI}\otimes\mathbbm{I}_{Q})\circ\Theta^{ j,k_2}_{pre}$, and $\hat\Gamma_{\cV^{k_3}[\cI]^j}=\Theta^{j,k_3}_{post} \circ\Sigma_{\cI}\circ\Theta^{ j,k_3}_{pre}$ .
   
   Now consider two sets of instruments $\tilde{\cI}_1$ and $\tilde{\cI}_2$ such that $\tilde{\cI}_i=\cW[\cI_i]$ for $i=1,2$. Then 
   \begin{align}
     \widehat{\cD}(\tilde\cI_1,\tilde\cI_2):=&\widehat{\cD}(\cW[\cI_1],\cW[\cI_2]),\nonumber\\
    \leq &\sum_{k_1}q(k_1)\widehat{\cD}(\cV^{k_1}[\cI_1],\cV^{k_1}[\cI_2])\nonumber\\
    &\qquad\qquad+\sum_{k_2}q(k_2)\widehat{\cD}(\cV^{k_2}[\cI_1],\cV^{k_2}[\cI_2]),\nonumber\\
    &\qquad\qquad\qquad+\sum_{k_3}q(k_3)\widehat{\cD}(\cV^{k_3}[\cI_1],\cV^{k_3}[\cI_2])\nonumber\\
    \leq &\sum_{k_1}q(k_1)\widehat{\cD}(\cI_1,\cI_2)+\sum_{k_2}q(k_2)\widehat{\cD}(\cI_1,\cI_2)\nonumber\\
    &\qquad\qquad\qquad+\sum_{k_3}q(k_3)\widehat{\cD}(\cI_1,\cI_2),\nonumber\\
    = & \widehat{\cD}(\cI_1,\cI_2).
\end{align}
Hence $\widehat{\cD}$ is monotonically non-increasing under the free transformations of entanglement preservability.
\end{proof}
\subsection{Proof of Theorem \ref{Th:Dist_monotone_SEP}}\label{App:Dist_monotone_SEP}
\begin{proof}
    The free transformation of the resource theory of strong entanglement preservability is of the form
     \begin{align}
       \overline{\Phi}^j_c=\sum_{k_1}q(k_1)\sum_{a,b}&\tilde{\Phi}^{j,b,k_1}_c\circ(\Phi^{a}_b\otimes\mathbbm{I}_{Q})\circ\Phi^{\prime  j,k_1}_{a}\nonumber\\
       &+\sum_{k_2}q(k_2)\sum_{a,b}\tilde{\Gamma}_c^{j,a,k_2}\circ\Phi^a_b\circ\tilde{\Delta}_a^{j,k_2}.
    \end{align}
    Symbolically, we represent it as $\overline{\cJ}=\sum_{k_1}q(k_1)\cV^{k_1}[\cI]+\sum_{k_2}q(k_2)\cV^{k_2}[\cI]:=\cW[\cI]$ where the CP trace non-increasing maps $\cV^{k_1}[\cI]=\{\cV^{k_1}[\cI]^{j,b}=\{\cV^{k_1}[\cI]^{j,b}_c=\sum_{a,b}\tilde{\Phi}^{j,b,k_1}_c\circ(\Phi^{a}_b\otimes\mathbbm{I}_{Q})\circ\Phi^{\prime j,k_1}_{a}\}\}$ and $\cV^{k_2}[\cI]=\{\cV^{k_2}[\cI]^{j}=\{\sum_{a,b}\tilde{\Gamma}_c^{j,a,k_2}\circ\Phi^a_b\circ\tilde{\Delta}_a^{j,k_2}\}\}$. We can also write 
    \begin{align}
        \hat\Gamma_{\overline{\mathbf{J}}^j}(\rho)=&\sum_c\overline{\Phi}_c^j(\rho)\otimes\ket{c}\bra{c}\nonumber\\ =&\sum_c\Big(\sum_{k_1}q(k_1)\sum_{a,b}\tilde{\Phi}^{j,b,k_1}_c\circ(\Phi^{a}_b\otimes\mathbbm{I}_{Q})\circ\Phi^{\prime  j,k_1}_{a}\nonumber\\
       &+\sum_{k_2}q(k_2)\sum_{a,b}\tilde{\Gamma}_c^{j,a,k_2}\circ\Phi^a_b\circ\tilde{\Delta}_a^{j,k_2}\Big)\otimes\ket{c}\bra{c}\nonumber\\
        =&\sum_{k_1}q(k_1)\hat\Gamma_{\cV^{k_1}[\cI]^j}+\sum_{k_2}q(k_2)\hat\Gamma_{\cV^{k_2}[\cI]^j}
    \end{align}
     Consider the following quantum channels $\Theta_{pre}^{j,k_1}:\cL(\overline{\cH})\rightarrow\cL(\cH\otimes\cH_{\Omega_{\mathbf{J}^{j,k_1}}}\otimes Q)$, $\Sigma_{\hat\Gamma_{\cI}}:\cL(\cH\otimes\cH_{\Omega_{\mathbf{J}^{j,k_1}}})\rightarrow\cL(\cK\otimes\cH_{\Omega_{\mathbf{I}^a}}\otimes\cH_{\Omega_{\mathbf{J}^{j,k_1}}})$, and $\Theta_{post}^{j,k_1}:\cL(\cK\otimes\cH_{\Omega_{\mathbf{I}^a}}\otimes\cH_{\Omega_{\mathbf{J}^{j,k_1}}}\otimes\cH_Q)\rightarrow\cL(\overline{\cK}\otimes\cH_{\Omega_{\overline{\mathbf{J}}^{j,k_1}}})$ such that for all $\rho\in\cL(\overline{\cH})$, for all $\sigma\in\cL(\cH\otimes\cH_{\Omega_{\mathbf{J}^{j,k_1}}})$, and for all $\omega\in\cL(\cK\otimes\cH_{\Omega_{\mathbf{I}^a}}\otimes\cH_{\Omega_{\mathbf{J}^{j,k_1}}}\otimes\cH_Q)$ we have
     \begin{align}
      \Theta_{pre}^{j,k_1}(\rho)=\mathbbm{I}_{\cH}\otimes\mathbf{SWAP}_{Q\leftrightarrow\cH_{\Omega_{\mathbf{J}^{j,k_1}}}}(\sum_{a''}\Phi^{\prime j,k_1}_{a''}(\rho)\otimes\ket{a''}\bra{a''}),
        \end{align}
        where $\{\ket{a''}\}$ is the orthonormal basis of Hilbert space $\cH_{\Omega_{\mathbf{J}^{j,k_1}}}$,
         \begin{align}
            \Sigma_{\hat\Gamma_{\cI}}(\sigma)=\sum_{a'}\hat{\Gamma}_{\mathbf{I}^{a^\prime}}(\bra{a^\prime}\sigma\ket{a^\prime})\otimes\ket{a^\prime}\bra{a^\prime},
        \end{align}
        where $\{\ket{a^\prime}\}$ is the orthonormal basis of Hilbert space $\cH_{\Omega_{\mathbf{J}^{j,k_1}}}$, $\hat{\Gamma}_{\mathbf{I}^{a^\prime}}=\sum_{b^\prime}\Phi^{a^\prime}_{b^\prime}\otimes\ket{b^{\prime}}\bra{b^{\prime}}$, and $\{\ket{b^\prime}\}$ is the orthonormal basis of Hilbert space $\cH_{\Omega_{\mathbf{I}^{a^\prime}}}$,
        \begin{align}
            &\Theta_{post}^{j,k_1}(\omega)=\sum_{c,a,b}\tilde\Phi^{j,b,k_1}_c(\bra{b,a}\omega\ket{b,a})\otimes\ket{c}\bra{c},
        \end{align}
        where $\{\ket{a}\}$ is the orthonormal basis of Hilbert space $\cH_{\Omega_{\mathbf{J}^{j,k_1}}}$, 
        $\{\ket{b}\}$ is the orthonormal basis of Hilbert space $\cH_{\Omega_{\mathbf{I}^a}}$, and $\{\ket{c}\}$ is the orthonormal basis of Hilbert space $\cH_{\Omega_{\overline{\mathbf{J}}^j}}$. Similarly considering
        \begin{align}
            \Theta_{pre}^{j,k_2}(\rho)=&\sum_{a''}\tilde{\Delta}^{ j,k_2}_{a''}(\rho)\otimes\ket{a''}\bra{a''},\nonumber\\
        \Theta_{post}^{j,k_2}(\omega)=&\sum_{c,a,b}\tilde{\Gamma}^{ j,a,k_3}_c(\bra{b,a}\omega\ket{b,a})\otimes\ket{c}\bra{c}.
        \end{align}
       Then, it can be easily verified that       
    \begin{align}
        \hat\Gamma_{\overline{\mathbf{J}}^j}(\rho)=\sum_{k_1} q(k_1)~\Theta^{j,k_1}_{post} &\circ(\Sigma_{\cI}\otimes\mathbbm{I}_{Q})\circ\Theta^{j,k_1}_{pre}\nonumber\\
       +&\sum_{k_2}q(k_2)~\Theta^{j,k_2}_{post} \circ\Sigma_{\cI}\circ\Theta^{j,k_2}_{pre}
    \end{align}
     Consider two quantum instruments of the form $\tilde{\cI}_1=\cW[\cI_1]$ and $\tilde{\cI}_2=\cW[\cI_2]$, then by using Eq. (\ref{Eq:supermap_contractive_distance}), we can write
    \begin{align}
    \widehat{\cD}(\tilde\cI_1,\tilde\cI_2):=&\widehat{\cD}(\cW[\cI_1],\cW[\cI_2]),\nonumber\\
    \leq &\sum_{k_1}q(k_1)\widehat{\cD}(\cV^{k_1}[\cI_1],\cV^{k_1}[\cI_2])\nonumber\\
    &\qquad\qquad+\sum_{k_2}q(k_2)\widehat{\cD}(\cV^{k_2}[\cI_1],\cV^{k_2}[\cI_2]),\nonumber\\
    \leq &\sum_{k_1}q(k_1)\widehat{\cD}(\cI_1,\cI_2)+\sum_{k_2}q(k_2)\widehat{\cD}(\cI_1,\cI_2)\nonumber\\
    = & \widehat{\cD}(\cI_1,\cI_2).
\end{align}
Hence $\widehat{\cD}$ is monotonically non-increasing under the free transformations of strong entanglement preservability.
\end{proof}

\subsection{Proof of Theorem \ref{Th:Dist_mono_IP}}\label{App:Dist_mono_IP}
\begin{proof}
    According to Theorem \ref{Th:free_op_incom_break}, the free transformation of incompatibility preservability is
      \begin{align}
       \overline{\Phi}^j_c=&\sum_{k_1}q(k_1)\sum_{a,b}p(a|j)\tilde{\Phi}^{j,b,k_1}_c\circ(\Phi^{a}_b\otimes\mathbbm{I}_{Q^{\prime}})\circ\Phi^{\prime * j,k_1}\nonumber\\
       &+\sum_{k_2}q(k_2)\sum_{a,b}\tilde{\Phi}^{*j,b,k_2}_c\circ(\Phi^{a}_b\otimes\mathbbm{I}_{Q})\circ\Phi^{\prime j,k_2}_{a}\nonumber\\
       &+\sum_{k_3}q(k_3)\sum_{a,b}\tilde{\Gamma}^{ j,a,k_3}_c\circ\Phi^{a}_b\circ\tilde{\Delta}^{ j,k_3}_a.
    \end{align}
    Symbolically, it can be written as $\overline{\cJ}=\sum_{k_1}q(k_1)\overline{\cV}^{k_1}[\cI]+\sum_{k_2}q(k_2)\cV^{k_2}[\cI]+\sum_{k_3}q(k_3)\cV^{k_3}[\cI]:=\cW[\cI]$ where the set of instruments $\overline{\cV}^{k_1}[\cI]:=\cV^{k_1}[\hat{\cI}]=\{\cV^{k_1}[\hat{\cI}]^{j,b}=\{\cV^{k_1}[\hat{\cI}]^{j,b}_c=\sum_{a,b}p(a|j)\tilde{\Phi}^{j,b,k_1}_c\circ(\Phi^{a}_b\otimes\mathbbm{I}_{Q^{\prime}})\circ\Phi^{\prime * j,k_1}\}\}$, $\cV^{k_2}[\cI]=\{\cV^{k_2}[\cI]^{j,b}=\{\cV^{k_2}[\cI]^{j,b}_c=\sum_{a,b}\tilde{\Phi}^{*j,b,k_2}_c\circ(\Phi^{a}_b\otimes\mathbbm{I}_{Q})\circ\Phi^{\prime j,k_2}_{a}\}\}$ and $\cV^{k_3}[\cI]=\{\cV^{k_3}[\cI]^{j,b}=\{\cV^{k_3}[\cI]_c^{j,b}=\sum_{a,b}\tilde{\Gamma}^{ j,a,k_3}_c\circ\Phi^{a}_b\circ\tilde{\Delta}^{ j,k_3}_a\}\}$ with $\hat{\cI}=\{\hat{\mathbf{I}}^j=\{\sum_a p(a|j)\Phi^a_b\}\}$. We can also write 
    \begin{align}
        \hat\Gamma_{\overline{\mathbf{J}}^j}(\rho)=&\sum_c\overline{\Phi}_c^j(\rho)\otimes\ket{c}\bra{c}\nonumber\\ =&\sum_c\Big(\sum_{k_1}q(k_1)\sum_{a,b}p(a|j)\tilde{\Phi}^{j,b,k_1}_c\circ(\Phi^{a}_b\otimes\mathbbm{I}_{Q^{\prime}})\circ\Phi^{\prime * j,k_1}\nonumber\\
       &\quad+\sum_{k_2}q(k_2)\sum_{a,b}\tilde{\Phi}^{*j,b,k_2}_c\circ(\Phi^{a}_b\otimes\mathbbm{I}_{Q})\circ\Phi^{\prime j,k_2}_{a}\nonumber\\
       &\quad+\sum_{k_3}q(k_3)\sum_{a,b}\tilde{\Gamma}^{ j,a,k_3}_c\circ\Phi^{a}_b\circ\tilde{\Delta}^{ j,k_3}_a\Big)\otimes\ket{c}\bra{c}\nonumber\\
       =&\sum_{k_1}q(k_1)\hat\Gamma_{\overline{\cV}^{k_1}[\cI]^j}+\sum_{k_2}q(k_2)\hat\Gamma_{\cV^{k_2}[\cI]^j}+\sum_{k_3}q(k_3)\hat\Gamma_{\cV^{k_3}[\cI]^j}.
    \end{align}
      We define a n-dimensional Hilbert space $\cH_{\hat{\cI}}$ with $n=|\hat{\cI}|$. Next, consider the following quantum channels $\Theta_{pre}^{j,k_1}:\cL(\overline{\cH})\rightarrow\cL(\cH\otimes\cH_{\hat{\cI}}\otimes Q^{\prime})$, $\Sigma_{\hat\Gamma_{\hat{\cI}}}:\cL(\cH\otimes\cH_{\hat{\cI}})\rightarrow\cL(\cK\otimes\cH_{\Omega_{\mathbf{I}^a}}\otimes\cH_{\hat{\cI}})$, and $\Theta_{post}^{k_1}:\cL(\cK\otimes\cH_{\Omega_{\mathbf{I}^a}}\otimes\cH_{\hat{\cI}}\otimes Q^{\prime})\rightarrow\cL(\overline{\cK}\otimes\cH_{\Omega_{\overline{\mathbf{J}}^{j,k_1}}})$ such that for all $\rho\in\cL(\overline{\cH})$, for all $\sigma\in\cL(\cH\otimes\cH_{\hat{\cI}})$, and for all $\omega\in\cL(\cK\otimes\cH_{\Omega_{\mathbf{I}^a}}\otimes\cH_{\hat{\cI}}\otimes Q^{\prime})$, we have
     \begin{align}
      \Theta_{pre}^{j,k_1}(\rho)=&\mathbbm{I}_{\cH}\otimes\mathbf{SWAP}_{Q^\prime\leftrightarrow\cH_{\hat{\cI}}}(\Phi^{\prime*j,k_1}(\rho)\otimes\ket{j}\bra{j}),
      \end{align}
      where $\{\ket{j}\}$ is the orthonormal basis of the Hilbert space $\cH_{\hat{\cI}}$,
      \begin{align}
            \Sigma_{\hat\Gamma_{\hat{\cI}}}(\sigma)=&\sum_{j'}\hat{\Gamma}_{\hat{\mathbf{I}}^{j^\prime}}(\bra{j'}\sigma\ket{j'})\otimes\ket{j^\prime}\bra{j^\prime},
        \end{align}
        where $\hat{\Gamma}_{\hat{\mathbf{I}}^{j^\prime}}=\sum_{a^\prime}p(a'|j)\hat{\Gamma}_{\mathbf{I}^{a^\prime}}$, $\hat{\Gamma}_{\mathbf{I}^{a^\prime}}=\sum_{b^\prime}\Phi^{a^\prime}_{b^\prime}\otimes\ket{b^{\prime}}\bra{b^{\prime}}$, $\{\ket{j'}\}$ is the orthonormal basis of the Hilbert space $\cH_{\hat{\cI}}$,  and $\{\ket{b^\prime}\}$ is the orthonormal basis of Hilbert space $\cH_{\Omega_{\hat{\mathbf{I}}^{a^\prime}}}$,
        \begin{align}
            \Theta_{post}^{k_1}(\omega)=\sum_{c,j'',b}\tilde\Phi^{j'',b,k_1}_c(\bra{b,j''}\omega\ket{b,j''})\otimes\ket{c}\bra{c}.
        \end{align}
        Here, $\{\ket{j''}\}$ is the orthonormal basis of the Hilbert space $\cH_{\hat{\cI}}$, $\{\ket{b}\}$ is the orthonormal basis of Hilbert space $\cH_{\Omega_{\mathbf{I}^a}}$, and $\{\ket{c}\}$ is the orthonormal basis of Hilbert space $\cH_{\Omega_{\overline{\mathbf{J}}^{j}}}$. Similarly consider
        \begin{align}
         \Sigma_{\hat\Gamma_{\cI}}(\sigma)=&\sum_{a'}\hat{\Gamma}_{\mathbf{I}^{a^\prime}}(\bra{a^\prime}\sigma\ket{a^\prime})\otimes\ket{a^\prime}\bra{a^\prime}\nonumber\\
      \Theta_{pre}^{j,k_2}(\rho)=&\mathbbm{I}_{\cH}\otimes\mathbf{SWAP}_{Q\leftrightarrow\cH_{\Omega_{\mathbf{J}^{j,k_2}}}}(\sum_{a''}\Phi^{\prime j,k_2}_{a''}(\rho)\otimes\ket{a''}\bra{a''}),\nonumber\\
       \Theta_{post}^{j,k_2}(\omega)=&\sum_{c,a,b}\tilde\Phi^{*j,b,k_2}_c(\bra{b,a}\omega\ket{b,a})\otimes\ket{c}\bra{c},\nonumber\\
       \Theta_{pre}^{j,k_3}(\rho)=&\sum_{a''}\tilde{\Delta}^{ j,k_3}_{a''}(\rho)\otimes\ket{a''}\bra{a''},\nonumber\\
        \Theta_{post}^{j,k_3}(\omega)=&\sum_{c,a,b}\tilde{\Gamma}^{ j,a,k_3}_c(\bra{b,a}\omega\ket{b,a})\otimes\ket{c}\bra{c}.
        \end{align}
        
         It can then be easily verified that       
    \begin{align}
        \hat\Gamma_{\overline{\mathbf{J}}^j}(\rho)=\sum_{k_1} q(k_1)~\Theta^{k_1}_{post} &\circ(\Sigma_{\hat{\cI}}\otimes\mathbbm{I}_{Q^{\prime}})\circ\Theta^{j,k_1}_{pre}\nonumber\\
       +&\sum_{k_2}q(k_2)~\Theta^{j,k_2}_{post} \circ(\Sigma_{\cI}\otimes\mathbbm{I}_{Q})\circ\Theta^{j,k_2}_{pre}\nonumber\\
       &\qquad\quad+\sum_{k_3}q(k_3)~\Theta^{j,k_3}_{post} \circ\Sigma_{\cI}\circ\Theta^{j,k_3}_{pre}.
    \end{align}
    
    Clearly, the quantum channels $\hat\Gamma_{\overline{\cV}^{k_1}[\cI]^j}=\Theta^{k_1}_{post} \circ(\Sigma_{\hat{\cI}}\otimes\mathbbm{I}_{Q^{\prime}})\circ\Theta^{ j,k_1}_{pre}$, $\hat\Gamma_{\cV^{k_2}[\cI]^j}=\Theta^{j,k_2}_{post} \circ(\Sigma_{\cI}\otimes\mathbbm{I}_{Q})\circ\Theta^{ j,k_2}_{pre}$, and $\hat\Gamma_{\cV^{k_3}[\cI]^j}=\Theta^{j,k_3}_{post} \circ\Sigma_{\cI}\circ\Theta^{ j,k_3}_{pre}$ . 
   
   Now consider two sets of instruments $\tilde{\cI}_1$ and $\tilde{\cI}_2$ such that $\tilde{\cI}_i=\cW[\cI_i]$ for $i=1,2$. Then 
   \begin{align}
     \widehat{\cD}(\tilde\cI_1,\tilde\cI_2):=&\widehat{\cD}(\cW[\cI_1],\cW[\cI_2]),\nonumber\\
    \leq &\sum_{k_1}q(k_1)\widehat{\cD}(\overline{\cV}^{k_1}[\cI_1],\overline{\cV}^{k_1}[\cI_2])\nonumber\\
    &\qquad\qquad+\sum_{k_2}q(k_2)\widehat{\cD}(\cV^{k_2}[\cI_1],\cV^{k_2}[\cI_2]),\nonumber\\
    &\qquad\qquad\qquad+\sum_{k_3}q(k_3)\widehat{\cD}(\cV^{k_3}[\cI_1],\cV^{k_3}[\cI_2])\nonumber\\
    \leq &\sum_{k_1}q(k_1)\widehat{\cD}(\cV^{k_1}[\hat{\cI}_1],\cV^{k_1}[\hat{\cI}_2])\nonumber\\
    &\qquad\qquad+\sum_{k_2}q(k_2)\widehat{\cD}(\cV^{k_2}[\cI_1],\cV^{k_2}[\cI_2]),\nonumber\\
    &\qquad\qquad\qquad+\sum_{k_3}q(k_3)\widehat{\cD}(\cV^{k_3}[\cI_1],\cV^{k_3}[\cI_2])\nonumber\\
    \leq &\sum_{k_1}q(k_1)\widehat{\cD}(\hat{\cI}_1,\hat{\cI}_2)+\sum_{k_2}q(k_2)\widehat{\cD}(\cI_1,\cI_2)\nonumber\\
    &\qquad\qquad\qquad+\sum_{k_3}q(k_3)\widehat{\cD}(\cI_1,\cI_2).\label{Eq:Dist_mono_IB_first}
\end{align}
Note that  
\begin{align}
    \cD_{\diamond}(\hat{\mathbf{I}}^j_1,\hat{\mathbf{I}}^j_1)\leq&\sum_a p(a|j)\cD_{\diamond}(\mathbf{I}^a_1,\mathbf{I}^a_2)\nonumber\\
    \leq&\max_a \cD_{\diamond}(\mathbf{I}^a_1,\mathbf{I}^a_2)\nonumber\\
    =&\widehat{\cD}(\cI_1,\cI_2).
\end{align}
Now as the above equation is true for any $j$ it implies that
\begin{align}
    \widehat{\cD}(\hat{\cI}_1,\hat{\cI}_2)=&\max_j\cD_{\diamond}(\hat{\mathbf{I}}^j_1,\hat{\mathbf{I}}^j_1)\nonumber\\
    \leq&\widehat{\cD}(\cI_1,\cI_2)
\end{align}
 Using this in Eq. \eqref{Eq:Dist_mono_IB_first}, we get
 \begin{align}
     \widehat{\cD}(\tilde\cI_1,\tilde\cI_2)\leq &\sum_{k_1}q(k_1)\widehat{\cD}(\hat{\cI}_1,\hat{\cI}_2)+\sum_{k_2}q(k_2)\widehat{\cD}(\cI_1,\cI_2)\nonumber\\
    &\qquad\qquad\qquad+\sum_{k_3}q(k_3)\widehat{\cD}(\cI_1,\cI_2)\nonumber\\
    \leq&\sum_{k_1}q(k_1)\widehat{\cD}(\cI_1,\cI_2)+\sum_{k_2}q(k_2)\widehat{\cD}(\cI_1,\cI_2)\nonumber\\
    &\qquad\qquad\qquad+\sum_{k_3}q(k_3)\widehat{\cD}(\cI_1,\cI_2)\nonumber\\
    =&\widehat{\cD}(\cI_1,\cI_2)
 \end{align}

Hence $\widehat{\cD}$ is monotonically non-increasing under the free transformations of incompatibility preservability.
\end{proof}

\subsection{Proof of Theorem \ref{Th:Dist_mono_SIP}}\label{App:Dist_mono_SIP}
\begin{proof}
    The free transformation of the resource theory of strong incompatibility preservability is of the form
     \begin{align}
       \overline{\Phi}^j_c=\sum_{k_1}q(k_1)\sum_{a,b}&\tilde{\Phi}^{j,b,k_1}_c\circ(\Phi^{a}_b\otimes\mathbbm{I}_{Q})\circ\Phi^{\prime  j,k_1}_{a}\nonumber\\
       &+\sum_{k_2}q(k_2)\sum_{a,b}\tilde{\Gamma}_c^{j,a,k_2}\circ\Phi^a_b\circ\tilde{\Delta}_a^{j,k_2}.
    \end{align}
    Symbolically, we represent it as $\overline{\cJ}=\sum_{k_1}q(k_1)\cV^{k_1}[\cI]+\sum_{k_2}q(k_2)\cV^{k_2}[\cI]:=\cW[\cI]$ where the set of instruments $\cV^{k_1}[\cI]=\{\cV^{k_1}[\cI]^{j,b}=\{\cV^{k_1}[\cI]^{j,b}_c=\sum_{a,b}\tilde{\Phi}^{j,b,k_1}_c\circ(\Phi^{a}_b\otimes\mathbbm{I}_{Q})\circ\Phi^{\prime j,k_1}_{a}\}\}$ and $\cV^{k_2}[\cI]=\{\cV^{k_2}[\cI]^{j}=\{\sum_{a,b}\tilde{\Gamma}_c^{j,a,k_2}\circ\Phi^a_b\circ\tilde{\Delta}_a^{j,k_2}\}\}$. We can also write 
    \begin{align}
        \hat\Gamma_{\overline{\mathbf{J}}^j}(\rho)=&\sum_c\overline{\Phi}_c^j(\rho)\otimes\ket{c}\bra{c}\nonumber\\ =&\sum_c\Big(\sum_{k_1}q(k_1)\sum_{a,b}\tilde{\Phi}^{j,b,k_1}_c\circ(\Phi^{a}_b\otimes\mathbbm{I}_{Q})\circ\Phi^{\prime  j,k_1}_{a}\nonumber\\
       &+\sum_{k_2}q(k_2)\sum_{a,b}\tilde{\Gamma}_c^{j,a,k_2}\circ\Phi^a_b\circ\tilde{\Delta}_a^{j,k_2}\Big)\otimes\ket{c}\bra{c}\nonumber\\
        =&\sum_{k_1}q(k_1)\hat\Gamma_{\cV^{k_1}[\cI]^j}+\sum_{k_2}q(k_2)\hat\Gamma_{\cV^{k_2}[\cI]^j}
    \end{align}
     Consider the following quantum channels $\Theta_{pre}^{j,k_1}:\cL(\overline{\cH})\rightarrow\cL(\cH\otimes\cH_{\Omega_{\mathbf{J}^{j,k_1}}}\otimes Q)$, $\Sigma_{\hat\Gamma_{\cI}}:\cL(\cH\otimes\cH_{\Omega_{\mathbf{J}^{j,k_1}}})\rightarrow\cL(\cK\otimes\cH_{\Omega_{\mathbf{I}^a}}\otimes\cH_{\Omega_{\mathbf{J}^{j,k_1}}})$, and $\Theta_{post}^{j,k_1}:\cL(\cK\otimes\cH_{\Omega_{\mathbf{I}^a}}\otimes\cH_{\Omega_{\mathbf{J}^{j,k_1}}}\otimes\cH_Q)\rightarrow\cL(\overline{\cK}\otimes\cH_{\Omega_{\overline{\mathbf{J}}^{j,k_1}}})$ such that for all $\rho\in\cL(\overline{\cH})$, for all $\sigma\in\cL(\cH\otimes\cH_{\Omega_{\mathbf{J}^{j,k_1}}})$, and for all $\omega\in\cL(\cK\otimes\cH_{\Omega_{\mathbf{I}^a}}\otimes\cH_{\Omega_{\mathbf{J}^{j,k_1}}}\otimes\cH_Q)$ we have
     \begin{align}
      \Theta_{pre}^{j,k_1}(\rho)=\mathbbm{I}_{\cH}\otimes\mathbf{SWAP}_{Q\leftrightarrow\cH_{\Omega_{\mathbf{J}^{j,k_1}}}}(\sum_{a''}\Phi^{\prime j,k_1}_{a''}(\rho)\otimes\ket{a''}\bra{a''}),
        \end{align}
        where $\{\ket{a''}\}$ is the orthonormal basis of Hilbert space $\cH_{\Omega_{\mathbf{J}^{j,k_1}}}$,
         \begin{align}
            \Sigma_{\hat\Gamma_{\cI}}(\sigma)=\sum_{a'}\hat{\Gamma}_{\mathbf{I}^{a^\prime}}(\bra{a^\prime}\sigma\ket{a^\prime})\otimes\ket{a^\prime}\bra{a^\prime},
        \end{align}
        where $\{\ket{a^\prime}\}$ is the orthonormal basis of Hilbert space $\cH_{\Omega_{\mathbf{J}^{j,k_1}}}$, $\hat{\Gamma}_{\mathbf{I}^{a^\prime}}=\sum_{b^\prime}\Phi^{a^\prime}_{b^\prime}\otimes\ket{b^{\prime}}\bra{b^{\prime}}$, and $\{\ket{b^\prime}\}$ is the orthonormal basis of Hilbert space $\cH_{\Omega_{\mathbf{I}^{a^\prime}}}$,
        \begin{align}
            &\Theta_{post}^{j,k_1}(\omega)=\sum_{c,a,b}\tilde\Phi^{j,b,k_1}_c(\bra{b,a}\omega\ket{b,a})\otimes\ket{c}\bra{c},
        \end{align}
        where $\{\ket{a}\}$ is the orthonormal basis of Hilbert space $\cH_{\Omega_{\mathbf{J}^{j,k_1}}}$, 
        $\{\ket{b}\}$ is the orthonormal basis of Hilbert space $\cH_{\Omega_{\mathbf{I}^a}}$, and $\{\ket{c}\}$ is the orthonormal basis of Hilbert space $\cH_{\Omega_{\overline{\mathbf{J}}^j}}$. Similarly considering
        \begin{align}
            \Theta_{pre}^{j,k_2}(\rho)=&\sum_{a''}\tilde{\Delta}^{ j,k_2}_{a''}(\rho)\otimes\ket{a''}\bra{a''},\nonumber\\
        \Theta_{post}^{j,k_2}(\omega)=&\sum_{c,a,b}\tilde{\Gamma}^{ j,a,k_3}_c(\bra{b,a}\omega\ket{b,a})\otimes\ket{c}\bra{c}.
        \end{align}
       Then, it can be easily verified that       
    \begin{align}
        \hat\Gamma_{\overline{\mathbf{J}}^j}(\rho)=\sum_{k_1} q(k_1)~\Theta^{j,k_1}_{post} &\circ(\Sigma_{\cI}\otimes\mathbbm{I}_{Q})\circ\Theta^{j,k_1}_{pre}\nonumber\\
       +&\sum_{k_2}q(k_2)~\Theta^{j,k_2}_{post} \circ\Sigma_{\cI}\circ\Theta^{j,k_2}_{pre}
    \end{align}
     Consider two quantum instruments of the form $\tilde{\cI}_1=\cW[\cI_1]$ and $\tilde{\cI}_2=\cW[\cI_2]$, then by using Eq. (\ref{Eq:supermap_contractive_distance}), we can write
    \begin{align}
    \widehat{\cD}(\tilde\cI_1,\tilde\cI_2):=&\widehat{\cD}(\cW[\cI_1],\cW[\cI_2]),\nonumber\\
    \leq &\sum_{k_1}q(k_1)\widehat{\cD}(\cV^{k_1}[\cI_1],\cV^{k_1}[\cI_2])\nonumber\\
    &\qquad\qquad+\sum_{k_2}q(k_2)\widehat{\cD}(\cV^{k_2}[\cI_1],\cV^{k_2}[\cI_2]),\nonumber\\
    \leq &\sum_{k_1}q(k_1)\widehat{\cD}(\cI_1,\cI_2)+\sum_{k_2}q(k_2)\widehat{\cD}(\cI_1,\cI_2)\nonumber\\
    = & \widehat{\cD}(\cI_1,\cI_2).
\end{align}
Hence $\widehat{\cD}$ is monotonically non-increasing under the free transformations of strong incompatibility preservability.
\end{proof}

\subsection{Proof of Theorem \ref{Th:Dist_mono_TC}}\label{App:Dist_mono_TC}
\begin{proof}
    Symbolically, the transformation in Eq. \eqref{Eq:free_operation_trad_comp} can be written as $\tilde\cJ=\cV[\cJ]$. Here $\mathbf{F} \in \mathscr{C}(\cH,\cH\otimes Q)$ and $\mathbf{K}^\lambda\in\mathscr{C}(\cK\otimes Q,\tilde\cK)$ with $ Q$ being an arbitrary auxiliary Hilbert space with $\{\ket{\lambda}\}$ being an orthonormal basis spanning $\cH_{\Lambda}$. We can also write
    \begin{align}
        \hat\Gamma_{\tilde{\mathbf{J}}^j}(\rho)=&\sum_b\tilde{\Phi}_b^j(\rho)\otimes\ket{b}\bra{b}\nonumber\\ =&\sum_b\sum_{\lambda,i,a}p(b\vert i, j,\lambda,a)q(i\vert j,\lambda)\mathbf{K}^\lambda\circ(\Phi_a^i\otimes \mathbbm{I})\nonumber\\ &\qquad\qquad\qquad\qquad\qquad\qquad\circ\mathbf{F}(\rho)\otimes\ket{b}\bra{b}\nonumber\\
        =&\hat\Gamma_{\tilde{\cV}[\cI]^j}
    \end{align}
    Consider the following quantum channels $\Theta_{pre}^j:\cL(\cH)\rightarrow\cL(\cH\otimes\cH_I\otimes Q\otimes\cH_{\Lambda})$, $\Sigma_{\hat\Gamma_{\cJ}}:\cL(\cH\otimes\cH_I)\rightarrow\cL(\cK\otimes\cH_{\Omega_{\mathbf{J}^i}}\otimes\cH_I)$, and $\Theta_{post}^j:\cL(\cK\otimes\cH_{\Omega_{\mathbf{J}^i}}\otimes\cH_I\otimes Q\otimes\cH_{\Lambda})\rightarrow\cL(\tilde\cK\otimes\cH_{\Omega_{\mathbf{\tilde{J}}^i}})$ such that for all $\rho\in\cL(\cH)$, for all $\sigma\in\cL(\cH\otimes\cH_I)$, and for all $\omega\in\cL(\cK\otimes\cH_{\Omega_{\mathbf{J}^i}}\otimes\cH_I\otimes Q\otimes\cH_{\Lambda})$ we have
      \begin{align}
      \Theta_{pre}^j(\rho)=\mathbbm{I}_{\cH}&\otimes\mathbf{SWAP}_{ Q\leftrightarrow\cH_I}\otimes\mathbbm{I}_{\cH_{\Lambda}}\nonumber\\
      &(\mathbf{F}(\rho)\otimes\sum_{i'',\lambda'}q(i''\vert j,\lambda')\ket{i''}\bra{i''}\otimes\ket{\lambda'}\bra{\lambda'}),
        \end{align}
        where $\{\ket{i''}\}$ is the orthonormal basis of Hilbert space $\cH_I$, $\{\ket{\lambda'}\}$ is the orthonormal basis of Hilbert space $\cH_{\Lambda}$.
        \begin{align}
            \Sigma_{\hat\Gamma_{\cJ}}(\sigma)=\sum_{i'}\hat{\Gamma}_{\mathbf{J}^i}(\bra{i'}\sigma\ket{i'})\otimes\ket{i'}\bra{i'}
        \end{align}
        where $\{\ket{i'}\}$ is the orthonormal basis of Hilbert space $\cH_I$, $\{\ket{a'}\}$ is the orthonormal basis of Hilbert space $\cH_{\Omega_{\mathbf{J}^i}}$ and $\mathbbm{I_{\overline{\mathfrak{R}}}}=\mathbbm{I}_{ Q}\otimes\mathbbm{I}_{\Lambda}$.
        \begin{align}
            &\Theta_{post}^j(\omega)=\sum_{b,a,i,\lambda}p(b\vert a,i,\lambda,j)\cK^\lambda(\bra{a,i,\lambda}\omega\ket{a,i,\lambda})\otimes\ket{b}\bra{b}.
        \end{align}
        Here $\{\ket{a}\}$ is the orthonormal basis of Hilbert space $\cH_{\Omega_{\mathbf{J}^i}}$, 
        $\{\ket{b}\}$ is the orthonormal basis of Hilbert space $\cH_{\Omega_{\mathbf{\tilde{J}}^i}}$, $\{\ket{i}\}$ is the orthonormal basis of Hilbert space $\cH_I$, $\{\ket{\lambda}\}$ is the orthonormal basis of Hilbert space $\cH_{\Lambda}$. 
      It can be easily verified that       
    \begin{align}
        \hat\Gamma_{\tilde{\mathbf{J}}^j}(\rho)=\Theta_{post}^j\circ(\Sigma_{\hat\Gamma_{\cJ}}\otimes\mathbbm{I}_{\overline{\mathfrak{R}}})\circ\Theta_{pre}^j(\rho)
    \end{align}
    which is of the form of Eq.(\ref{gensupermap}). Consider two quantum instruments of the form $\tilde{\cI}_1=\cV[\cI_1]$ and $\tilde{\cI}_2=\cV[\cI_2]$, then by using Eq. (\ref{Eq:supermap_contractive_distance}), we can write
    \begin{align}
    \widehat{\cD}(\tilde\cI_1,\tilde\cI_2)=&\widehat{\cD}(\cV[\cI_1],\cV[\cI_2]),\nonumber\\
    \leq & \widehat{\cD}(\cI_1,\cI_2).
\end{align}
Hence $\widehat{\cD}$ is monotonically non-increasing under the free transformations of the resource theory of traditional incompatibility.
\end{proof}

\subsection{Proof of Theorem \ref{Th:Dist_mon_PC}}\label{App:Dist_mon_PC}
\begin{proof}
    According to Theorem \ref{Th:free_op_para_comp}, the free transformation of parallel compatibility is
     \begin{align}
    \hat\Gamma_{\overline{\mathbf{J}}^j}(\rho)=&\sum_c\overline{\Phi}_c^j(\rho)\otimes\ket{c}\bra{c}\nonumber\\ =&\sum_c\Big(\sum_{k_1}q(k_1)\sum_{l,b}\tilde{\Phi}^{j,b,k_1}_c\circ(\Phi^{*l}_b\otimes\mathbbm{I}_{Q})\circ\Phi^{\prime  j,k_1}_{l}\nonumber\\
       &+\sum_{k_2}q(k_2)\sum_{b}\tilde{\Gamma}_c^{j,b,k_2}\circ\hat{\Phi}^{j}_b\circ\tilde{\Delta}^{j,k_2}\Big)\otimes \ket{c}\bra{c}\nonumber\\
       =&\sum_{k_1}q(k_1)\hat\Gamma_{\overline{\cV}^{k_1}[\cI]^j}+\sum_{k_2}q(k_2)\hat\Gamma_{\overline{\cV}^{k_2}[\cI]^j}
\end{align}
Here $\overline{\cV}^{k_1}[\cI]^j:=\hat\Gamma_{\cV^{k_1}[\cI^*]^j}=\sum_c\sum_{l,b}\tilde{\Phi}^{j,b,k_1}_c\circ(\Phi^{*l}_b\otimes\mathbbm{I}_{Q})\circ\Phi^{\prime  j,k_1}_{l}\otimes\ket{c}\bra{c}$ and $\overline{\cV}^{k_2}[\cI]^j:=\hat\Gamma_{\cV^{k_2}[\hat{\cI}]^j}=\sum_c\sum_{b}\tilde{\Gamma}_c^{j,b,k_2}\circ\hat{\Phi}^{j}_b\circ\tilde{\Delta}^{j,k_2}\otimes\ket{c}\bra{c}$.

   Symbolically, we represent $\overline{\cJ}=\sum_{k_1}q(k_1)\cV^{k_1}[\cI^*]+\sum_{k_2}q(k_2)\cV^{k_2}[\hat{\cI}]:=\cW[\cI]$ where the set of instruments $\cV^{k_1}[\cI^*]=\{\cV^{k_1}[\cI^*]^{j,b}=\{\cV^{k_1}[\cI^*]^{j,b}_c=\sum_{l,b}\tilde{\Phi}^{j,b,k_1}_c\circ(\Phi^{*l}_b\otimes\mathbbm{I}_{Q})\circ\Phi^{\prime j,k_1}_{l}\}\}$ and $\cV^{k_2}[\hat{\cI}]=\{\cV^{k_2}[\hat{\cI}]^{j}=\{\sum_{b}\tilde{\Gamma}_c^{j,b,k_2}\circ\hat{\Phi}^{j}_b\circ\tilde{\Delta}^{j,k_2}\}\}$.
   
     Consider the following quantum channels $\Theta_{pre}^{j,k_1}:\cL(\overline{\cH})\rightarrow\cL(\cH\otimes\cH_{\Omega_{\mathbf{J}^{j,k_1}}}\otimes Q)$, $\Sigma_{\hat\Gamma_{\cI}}:\cL(\cH\otimes\cH_{\Omega_{\mathbf{J}^{j,k_1}}})\rightarrow\cL(\cK\otimes\cH_{\Omega_{\mathbf{I^*}^l}}\otimes\cH_{\Omega_{\mathbf{J}^j}})$, and $\Theta_{post}^{j,k_1}:\cL(\cK\otimes\cH_{\Omega_{\mathbf{I^*}^l}}\otimes\cH_{\Omega_{\mathbf{J}^{j,k_1}}}\otimes\cH_Q)\rightarrow\cL(\overline{\cK}\otimes\cH_{\Omega_{\overline{\mathbf{J}}^j}})$ such that for all $\rho\in\cL(\overline{\cH})$, for all $\sigma\in\cL(\cH\otimes\cH_{\Omega_{\mathbf{J}^{j,k_1}}})$, and for all $\omega\in\cL(\cK\otimes\cH_{\Omega_{\mathbf{I^*}^l}}\otimes\cH_{\Omega_{\mathbf{J}^{j,k_1}}}\otimes\cH_Q)$ we have
     \begin{align}
      \Theta_{pre}^j(\rho)=\mathbbm{I}_{\overline{\cH}}\otimes\mathbf{SWAP}_{Q\leftrightarrow\cH_{\Omega_{\mathbf{J}^j}}}(\sum_{l''}\Phi'^j_{l''}(\rho)\otimes\ket{l''}\bra{l''}),
        \end{align}
        where $\{\ket{l''}\}$ is the orthonormal basis of Hilbert space $\cH_{\Omega_{\mathbf{J}^j}}$.
         \begin{align}
            \Sigma_{\hat\Gamma_{\cI^*}}(\sigma)=\sum_{l'}\hat{\Gamma}_{\mathbf{I}^{*l'}}(\bra{l'}\sigma\ket{l'})\otimes\ket{l'}\bra{l'},
        \end{align}
        where $\{\ket{l'}\}$ is the orthonormal basis of Hilbert space $\cH_{\Omega_{\mathbf{J}^j}}$, $\hat{\Gamma}_{\mathbf{I}^{*l^\prime}}=\sum_{b^\prime}\Phi^{*l^\prime}_{b^\prime}\otimes\ket{b^{\prime}}\bra{b^{\prime}}$, and $\{\ket{b'}\}$ is the orthonormal basis of Hilbert space $\cH_{\Omega_{\mathbf{I}^{*l'}}}$.
        \begin{align}
            &\Theta_{post}^j(\omega)=\sum_{c,l,b}\tilde\Phi^{l,b}_c(\bra{b,l}\omega\ket{b,l})\otimes\ket{c}\bra{c}.
        \end{align}
        Here $\{\ket{l}\}$ is the orthonormal basis of Hilbert space $\cH_{\Omega_{\mathbf{J}^j}}$, 
        $\{\ket{b}\}$ is the orthonormal basis of Hilbert space $\cH_{\Omega_{\mathbf{I}^{*l}}}$, and $\{\ket{c}\}$ is the orthonormal basis of Hilbert space $\cH_{\Omega_{\overline{\mathbf{J}}^j}}$.

        Similarly, by defining n-dimensional Hilbert space $\cH_{\hat{\cI}}$ with $n=|\hat{\cI}|$ we consider the quantum channels $\Theta_{pre}^{j,k_2}:\cL(\overline{\cH})\rightarrow\cL(\cH\otimes\cH_{\hat{\cI}})$, $\Sigma_{\hat{\Gamma}_{\hat{\cI}}}:\cL(\cH\otimes\cH_{\hat{\cI}})\rightarrow\cL(\cK\otimes\cH_{\Omega_{\hat{\mathbf{I}}^j}}\otimes\cH_{\hat{\cI}})$, and $\Theta_{post}^{k_2}:\cL(\cK\otimes\cH_{\Omega_{\hat{\mathbf{I}}^j}}\otimes\cH_{\hat{\cI}})\rightarrow\cL(\overline{\cK}\otimes\cH_{\Omega_{\overline{\mathbf{J}}^j}})$ such that for all $\rho\in\cL(\cH)$, $\sigma\in\cL(\cH\otimes\cH_{\hat{\cI}})$ and $\omega\in\cL(\cK\otimes\cH_{\Omega_{\hat{\mathbf{I}}^j}}\otimes\cH_{\hat{\cI}})$, we have
        \begin{align}
    \Theta_{pre}^{j,k_2}(\rho)=\tilde{\Delta}^{j,k_2}(\rho)\otimes\ket{j}\bra{j},
    \end{align}
    where $\{\ket{j}\}$ is the orthogonal basis of the Hilbert space $\cH_{\hat{\cI}}$,
    \begin{align}
    \Sigma_{\hat{\Gamma}_{\hat{\cI}}}(\sigma)=\sum_{j'}\hat{\Gamma}_{\hat{\mathbf{I}}^{j'}}(\bra{j'}\sigma\ket{j'})\otimes\ket{j'}\bra{j'},
    \end{align}
    where $\{\ket{j'}\}$ is the orthogonal basis of the Hilbert space $\cH_{\hat{\cI}}$, $\hat{\Gamma}_{\hat{\mathbf{I}}^{j'}}=\sum_{b'}\hat{\Phi}^{j'}_{b'}\otimes\ket{b'}\bra{b'}$, and $\{\ket{b'}\}$ is the orthogonal basis of the Hilbert space $\cH_{\Omega_{\hat{\mathbf{I}}^{j'}}}$,
    \begin{align}
    \Theta_{post}^{k_2}(\omega)=\sum_{j'',c,b}\tilde\Gamma^{j'',b,k_2}_c(\bra{b,j''}\omega\ket{b,j''})\otimes\ket{c}\bra{c}.
\end{align}
Here $\{\ket{j''}\}$ is the orthogonal basis of the Hilbert space $\cH_{\hat{\cI}}$, $\{\ket{b}\}$ is the orthogonal basis of the Hilbert space $\cH_{\Omega_{\hat{\mathbf{I}}^{j''}}}$ and $\{\ket{c}\}$ is the orthogonal basis of the Hilbert space $\cH_{\Omega_{\overline{\mathbf{J}}^j}}$. Then, it can be easily verified that       
    \begin{align}
        \hat\Gamma_{\overline{\mathbf{J}}^j}(\rho)=\sum_{k_1}q(k_1)\Theta_{post}^{j,k_1}\circ(\Sigma_{\hat\Gamma_{\cI^*}}\otimes&\mathbbm{I}_{\overline{\mathfrak{R}}})\circ\Theta_{pre}^{j,k_1}(\rho)\nonumber\\
        &+\sum_{k_2}q(k_2)\Theta^{k_2}_{post} \circ\Sigma_{\hat\Gamma_{\hat{\cI}}}\circ\Theta^{j,k_2}_{pre}
    \end{align}
    Consider two quantum instruments of the form $\tilde{\cI}_1=\cW[\cI_1]$ and $\tilde{\cI}_2=\cW[\cI_2]$, then by using Eq. (\ref{Eq:supermap_contractive_distance}), we can write
    \begin{align}
    \widehat{\cD}(\tilde\cI_1,\tilde\cI_2):=&\widehat{\cD}(\cW[\cI_1],\cW[\cI_2]),\nonumber\\
    \leq &\sum_{k_1}q(k_1)\widehat{\cD}(\overline{\cV}^{k_1}[\cI_1],\overline{\cV}^{k_1}[\cI_2])\nonumber\\
    &\qquad\qquad+\sum_{k_2}q(k_2)\widehat{\cD}(\overline{\cV}^{k_2}[\hat{\cI}_1],\overline{\cV}^{k_2}[\cI_2]),\nonumber\\
    = &\sum_{k_1}q(k_1)\widehat{\cD}(\cV^{k_1}[\cI^*_1],\cV^{k_1}[\cI^*_2])\nonumber\\
    &\qquad\qquad+\sum_{k_2}q(k_2)\widehat{\cD}(\cV^{k_2}[\hat{\cI}_1],\cV^{k_2}[\hat{\cI}_2]),\nonumber\\
    \leq &\sum_{k_1}q(k_1)\widehat{\cD}(\cI^*_1,\cI^*_2)+\sum_{k_2}q(k_2)\widehat{\cD}(\hat{\cI}_1,\hat{\cI}_2).
    \label{Eq:parallel_D_monotonicity} 
\end{align}
Note that $\hat{\Gamma}_{\hat{\mathbf{I}}_1^{j}}=\hat{\Gamma}_{\hat{\mathbf{I}}_2^{j}}$ for all $j\in\{m+1,m+2,\ldots,n\}$ (as same free transformation $\overline{\cV}^{k_2}$ is applied on both sets $\cI_1$ and $\cI_2$) implies $\cD_{\Diamond}(\hat{\Gamma}_{\hat{\mathbf{I}}_1^{j}},\hat{\Gamma}_{\hat{\mathbf{I}}_2^{j}})=0$ for all $j\in\{m+1,m+2,\ldots,n\}$ and same permutation $\pi$ is applied to both $\cI_1$ and $\cI_2$ (again as same free transformation $\overline{\cV}^{k_2}$ is applied on both sets $\cI_1$ and $\cI_2$). Then
\begin{align}
   \widehat{\cD}(\hat{\cI}_1,\hat{\cI}_2)=&\max_{j\in\{1,\ldots,n\}}\cD_{\Diamond}(\hat{\Gamma}_{\hat{\mathbf{I}}_1^{j}},\hat{\Gamma}_{\hat{\mathbf{I}}_2^{j}})\nonumber\\
   =&\max_{j\in\{1,\ldots,m\}}\cD_{\Diamond}(\hat{\Gamma}_{\hat{\mathbf{I}}_1^{j}},\hat{\Gamma}_{\hat{\mathbf{I}}_2^{j}})\nonumber\\
   &=\widehat{\cD}(\cI_1,\cI_2) 
   \label{Eq:D_Icap_I_equality}
\end{align}
Similarly, we can show that $\widehat{\cD}(\cI^*_1,\cI^*_2)=\widehat{\cD}(\cI_1,\cI_2)$. Thus, using Eq. \eqref{Eq:parallel_D_monotonicity} and Eq. \eqref{Eq:D_Icap_I_equality} we obtain
\begin{align}
    \widehat{\cD}(\tilde\cI_1,\tilde\cI_2)\leq &\sum_{k_1}q(k_1)\widehat{\cD}(\cI^*_1,\cI^*_2)+\sum_{k_2}q(k_2)\widehat{\cD}(\hat{\cI}_1,\hat{\cI}_2)\nonumber\\
    = &\sum_{k_1}q(k_1)\widehat{\cD}(\cI_1,\cI_2)+\sum_{k_2}q(k_2)\widehat{\cD}(\cI_1,\cI_2)\nonumber\\
     = & \widehat{\cD}(\cI_1,\cI_2) 
\end{align}
Hence $\widehat{\cD}$ is monotonically non-increasing under the free transformations of parallel compatibility.
\end{proof}

\section{Clarification on Remark \ref{Re:EBC_compact}}\label{App:EBC_compact}
Let us consider two instruments $\mathbf{I}=\{\Lambda_a\}_{a=1}^{\Omega_{\mathbf{I}}}\in\mathscr{I}(\cH_A,\cK_A)$ and $\mathbf{J}=\{\Phi_b\}_{b=1}^{\Omega_{\mathbf{J}}}\in\mathscr{F}(\cH_A,\cK_A)$ with $\Omega_{\mathbf{J}} \geq \Omega_{\mathbf{I}}$. The distance between these two instruments can be written as
\begin{align}
    \cD_{\Diamond}(\mathbf{I}, \mathbf{J})=&\cD_{\Diamond}(\hat{\Gamma}_{\mathbf{I}}, \hat{\Gamma}_{\mathbf{J}})\nonumber\\
    =&\Big|\Big|\sum_{a=1}^{\Omega_{\mathbf{I}}}\Lambda_a\otimes\ket{a}\bra{a}-\sum_{b=1}^{\Omega_{\mathbf{J}}}\Phi_b\otimes\ket{b}\bra{b}\Big|\Big|_{\Diamond}\nonumber\\
    =&\max_{\rho_{AB}\in\cS(\cH_A\otimes\cH_B)}\Big|\Big|\sum_{a=1}^{\Omega_{\mathbf{I}}}(\Lambda_a\otimes\mathbbm{I}_{\cH_B})(\rho_{AB})\otimes\ket{a}\bra{a}\nonumber\\
    &\qquad-\sum_{b=1}^{\Omega_{\mathbf{J}}}(\Phi_b\otimes\mathbbm{I}_{\cH_B})(\rho_{AB})\otimes\ket{b}\bra{b}~\Big|\Big|_{1}\nonumber\\
    =&\max_{\rho_{AB}\in\cS(\cH_A\otimes\cH_B)}\Big|\Big|\sum_{a=1}^{\Omega_{\mathbf{I}}}((\Lambda_a-\Phi_a)\otimes\mathbbm{I}_{\cH_B})(\rho_{AB})\otimes\ket{a}\bra{a}\nonumber\\
    &\qquad-\sum_{b={\Omega_{\mathbf{I}}}+1}^{\Omega_{\mathbf{J}}}(\Phi_b\otimes\mathbbm{I}_{\cH_B})(\rho_{AB})\otimes\ket{b}\bra{b}\Big|\Big|_{1}\nonumber\\
    =&\max_{\rho_{AB}\in\cS(\cH_A\otimes\cH_B)}\Bigg(\sum_{a=1}^{\Omega_{\mathbf{I}}}\Big|\Big|((\Lambda_a-\Phi_a)\otimes\mathbbm{I}_{\cH_B})(\rho_{AB})\Big|\Big|_1\nonumber\\
    &\qquad+\sum_{b={\Omega_{\mathbf{I}}}+1}^{\Omega_{\mathbf{J}}}\Big|\Big|(\Phi_b\otimes\mathbbm{I}_{\cH_B})(\rho_{AB})\Big|\Big|_{1}\Bigg)\nonumber\\
    =&\max_{\rho_{AB}\in\cS(\cH_A\otimes\cH_B)}\Bigg(\sum_{a=1}^{\Omega_{\mathbf{I}}}\Big|\Big|((\Lambda_a-\Phi_a)\otimes\mathbbm{I}_{\cH_B})(\rho_{AB})\Big|\Big|_1\nonumber\\
    &\qquad+\sum_{b={\Omega_{\mathbf{I}}}+1}^{\Omega_{\mathbf{J}}}\tr[(\Phi_b\otimes\mathbbm{I}_{\cH_B})(\rho_{AB})]\Bigg)\nonumber\\
    =&\max_{\rho_{AB}\in\cS(\cH_A\otimes\cH_B)}\Bigg(\sum_{a=1}^{\Omega_{\mathbf{I}}}\Big|\Big|((\Lambda_a-\Phi_a)\otimes\mathbbm{I}_{\cH_B})(\rho_{AB})\Big|\Big|_1\nonumber\\
     &\qquad+\Big|\Big|\sum_{b={\Omega_{\mathbf{I}}}+1}^{\Omega_{\mathbf{J}}}(\Phi_b\otimes\mathbbm{I}_{\cH_B})(\rho_{AB})\Big|\Big|_{1}\Bigg)\nonumber
    \end{align}
    \begin{align}
      \geq&\max_{\rho_{AB}\in\cS(\cH_A\otimes\cH_B)}\Bigg(\sum_{a=1}^{\Omega_{\mathbf{I}}-1}\Big|\Big|((\Lambda_a-\Phi_a)\otimes\mathbbm{I}_{\cH_B})(\rho_{AB})\Big|\Big|_1\nonumber\\
    &\qquad+\Big|\Big|((\Lambda_{\Omega_{\mathbf{I}}}-\Phi_{\Omega_{\mathbf{I}}}-\sum_{b=\Omega_{\mathbf{I}}+1}^{\Omega_{\mathbf{J}}}\Phi_b)\otimes\mathbbm{I}_{\cH_B})(\rho_{AB})\Big|\Big|_{1}\Bigg)\label{Eq:Dist_n_bound}
\end{align}
Here in the sixth line, we have used the fact that $(\Phi_b\otimes\mathbbm{I}_{\cH_B})(\rho_{AB})$ is positive $\forall b$, in the seventh line we have used the linearity of the trace, and the eighth line comes from the fact that the trace norm satisfies the triangle inequality. Next, we define an instrument $\mathbf{J}^{\prime}=\{\Phi^{\prime}_b\}_{b=1}^{\Omega_{\mathbf{I}}}$ such that
\begin{align}
    \Phi^{\prime}_b=&\Phi_b, \qquad\qquad\qquad \text{for}~b=1,2,\ldots,{\Omega_{\mathbf{I}}}-1\\
    \Phi^{\prime}_{\Omega_{\mathbf{I}}}=&\Phi_{\Omega_{\mathbf{I}}}+\sum_{b=n+1}^{_{\Omega_{\mathbf{J}}}}\Phi_b
\end{align}
Thus, from Eq. \ref{Eq:Dist_n_bound} we can infer that
\begin{align}
    \cD_{\Diamond}(\mathbf{I}, \mathbf{J})\geq&\max_{\rho_{AB}\in\cS(\cH_A\otimes\cH_B)}\Big(\sum_{a=1}^{\Omega_{\mathbf{I}}-1}\Big|\Big|(\Lambda_a-\Phi_a)\otimes\mathbbm{I}_{\cH_B}~\rho_{AB}\Big|\Big|_1\nonumber\\
    &\qquad\quad\quad+\Big|\Big|\Lambda_{\Omega_{\mathbf{I}}}-\Phi_{\Omega_{\mathbf{I}}}-\sum_{b=\Omega_{\mathbf{I}}+1}^{\Omega_{\mathbf{J}}}\Phi_b)\otimes\mathbbm{I}_{\cH_B}~\rho_{AB}\Big|\Big|_{1}\Big)\nonumber\\
    =&\max_{\rho_{AB}\in\cS(\cH_A\otimes\cH_B)}\Big(\sum_{a=1}^{\Omega_{\mathbf{I}}}\Big|\Big|(\Lambda_a-\Phi^{\prime}_a)\otimes\mathbbm{I}_{\cH_B}~\rho_{AB}\Big|\Big|_1\Big)\nonumber\\
    =&\max_{\rho_{AB}\in\cS(\cH_A\otimes\cH_B)}\Big|\Big|\Big(\sum_{a=1}^{\Omega_{\mathbf{I}}}(\Lambda_a-\Phi^{\prime}_a)\otimes\ket{a}\bra{a}\Big)\otimes\mathbbm{I}_{\cH_B}~\rho_{AB}\Big|\Big|_1\nonumber\\
    =&\Big|\Big|\Big(\sum_{a=1}^{\Omega_{\mathbf{I}}}(\Lambda_a-\Phi^{\prime}_a)\otimes\ket{a}\bra{a}\Big)\Big|\Big|_{\Diamond}\nonumber\\
    =&\cD_{\Diamond}(\hat{\Gamma}_{\mathbf{I}}, \hat{\Gamma}_{\mathbf{J}^{\prime}})\nonumber\\
    =&\cD_{\Diamond}(\mathbf{I}, \mathbf{J}^{\prime}).
\end{align}
 Thus we can conclude that, for a given instrument $\mathbf{I}=\{\Lambda_a\}_{a=1}^{\Omega_{\mathbf{I}}}$, with an outcome set ${\Omega_{\mathbf{I}}}$ and another instrument $\mathbf{J}=\{\Phi_b\}_{b=1}^{\Omega_{\mathbf{J}}}$ with outcome set $\Omega_{\mathbf{J}}$ with $\Omega_{\mathbf{J}}\geq {\Omega_{\mathbf{I}}}$, there always exist an instrument $\mathbf{J}^{\prime}=\{\Phi_b\}_{b=1}^{\Omega_{\mathbf{J}^{\prime}}}$ with $\Omega_{\mathbf{J}^{\prime}}=\Omega_{\mathbf{I}}$ such that
\begin{align}
    \cD_{\Diamond}(\mathbf{I}, \mathbf{J})\geq\cD_{\Diamond}(\mathbf{I}, \mathbf{J}^{\prime}).\label{Eq:Dist_montone_compact}
\end{align}

Let us consider two sets of instruments $\cI=\{\mathbf{I}_i\}$ and $\cJ=\{\mathbf{J}_i\}$. Now Eq. \ref{Eq:Dist_montone_compact} holds for every $i$. Then, there will exist another set of instruments $\cJ^{\prime}=\{\mathbf{J}_i^{\prime}\}$, with $\Omega_{\mathbf{J}_i^{\prime}}=\Omega_{\mathbf{I}_i},~\forall~i$, such that
\begin{align}
    \widehat{\cD}(\cI, \cJ)=&\max_i\cD_{\Diamond}(\mathbf{I}_i, \mathbf{J}_i)\nonumber\\
    \geq&\max_i\cD_{\Diamond}(\mathbf{I}_i, \mathbf{J}^{\prime}_i)\nonumber\\
    =&\widehat{\cD}(\cI, \cJ^{\prime})\label{Eq:Dist_montone_set_compact}
\end{align}
Now, consider the resource theory of entanglement preservability. A free set of instruments in this resource theory is a set of instruments whose CP maps are entanglement-breaking. In other words, they all have the form in Eq. \eqref{Eq:EB_Oper}. Consider such a free set $\cJ=\{\mathbf{J}_i=\{\Phi^i_a\}_a\}_i$. Here $\Phi^i_a$ is of the form in Eq. \eqref{Eq:EB_Oper} for every $a,i$. Note that, for every $i$, $\Omega_{\mathbf{J}_i}$ can be arbitrary but finite. We next denote by $\cJ^{\prime}=\{\mathbf{J}^{\prime}_i=\{\Phi^{\prime i}_a\}_{a=1}^{\Omega_{\mathbf{J}^{\prime}_i}}\}_i$, where $\Omega_{\mathbf{J}^{\prime}_i}=\Omega_{\mathbf{I}_i}~\forall ~i$, those set of instruments that are related to the instruments in the set $\cJ$ such that for every $i$ 

\begin{align}
    \Phi^{\prime i}_a=&\Phi_a, \qquad\qquad\qquad \text{for}~a=1,2,\ldots,{\Omega_{\mathbf{I}_i}}-1\\
    \Phi^{\prime}_{\Omega_{\mathbf{I}_i}}=&\Phi_{\Omega_{\mathbf{I}_i}}+\sum_{a=\Omega_{\mathbf{I}_i}+1}^{\Omega_{\mathbf{J}_i}}\Phi_a.
\end{align}

From this definition, we note that if $\Phi^i_a$ are of the form in Eq. \eqref{Eq:EB_Oper} for all $a,i$, then $\Phi^{\prime i}_a$ are also of the same form, and hence, entanglement-breaking, for all $a,i$. Thus we can conclude that in the resource theory of entanglement-preservability, if $\cJ$ is a free set of instruments then $\cJ^{\prime}$ is also a free set of instruments. Combining this fact with Eq. \eqref{Eq:Dist_montone_set_compact}, for a given set of instruments $\cI=\{\mathbf{I}_i\}$, we can restrict the optimization process in the definition of the resource measures $\mathbbm{R}$ in Eq. \eqref{Eq:Def_res_meas} to those sets of free sets of instruments in which free sets of instruments contain instruments with a fixed outcome set $\Omega_{\mathbf{I}_i}$, where $i$ indexes the instruments in that free set of instruments. As all such free sets of instruments are compact, the minimum can always be achieved. Similar logic holds for the resource measure $\overline{\mathbbm{R}}$ in Eq. \eqref{Eq:Def_res_ext_meas} also.

As one can easily provide logic to arrive at the same conclusion for all other instrument-based resource theories considered in this paper, we will not explicitly provide those details here.
\end{document}